\documentclass[11pt]{article}

\date{\vspace{-7ex}}

\usepackage[a4paper,top=2.5cm,bottom=2.5cm,left=3cm,right=3cm,marginparwidth=1.75cm]{geometry}
\usepackage{setspace}
\usepackage[dvipsnames]{xcolor} 
\usepackage{amsthm}
\newtheorem{theorem}{Theorem}

\usepackage{graphicx}
\usepackage{dcolumn}
\usepackage{bm}
\usepackage{braket}
\usepackage{amsmath, amssymb}
\usepackage{textcomp}
\usepackage{textgreek}
\usepackage{tikz}
\usepackage{authblk} 

\usepackage{physics} 

\newcommand{\ex}{\mathop{\mathbb{E}}}
\renewcommand{\var}{\operatorname{Var}}

\usepackage{hyperref}
\hypersetup{colorlinks=true, citecolor=blue, urlcolor=blue, linkcolor=blue}

\usepackage{amsthm}
\newtheorem{proposition}{Proposition}

\usepackage{mathtools}
\usepackage{dsfont}
\usepackage[backend=biber, style=nature, sorting=none, hyperref=true, url=false, doi=false, eprint=false, date=year]{biblatex}
\title{The trainability of photonic quantum circuits}

\author[1]{Alexander Makarovskiy\textsuperscript{*\,$\dagger$}}
\author[2]{Adam Taylor\textsuperscript{*\,$\ddagger$}}
\author[2]{Zhenghao Li}
\author[2]{Michael Hanks}
\author[2]{Aubrey Clark}
\author[2]{M.~S.~Kim}
\author[3]{Ian Walmsley}
\author[1]{William Clements\textsuperscript{\S}}

\affil[1]{{\small ORCA Computing, London W2 6LA, United Kingdom}}
\affil[2]{{\small Blackett Laboratory, Imperial College London, London SW7 2AZ, United Kingdom}}
\affil[3]{{\small Oxford Quantum Institute, Clarendon Laboratory, University of Oxford, Oxford OX3 1PU, United Kingdom}}

\begin{document}

\maketitle

\begin{center}
\footnotesize
\textsuperscript{*}These authors contributed equally.\\
\textsuperscript{$\dagger$}\href{mailto:alexander.makarovskiey@gmail.com}{alexander.makarovskiy@gmail.com},
\textsuperscript{$\ddagger$}\href{mailto:adam.taylor18@imperial.ac.uk}{adam.taylor18@imperial.ac.uk}, 
\textsuperscript{\S}\href{mailto:wclements@orcacomputing.com}{wclements@orcacomputing.com.}
\end{center}

\begin{abstract}
Variational quantum algorithms are a leading approach to near-term quantum computing, but their scalability can be limited by barren plateaus and the sampling cost of resolving small changes in the loss landscape. Here, we study the trainability of passive linear-optical quantum circuits and introduce a framework based on the ratio of sample variance to circuit variance. This ratio determines the number of circuit samples required to resolve local loss differences and gradients to proportional accuracy. We apply this framework to photon-number observables and identify both trainable and non-trainable regimes. Supported by analytic results and a numerically observed polynomial decay of the circuit variance, we find that fixed-order photon-number polynomials require only polynomially many samples as the system size grows, whereas high-order polynomials and observables based on output probabilities generally require exponentially many samples. Within the trainable regime, we further identify classes of observables in which quantum estimation achieves a polynomial speed-up over multiple classical methods. Within this family, neural network observables provide one practical construction that allow measurement outcomes to be efficiently processed into the desired polynomial. These results establish photonic variational quantum computing as a promising platform for near-term applications.
\end{abstract}

\section{Introduction}

    Current quantum hardware remains in the Noisy Intermediate-Scale Quantum (NISQ) regime \cite{preskill2018quantum}, with limited qubit counts, connectivity, and noise levels that place celebrated quantum algorithms such as Shor's algorithm for factoring large numbers beyond reach \cite{shor1999polynomial, gidney2021factor}. Variational quantum computing is a leading approach for utilising such devices towards applications, encompassing variational quantum algorithms (VQAs) and quantum machine learning (QML) protocols \cite{cerezo2021variational, mcclean2016theory}. Variational circuits have found applications in optimisation \cite{farhi2014quantum, zhou2020quantum}, quantum chemistry \cite{mcardle2020quantum, peruzzo2014variational}, machine learning \cite{biamonte2017quantum, abbas2021power}, linear algebra \cite{xu2021variational}, state preparation \cite{consiglio2024variational,castro2024variational} and quantum simulation \cite{jones2019variational,endo2020variational}.

    The optimism around VQAs and QML has been tempered by a series of no-go results establishing fundamental obstacles to training parameterised quantum circuits. Foremost among these is the barren plateau (BP) phenomenon: a guarantee that, for a sufficiently expressive and randomly initialised circuit, the optimisation landscape is almost entirely flat away \cite{mcclean2018barren}. There can still exist a small number of exponentially concentrated minima, referred to as narrow gorges \cite{mcclean2018barren, arrasmith2022equivalence}, but the probability of finding one is exponentially small. Training over such a landscape requires exponential precision to resolve gradients, followed by a difficult search for the rare regions where pronounced minima exist \cite{arrasmith2022equivalence,larocca2025barren}. These problems can be made worse by the choice of observable \cite{cerezo2021cost}, excess entanglement \cite{marrero2021entanglement,patti2021entanglement} or noise in real quantum hardware \cite{wang2021noise}. A substantial body of work has gone into diagnosing barren plateaus in gate-based regimes \cite{ragone2024lie,larocca2022diagnosing}, and considerable effort has since gone into identifying circuit families and methods that avoid them \cite{grant2019initialization, pesah2021absence, skolik2021layerwise, patti2021entanglement, sack2022avoiding}. However, many of the architectures shown to be barren-plateau-free have also turned out to be efficiently classically simulable \cite{cerezo2025provable_absence,goh2025lie,bermejo2026quantum}. These results have made it challenging to scale experimental demonstrations of variational algorithms past a few tens of qubits.

    In contrast, recent experiments with photonic quantum circuits have successfully demonstrated quantum machine learning tasks at a scale significantly beyond the limitations implied by the theory around barren plateaus. In \cite{gong2025enhanced}, the authors successfully encode pixel data from the MNIST handwritten-digit dataset into the Hilbert space spanned by a quantum photonic system with 8176 modes and over 2000 photons and perform image classification. In \cite{cimini2026large}, the authors perform quantum reservoir computing with a 400-mode photonic circuit and several hundreds of photons. These results warrant further investigation into how photonic circuits differ in trainability. 
    
    However, the previous literature on investigating quantum circuit trainability is not directly applicable to photonics. Near-term photonic architectures often rely on parametrisable linear-optical circuits with quantum state inputs. While not universal for quantum computation, sampling from such a circuit is classically hard despite its relative experimental simplicity~\cite{aaronson2011computational, deshpande2022quantum, hangleiter2023computational, li2025complexity}. This architecture thus provides a promising platform for demonstrating quantum advantage. The natural configurability, absence of decoherence, and high operation speeds of photonic circuits also make them naturally suited for implementing variational protocols. However, the highly constrained gate set, the fact that photonic circuits process particles occupying modes rather than a register of qubit subsystems, and observables which do not naturally map onto qubit-based equivalents mean that practical analyses of the trainability of photonic VQAs require explicit consideration. 

    Here, we shed light on the trainability of photonic quantum circuits. We define the notions of circuit and sample variance which govern the magnitude of circuit gradients, and sample complexity of estimating them. While natural bounds on observables in gate-based regimes make the identification of circuit variance through barren plateaus often sufficient for diagnosing trainability, we will show that many photonic observables require careful consideration of both variances.  We motivate the ratio between them as a metric for trainability, and use this to identify examples of natural photonic observables that are untrainable. We then focus on a different class of observables consisting of linear sums of products of photon numbers. We show that this class of observables can exhibit polynomially scaling training complexity with increasing system size, but can be efficiently simulated classically. Nonetheless, within this trainable regime we find families of observables for which we can obtain a polynomial speed-up over a different classical methods. We also introduce a neural network observable as an efficient way of implementing a these photon-number polynomials (PNPs). We illustrate these results in both simulation on relevant regimes and on real hardware. 
    
    Our results establish that photonic quantum circuits with certain classes of observables can scale to large sizes while maintaining trainability, and also polynomial speed-up over a set of current classical simulation methods. This is promising for finding practical quantum advantage with recently developed applications, ranging from machine learning~\cite{gong2025enhanced, li2026machine, hoch2025quantum}, reservoir computing~\cite{cimini2026large}, to generative modelling~\cite{kolarovszki2026generative, bacarreza2025quantum}.

\section{Background}
\subsection{Variational photonic circuits}

We consider the general variational quantum computing setup of an initial quantum state $\rho$, a parametrised quantum circuit $\hat{U}(\bm{\theta})$ with parameters $\bm{\theta}$, and measurement of an observable $\hat{O}$. The objective of VQAs or QML models is to find parameters $\bm{\theta}$ minimising the value of the associated loss function, defined as the expectation value
\begin{align}
\begin{split}
\label{eq:loss_function_definition}
    \ell(\bm{\theta}) = \tr[\hat{O} \, \hat{U}(\bm{\theta}) \rho \hat{U}^\dag(\bm{\theta})].
\end{split}
\end{align}
More general loss functions can be constructed by considering multiple input states or observables. We will focus on the simple case of Eq.~\eqref{eq:loss_function_definition}, as results in this regime can be extrapolated to the more general cases. Additionally, while parametrised photonic circuits may incorporate mid-circuit measurements and feedforward operations~\cite{Chabaud2021quantummachine, hoch2025quantum, monbroussou2025towards, monbroussou2025photonic}, we focus exclusively on the dynamics of passive linear optics.

Parameterised quantum circuits are generally trained through a hybrid quantum--classical optimisation loop: a quantum device is used to estimate the loss function at various parameter settings, often to build gradient estimates. These estimates are then supplied to a classical optimiser which updates the parameters towards reducing the loss function. Repeating this procedure iteratively trains the parameters of the circuit towards a loss-minimising regime. Although this framework is largely independent of the underlying hardware platform, the majority of existing work has focussed on qubit-based architectures. In Fig.~\ref{fig:summary}, we illustrate a photonic quantum circuit as part of a variational training loop.

\begin{figure}[ht]
\centering
    \includegraphics[width=1\textwidth]{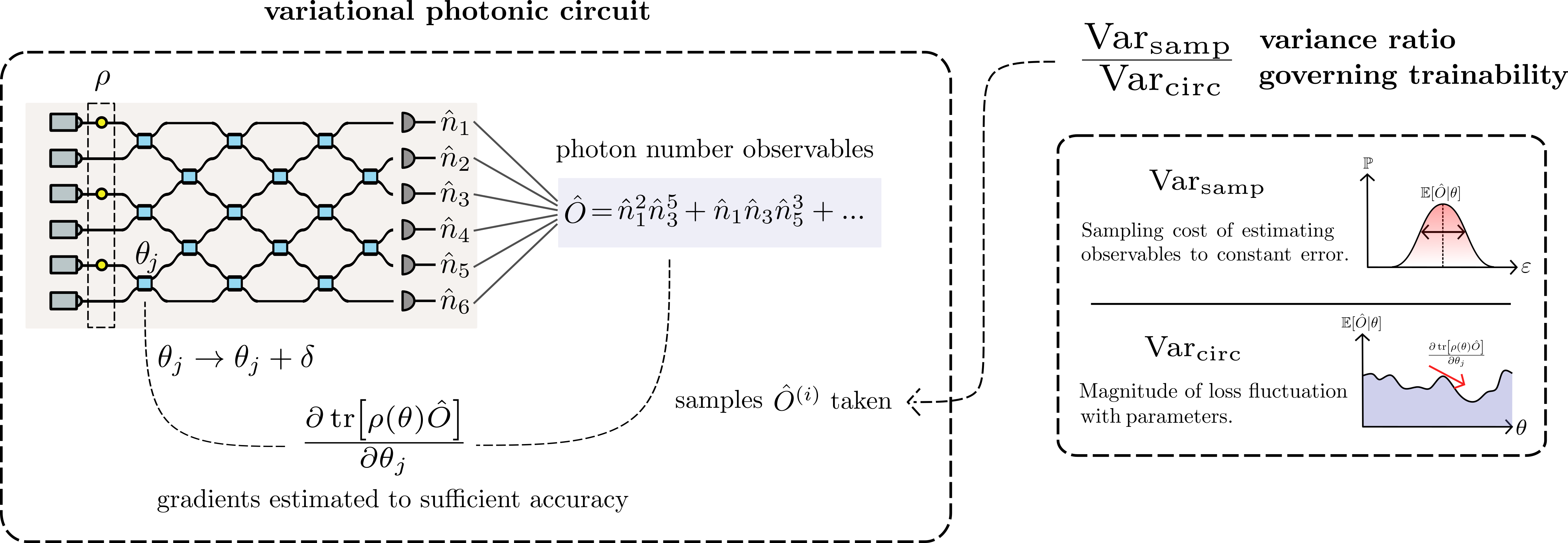}
    \caption{We illustrate the class of VQAs considered in this paper. The ratio between two notions of variance, sample and circuit variance, give a measure of the number of samples needed to estimate accurate gradients to train a photonic circuit with respect to a specific observable. Circuit variance governs the fluctuations of the expected value of an observable with respect to circuit parameters, and exponentially small circuit variance is known as a \textit{barren plateau}. Sample variance determines the number of samples needed to estimate the value of an observable within a given error. We study the trainability of photon-number observables, which exhibit a polynomially scaling variance ratio, making them a suitable for implementing efficiently trainable variational quantum circuits}
    \label{fig:summary}
\end{figure}

We now turn to the photonic setting that is the focus of this work. We study variational quantum protocols implemented on photonic platforms, with parametrised circuits given by passive linear optical networks (PLONs). We restrict our attention to input states with a fixed total photon number.

An $M$-mode photonic system is described by the creation and annihilation operators $\hat{a}^\dag_i$ and $\hat{a}_i$ respectively for $i\in\{1,\dots,M\}$, satisfying the canonical commutation relations
\begin{align}
\begin{split}
    [\hat{a}_i,\hat{a}_j] = [\hat{a}_i^\dag, \hat{a}_j^\dag] = 0, \quad \quad [\hat{a}_i, \hat{a}_j^\dag] = \delta_{ij}.
\end{split}
\end{align}
A Fock state has a definite photon number in each mode and is defined by
\begin{align}
\begin{split}
    \ket{\bm{n}} = \ket{n_1, \dots, n_M} 
    = \prod_{i=1}^M \frac{(\hat{a}_i^\dag)^{n_i}}{\sqrt{n_i!}} \ket{0}
\end{split}
\end{align}
These are eigenstates of the photon-number operator, $\hat{n}_i = \hat{a}_i^\dag \hat{a}_i$, satisfying $\hat{n}_i \ket{\bm{n}} = n_i \ket{\bm{n}}$. The $n$-photon subspace, $\mathcal{H}_{n,M}$ is spanned by all Fock states with total photon number $n$,
\begin{align}
\begin{split}
    \mathcal{H}_{n,M} = \textrm{span}\big\{\ket{\bm n} \, | \, \bm{n} \in \Phi_M^n \big\}, \quad \Phi_M^n = \{\bm{m} \in \mathbb{N}_0^M \, | \, \sum_{i=1}^M m_i = n\},
\end{split}
\end{align}
and has dimension $|\Phi_M^n| = {n+M-1\choose n}$.

Passive linear optical networks implement transformations that are linear in the creation/annihilation operators and preserve total photon number. Each such transformation is described by a unitary matrix acting linearly on the $M$ modes, i.e. an element of $\mathrm{U}(M)$. For a particular $U\in\mathrm{U}(M)$, there is an associated optical unitary operator $\hat{U}$ over the entire Hilbert space whose action is
\begin{align}
\begin{split}
    \hat{U} \hat{a}_i^\dag \hat{U}^\dag = \sum_{j=1}^M U_{ij}\hat{a}_j^\dag .
\end{split}
\end{align}
Since these transformations preserve total photon number, $\hat{U}$ maps each fixed photon number subspace to itself.

Experimentally, PLONs are constructed from beam splitters and phase shifters. A beam splitter provides a programmable coupling between two modes, controlling how much amplitude is exchanged between them. A phase shifter implements a per-mode tunable phase. Arbitrary linear unitary transformations over $M$ modes can be implemented with a circuit of $M(M-1)/2$ beam splitters with a final layer of phase shifters \cite{reck1994experimental, clements2016optimal, bell2021further}.

However, universality over $\mathrm{U}(M)$ only provides a universal $n$-photon \textit{linear-optical} unitary group, which is a proper subgroup of the full unitary group on the $n$-photon Hilbert space $\mathcal{H}_{n,M}$~\cite{bouland2014generation}. The induced dynamics are restricted to only a subspace of $\mathcal{H}_{n,M}$ decided by the initial state. As such, parametrised PLONs with Haar-random unitary circuit drawn from $\textrm{U}(M)$ are less expressive than circuit-based model of quantum computing, where a random unitary can instead be drawn from the Haar measure over the whole Hilbert space. Despite these restricted dynamics, sampling from the output distribution of a linear-optical circuit, a task known as boson sampling, remains hard to simulate classically \cite{aaronson2011computational}.

A natural observable to consider for photonic circuits is a product of photon numbers across modes \cite{cardin2024photon}, i.e. observables of the form:
\begin{equation}
\hat{M}_{\bm{\alpha}} = \bigotimes_{i = 1}^K \hat{n}_i^{\alpha_i},
\end{equation}
We call the subset of $K\le M$ modes the observable acts on its \textit{support}, and $|\bm\alpha| = \sum_i\alpha_i$ it's \textit{order}.
For some state $\rho$, the expected value of the observable, $\tr[\hat{M}_{\bm{\alpha}} \rho]$, is a moment of the photon-number distribution across the modes. We will refer to observables of the form $\hat{M}_{\bm{\alpha}}$ as photon-number monomials (PNMs).

Any observable diagonal in the photon number basis can be written as a linear sum of PNMs. As such, a study of the trainability of PNM observables generally tells us about the trainability of photonic circuits with respect to arbitrary objective functions.

\subsection{Variances as a measure of trainability}
\label{subsec:barren_plateaus_and_trainability}

\begin{figure}[ht]
    \hspace{2.5cm}
    \includegraphics[width=0.8\textwidth]{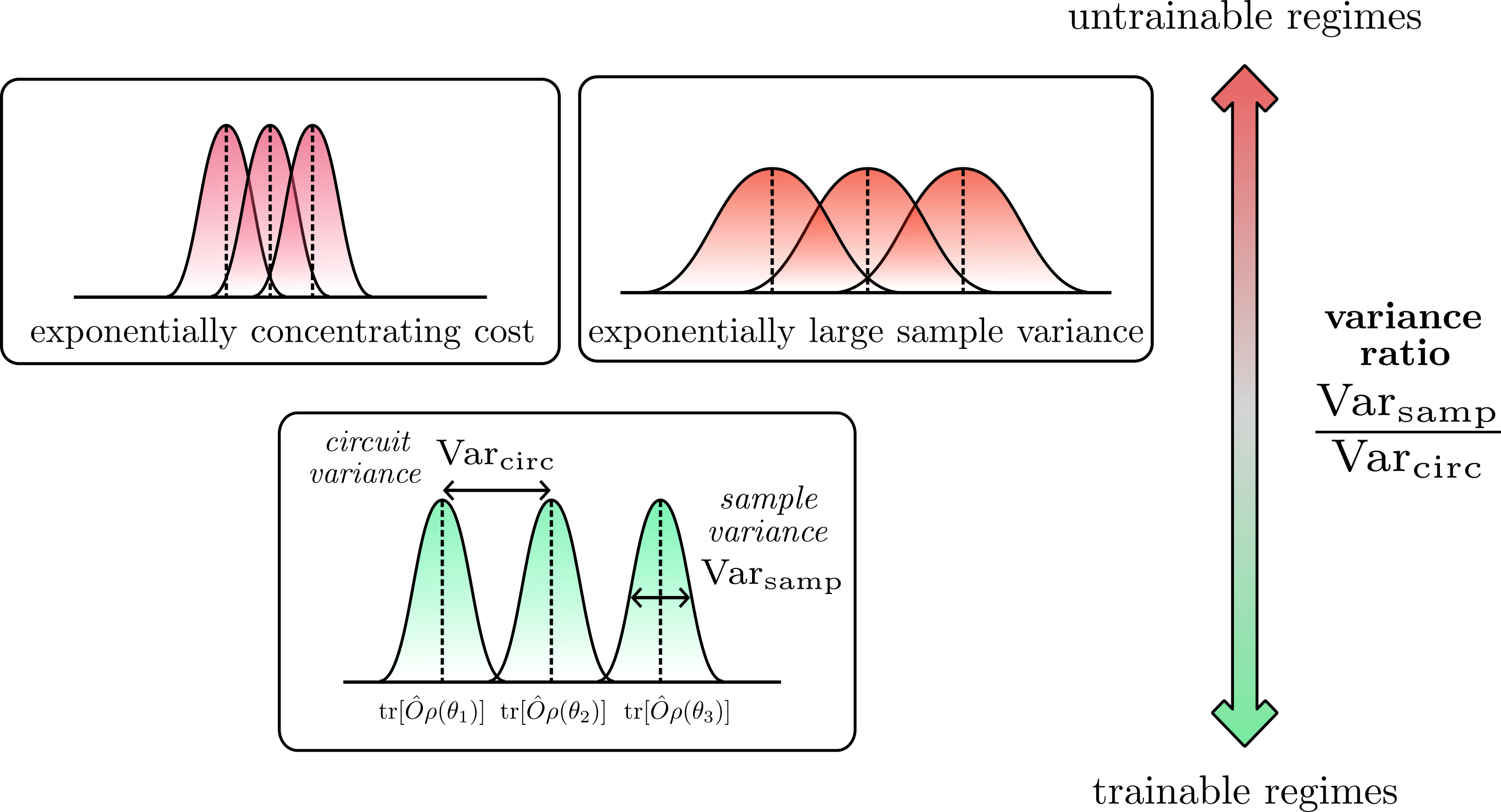}
    \caption{We visualise the distribution of a finite-sample estimate of a fixed observable at different parameter settings. For regimes with a high variance ratio, two observable estimates from different parameter settings will likely not be sufficiently distinct to resolve for which parameter set the ground-truth loss was smaller, due to the overlap in the distributions of the estimates. This can arise either through the means of the observable at different parameter settings being exponentially concentrated (top left), or through the variance of sampling the observable being exponentially large (top right).}
    \label{fig:variance_ratio_regimes}
\end{figure}

    Training the parameters of a quantum circuit requires estimating the gradients of the loss function with respect to a given circuit parameter. Generally, the calculation of these gradients analytically is intractable. Practically, we need to sample the quantum circuit at shifted parameter settings, and use estimates of the loss to construct an estimate of the gradient. Throughout this work, we say a circuit--observable pair is \textit{trainable} if its gradients can be resolved to within proportional error with the number of samples required to estimate them growing at most polynomially with the size of the circuit. Conversely, if the number of samples needed to resolve gradients grows exponentially in system size, we say that the choice of circuit is not trainable.  In what follows we phrase the trainability requirements in terms of resolving the loss itself; in Appendix~\ref{appendix:bp_equivalence} we show that, under mild assumptions, this is sufficient to draw the same conclusions for gradients \cite{arrasmith2022equivalence}.

    This notion of trainability is local: given a parameter setting $\bm\theta$, it asks whether we can efficiently identify an update $
    \bm{\delta\theta}$ such that, with high probability, the new parameters satisfy $\ell(\bm{\theta} + \bm e_j \delta\theta_j) \le \ell(\bm\theta)$ for all $j$. In other words, we ask only whether the estimated loss function and gradients contain enough information to make a loss-decreasing parameter update along each parameter direction. We do not address the global optimisation landscape, including the presence of poor local minima, the number of optimisation steps required for convergence, or the training rate needed to avoid overshooting minima.

    The first quantity determining the sampling cost is the variance of the measurement outcomes themselves. For a given parameter setting $\bm{\theta}$, a single run of the circuit yields a measurement outcome of the observable $\hat{O}$ applied to the output state $\rho(\bm{\theta}) = \hat{U}(\bm\theta) \rho \hat{U}^\dag(\bm\theta)$. This distribution has variance $\var[\hat{O} | \bm\theta] = \tr[\hat{O}^2 \rho(\bm{\theta})] - \tr[\hat{O} \rho(\bm{\theta})]^2$. 
    The expectation value is estimated by averaging the outcomes of repeated measurements, with the estimation error decreasing as the number of samples increases. The rate of this decrease is governed by the measurement variance $\var[\hat{O}|\bm\theta]$ such that a broader distribution of outcomes requires more samples to achieve a given precision.
    Since we are interested in the behaviour of circuits under random initialisation, we average over the distribution of parameters, and study the \textit{sample variance} $\text{Var}_{\textrm{samp}}(\hat{O}) := \ex_{\bm{\theta}}\left[\var[\hat{O}|\bm\theta]\right]$. This is the expected variance of measurement outcomes when estimating the loss at a typical parameter setting.

    The sample variance alone does not tell us how many samples are enough for resolving loss accurately. We need to know the error to which it is sufficient to estimate. The \textit{circuit variance} $\text{Var}_{\text{circ}}(\hat{O}) := \var_{\bm{\theta}}\left[\ell(\bm{\theta}; \hat{O}, \rho)\right]$ measures how much the ground-truth loss fluctuates as the circuit parameters are varied, and hence the typical magnitude of the loss differences we must distinguish. In order to resolve a loss-minimising direction, the error of our estimates must be proportional to the magnitude of these differences: the fluctuations of the error must be smaller than the signal of the loss we wish to estimate, otherwise the estimates carry no information about which parameter settings decrease the loss. The scaling of the circuit variance therefore sets the precision to which the loss must be estimated.

    Together, the sample variance sets the cost of achieving a given estimation error, and the circuit variance sets the error we need to achieve. The number of samples needed to estimate observables to resolve a loss minimising direction is then governed by the ratio of these two quantities, $\text{Var}_{\text{samp}} / \text{Var}_{\text{circ}}$. In Appendix \ref{app:variance_ratio_sufficiency} we make this precise, showing that a number of samples scaling with this ratio is sufficient to resolve loss differences at different parameter settings, and hence gradients, to within proportional error. We therefore take this \textit{variance ratio} as our measure of trainability. Circuit--observable pairs whose variance ratio grows at most polynomially with system size are efficiently trainable in the sense that we can efficiently update the parameter in the correct direction. In contrast, an exponentially growing ratio demands exponentially many samples. Fig.~\ref{fig:variance_ratio_regimes} illustrates the two ways in which the ratio can become large, namely via the loss values at different parameter settings concentrating exponentially around their mean, or the spread of the estimates around each loss value growing exponentially. In either case, estimates of the loss at different parameter settings overlap too strongly to resolve which is smaller.

    For gate-based circuits, natural bounds on the observable magnitude, such as a unit eigenspectrum, bound the sample variance to scale at most polynomially \cite{cerezo2021cost}. In this setting the numerator of the variance ratio is benign, and it suffices to study the circuit variance alone. An exponentially decaying circuit variance, known as a \textit{barren plateau} \cite{mcclean2018barren, arrasmith2022equivalence}, then implies that exponentially many samples are needed to resolve exponentially small gradients, and its absence is sufficient to expect efficient training. These assumptions do not generally hold for photonic circuits. Natural photonic observables, such as products of multiple photon numbers, are unbounded and can take exponentially large values in the number of photons, so the sample variance itself can grow exponentially even in the absence of a barren plateau. In Sec.~\ref{subsec:trainable_regimes} we make this concrete, studying photonic observables that fail to be trainable through either an exponentially large sample variance, or exponentially concentrating circuit variance. Both variances must therefore be considered to conclude trainability in photonic circuits.

\section{Trainability of photonic observables}
\label{sec:trainability}
Motivated by the trainability conditions established in Sec.~\ref{subsec:barren_plateaus_and_trainability}, we derive expressions for the circuit variance and sample variance of photon-number observables in the Haar-average regime. To do so, we compute the Haar-averaged first and second moments of photon-number observables, from which both variances follow directly. In Sec.~\ref{subsec:variance_calculations}, we present the analytic results for the first moment and second moment. In Sec.~\ref{subsec:trainable_regimes}, we apply these equations to calculate the trainability ratio of different photonic observables. We start with photon-number monomials and polynomials, showing that the complexity of training scales with the order of the polynomial, and identifying fixed-order polynomials as trainable. We then discuss neural network observables as a natural implementation of photon-number polynomials. In Sec.~\ref{subsec:trainable_regimes} we present negative results about the non-trainability of probability output estimating observables. We present the tools developed to produce these results and derivations in the Appendix.

\subsection{Variance calculations for photon-number observables}
\label{subsec:variance_calculations}
In this section, we present the analytic expressions for the first and second moments used in the variance ratio, for both PNMs and PNPs.
Writing $\mu_k(\hat{O}) = \ex_U\big[\tr[\hat{O} \, U \rho U^\dag]^k\big]$ for the $k^{\mathrm{th}}$ Haar-average moment of the loss, the circuit and sample variances can be written as differences of its first and second moments,
\begin{align}
    \var_{\text{samp}}(\hat{O}) &= \mu_1(\hat{O}^2) - \mu_2(\hat{O}), \\
    \var_{\text{circ}}(\hat{O}) &= \mu_2(\hat{O}) - \mu_1(\hat{O})^2. 
\end{align}
The trainability of an observable--initial state pair therefore reduces to computing the first two moments over the Haar measure.

We consider $\mu_1$ first. The Haar-averaged output state of a linear optical network acting on an $n$-photon input state is the maximally mixed state over the $n$-photon subspace -- see Appendix \ref{appendix:average_linear_optical_network} for details. As a result, the Haar-averaged first moment of an observable reduces to its value over an equally weighted sum over all possible output patterns.
In Appendix \ref{appendix:first_moment}, we prove the following about photon number monomials;
\begin{theorem}[First moment of photon-number monomials]
\label{thm:first_moment_main}
    Let $\hat{M}_{\bm{\alpha}}$ be a PNM with $\bm{\alpha} = (\alpha_1, \dots, \alpha_K)$, and let $\rho$ be an input state with total photon number $n$. The expectation value of the observable in the Haar-average regime is given by,
    \begin{align}
    \label{eq:pnm_first_moment_main}
        \mu_1(\bm{\alpha}; n) = \sum_{P=K}^{|\bm\alpha|} \frac{{n+M-1\choose n-P}}{{n+M-1\choose n}}  \sum_{\bm{p}\in \Phi_K^P} \left(\prod_{i=1}^K p_i! \, S(\alpha_i, p_i)\right),
    \end{align}
    where $S(\alpha_i, p_i)$ are Stirling numbers of the second kind. Furthermore, if the photon density is constant with $n=\nu M$, then in the large-$M$ limit this approaches a product of ordered Bell polynomials, $F_a(x) = \sum_{p=0}^a p! S(a, p) x^p$, up to $\mathrm{O}(1/M)$ corrections,
    \begin{align}
    \begin{split}
    \label{eq:first_moment_asymptotic_main}
        \mu_1(\bm\alpha; n=\nu M) \,\xrightarrow[M \to \infty]{} \, \left(\prod_{j=1}^K F_{\alpha_j}(\nu) \right) + \textrm{O}(1/M).
    \end{split}
    \end{align}
    Crucially, if $K$ and $|\bm\alpha|$ do not scale with $M$, then the first moment asymptotically approaches a constant value.
\end{theorem}
Any PNP first moment can be written as a sum over PNM first moments via $\mu_1\big(\sum_{\bm\alpha} c_{\bm\alpha} \hat{M}_{\bm\alpha}\big) = \sum_{\bm\alpha} c_{\bm\alpha} \mu_1(\bm\alpha)$ and hence Theorem~\ref{thm:first_moment_main} is applicable to polynomials too.
We can now analyse the scaling of the first moment in terms of the order of the monomial and the circuit size. For a constant-density regime $n = \nu M$, where the first moment is given by a product of ordered Bell polynomials in Eq.~\eqref{eq:first_moment_asymptotic_main}. 

In Appendix \ref{appendix:first_moment_asymptotics}, we show that for large $|\bm\alpha|$ the first moment asymptotically scales as
\begin{align}
\label{eq:first_moment_order_scaling_main}
    \mu_1(\bm{\alpha}; n = \nu M) \to \prod_{i=1}^K F_{\alpha_i}(\nu) \sim \prod_{i=1}^K \frac{\alpha_i!}{(1 + \nu) \ln(1 + 1/\nu)^{\alpha_i + 1}},
\end{align}
This grows factorially in the order and exponentially in the number of modes the PNM is supported on. In Fig.~\ref{fig:pnm_scaling_with_weight}, we plot $\mu_1(\bm\alpha)$ against $|\bm\alpha|$ for fixed $K$, and find that it does indeed grow superexponentially.\\
\\
Calculation of the second moment requires the second-order Haar moments of the output state. We apply the techniques developed by Mhiri \textit{et al} \cite{mhiri2026bosonsamplingdiluteregime}, which decompose the second moment over the irreducible representations of $\mathrm{U}(M)$ within the $n$-photon subspace to give \cite{mhiri2026bosonsamplingdiluteregime}
\begin{align}
\begin{split}
    \mu_2(\hat{O};\rho) = \sum_{k=0}^n \frac{1}{d_k^{(n)}} ||\rho_k^{(n)}||_2^2 \, ||\hat{O}_k^{(n)}||_2    
\end{split}
\end{align}
Here, $\hat{X}_k^{(n)}$ is a projector onto the $k^{\textrm{th}}$ irreducible representation of $\mathrm{U}(M)$ and $||\bullet||_2^2 = \tr[\bullet^\dag \bullet]$ is the Hilbert-Schmidt norm. In Appendix \ref{appendix:second_moment}, we present or derive the projected norms for Fock states, entangled states, PNMs and PNPs. These values can be efficiently computed using discrete convolutions, allowing us to numerically study the scaling of the second moments.

\begin{figure}[ht]
    \centering
    \includegraphics[width=0.9\linewidth]{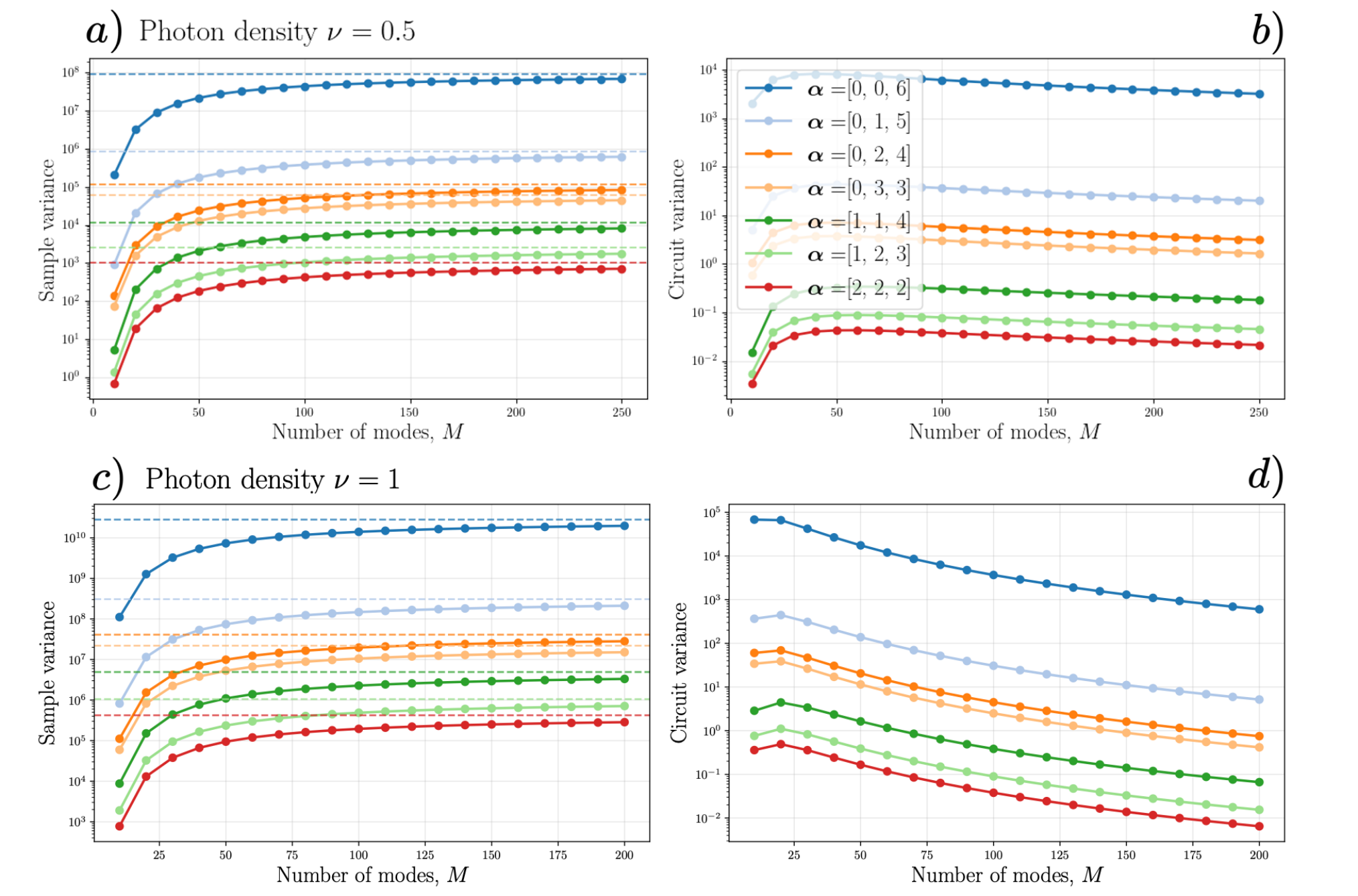}
    \caption{Sample and circuit variances as functions of the number of modes $M$ for PNMs satisfying $|\bm{\alpha}|=6$ and $K\leq 3$. Panels (a) and (c) show the sample variance for photon densities $\nu=0.5$ and $\nu=1$ respectively. The dotted lines denote the asymptotic limits $\left.\mu_1(2\bm{\alpha})\right|_{M\rightarrow\infty}$, towards which the sample variances converge. Panels (b) and (d) show the corresponding circuit variances. In both cases, the circuit variance decays only very slowly with $M$. Its visibly non-linear behaviour on the logarithmic vertical scale suggests that the decay is not exponential. We therefore conjecture that the circuit variance instead decays polynomially with $M$, as stated in Eq.~\eqref{eq:circuit_variance_pnm_decay_conjecture}.}
    \label{fig:fixed_pnm_scaling_analysis}
\end{figure}  
Unlike the first moment results, analytic asymptotic formulae for the second moments are less clear. Through the law of total expectation and Jensen's inequality, we can guarantee that $\mu_1(\bm\alpha)^2 \le \mu_2(\bm\alpha) \le \mu_1(2\bm\alpha)$; by Eq.~\eqref{eq:first_moment_asymptotic_main} we know that for fixed $n=\nu M$ and $K, |\bm\alpha|\ll M$ then the second moment is necessarily between two system-size independent constants, and thus itself approaches a constant. We numerically observe that $\mu_2(\bm\alpha)$ approaches the lower end, with $\mu_2(\bm\alpha) - \mu_1(\bm\alpha)^2$ appearing to decay slowly with system size. We conjecture that this decay is polynomial with respect to $M$. In Fig.~\ref{fig:n_alpha_plot_fixed_pnm_decay}, we plot $\var_{\textrm{circ}}(\hat{n}^\alpha)$ against $M$ for different choices of $\alpha$ with initial state $\ket{\bm{1},\bm{0}}$ and constant density $\nu$. We find all PNM first moments peak at some $M_{\textrm{max}}$ before slowly decaying, and the value of $M_{\textrm{max}}$ grows with $\alpha$. From the numerics alone we cannot perfectly distinguish the exact nature of the slow decay component, but we conjecture it is of the form
\begin{align}
\begin{split}
\label{eq:circuit_variance_pnm_decay_conjecture}
    \var_{\textrm{circ}}(\bm\alpha) \sim \frac{\mu_1(\bm\alpha)^2}{\textrm{poly}(M)}.
\end{split}
\end{align}
This is well justified by Fig.~\ref{fig:fixed_pnm_scaling_analysis} where we see non-linear decay in a log-plot of the circuit variance across a series of low order PNMs. In Fig.~\ref{fig:n_alpha_plot_fixed_pnm_decay}, we again see a very slow decay of the circuit variance across a variety of photon-number monomials $\hat{n}^\alpha$.
The analytic formulae above, together with these numerical observations, characterise the quantities of sample and circuit variance governing trainability for photon number monomials and give insight into more general photon number polynomials.

\subsection{Trainable and non-trainable regimes}
\label{subsec:trainable_regimes}
We apply the results calculated above to analyse the trainability of different photon number polynomial regimes. We also analyse the trainability ratio of output probability estimation observables that have previously been used in the literature. Based on previous tradeoffs that have been observed between trainability and simulability, we comment on the classical simulability of these observables.
 
\subsubsection{Photon-number monomials}
\begin{figure}[ht]
    \centering
    \includegraphics[width=0.95\linewidth]{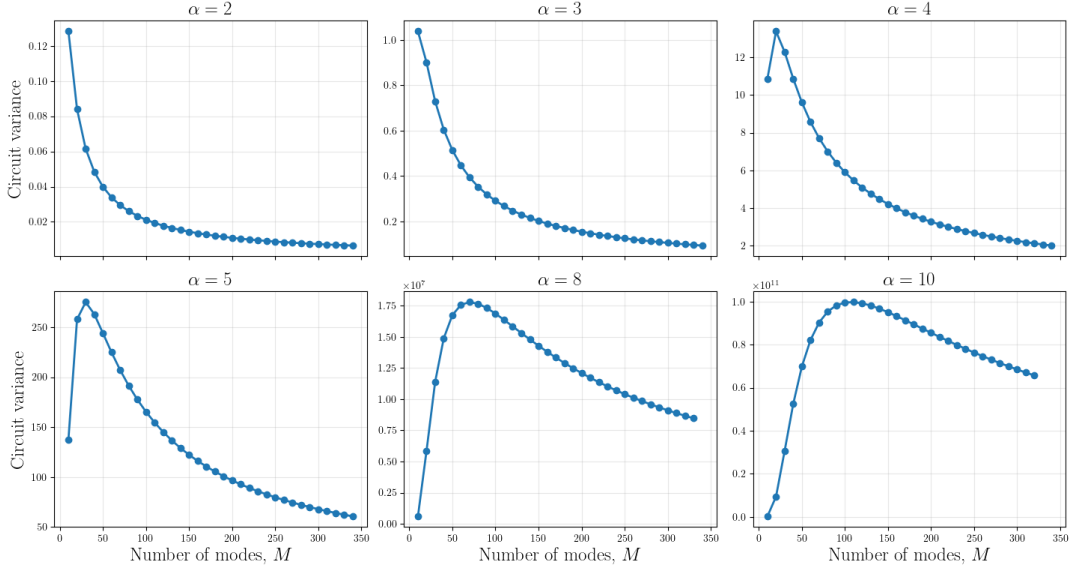}
    \caption{Circuit variance as a function of the number of modes $M$ for observables $\hat{n}^{\alpha}$ at several values of $\alpha$. The input state is $\ket{\bm{1},\bm{0}}$ with photon density $\nu=0.5$. For each $\alpha$, the circuit variance initially increases before entering a regime of very slow decay, with the maximum shifting to larger $M$ as $\alpha$ increases. Although deriving the asymptotic decay analytically is challenging, the numerical results suggest a polynomial dependence on $M$, consistent with the scaling analysis in Fig.~\ref{fig:fixed_pnm_scaling_analysis}.}
    \label{fig:n_alpha_plot_fixed_pnm_decay}
\end{figure}
\label{subsubsec:pnps}
The scaling results of Sec.~\ref{subsec:variance_calculations} allow us to directly calculate the variance ratio of photon-number monomials. We will again consider the regime where $n = \nu M$. Using the observations that the circuit variance decays polynomially, $\var_{\text{circ}}(\bm\alpha) \approx \mu_1(\bm{\alpha})^2 / \mathrm{poly}(M)$, and that the sample variance is dominated by the highest moment, $\var_{\text{samp}}(\bm\alpha) \approx \mu_1(2\bm{\alpha})$, the order scaling of Eq.~\eqref{eq:first_moment_order_scaling_main} gives
\begin{align}
\begin{split}
\label{eq:pnm_variance_ratio_main}
    \frac{\var_{\text{samp}}(\bm\alpha)}{\var_{\text{circ}}(\bm\alpha)}&
    \approx \mathrm{poly}(M) \cdot \frac{\mu_1(2\bm{\alpha})}{\mu_1(\bm{\alpha})^2}
    \\
    &\xrightarrow[M\to\infty]{}\; \mathrm{poly}(M) \cdot \left(\prod_{i=1}^K \binom{2\alpha_i}{\alpha_i} (1+\nu) \ln(1 + 1/\nu)\right).
\end{split}
\end{align}
The variance ratio scales polynomially in system size, with a factor that grows exponentially with the total order and support of the photon number monomial observable.

For a fixed-order monomial at constant photon density, the exponential factor in Eq.~\eqref{eq:pnm_variance_ratio_main} is a system-size-independent constant, and the variance ratio is strictly polynomial in $M$. This means fixed-order PNMs exhibit no barren plateaus as the system size scales, and their losses can be resolved to the scale of their circuit fluctuations with polynomially many shots. Therefore, we conclude that variational photonic circuits can be efficiently trained on fixed photon-number monomials.
 
Monomials whose total order $|\bm{\alpha}|$ grows linearly with the number of modes $M$ behave differently. The exponential factor in Eq.~\eqref{eq:pnm_variance_ratio_main} then grows exponentially with \emph{system size}, so we do not expect such observables to be trainable. 
For example, for the monomial $\bm\alpha = (1, \dots, 1)$, which takes a product of photon numbers over $K\le M$, consider how the variances scale. Just looking at the circuit variance, we find a scaling
\begin{align}
\begin{split}
    \var_{\textrm{circ}}(\bm\alpha) \approx \nu^{2K} / \textrm{poly}(M).
\end{split}
\end{align}
One might think this observable exhibits the opposite of a barren plateau -- in the oversaturated regime with $\nu > 1$, the circuit variance grows exponentially fast with $M$ when $K \propto M$. However, the sample variance grows exponentially faster
\begin{align}
\begin{split}
    \var_{\textrm{samp}}(\bm\alpha) \ge (2^K - 1)\nu^{2K}
\end{split}
\end{align}
and hence the variance ratio scales as $\Omega(2^K)$, which is also exponential in $M$ when $K \propto M$.
Resolving accurate loss function gradients governed by such observables therefore requires exponentially many samples, and it is not sufficient to only study the circuit variance to identify this. Intuitively, the observable exhibits large fluctuations with parameters, but also has rare large magnitude events contributing to the majority of this, meaning estimating these values through sampling is exponentially hard.

\begin{figure}[ht]
    \centering
    \includegraphics[width=0.8\textwidth]{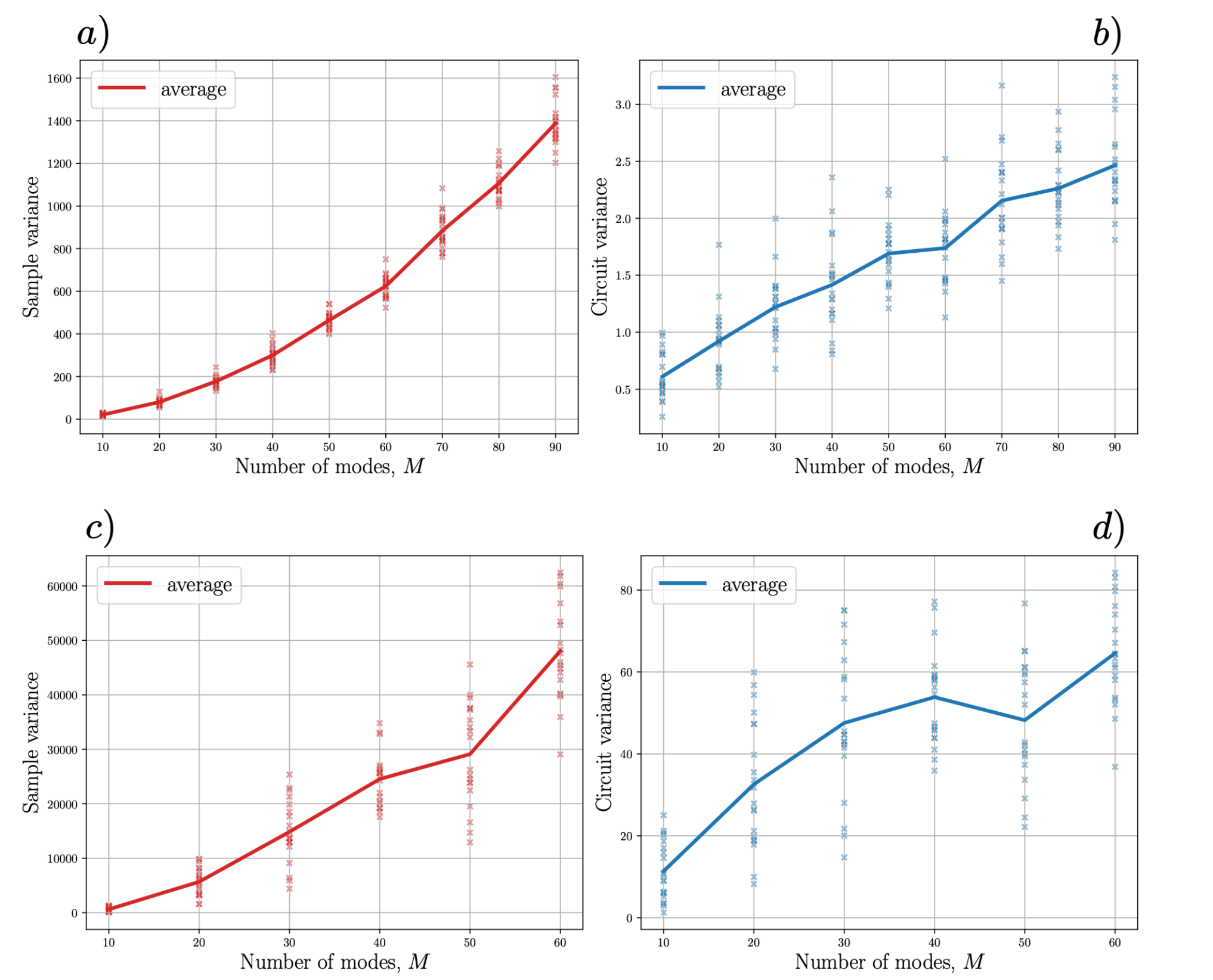}
    \caption{Numerical scaling of the circuit and sample variances for two classes of photon-number polynomials (PNPs). For each value of $M$, we generate 20 PNPs with independently sampled random coefficients. Individual instances are shown as crosses, while solid lines indicate the corresponding averages. In panels (a) and (b), the PNPs contain all photon-number monomials (PNMs) satisfying $|\bm{\alpha}|\leq 2$, with coefficients sampled uniformly from $[-1,1]$. Since the number of terms scales as $|A|\sim \mathrm{O}(M^2)$, the average sample variance in (a) exhibits approximately quadratic growth. The circuit variance in (b) grows approximately linearly, indicating that the increasing number of terms compensates for the decay of the circuit variance of the constituent PNMs. In panels (c) and (d), we consider PNPs of order $|\bm{\alpha}|\leq 4$, each containing $|A|=5M$ randomly selected PNMs with coefficients sampled uniformly from $[-1,1]$. The average sample variance grows approximately linearly, consistent with the predicted scaling $\mathrm{O}(|A|)$. The circuit variance also increases with $M$, although it shows a weak tendency towards a plateau at larger system sizes. In both cases, the absence of a decaying circuit variance and polynomially scaling sample variance indicates that the corresponding loss functions remain trainable over the system sizes considered.}
    \label{fig:pnp_variance_numerics}
\end{figure}
\subsubsection{Photon-number polynomials}
A single photon-number monomial is a restricted class of observable, and their product form makes them susceptible to efficient classical simulation methods \cite{limClassicalAlgorithmsEstimating2025}. We therefore focus on photon number polynomials, expressed as $\hat{C} = \sum_{\bm{\alpha} \in A} c_{\bm{\alpha}} \hat{M}_{\bm{\alpha}}$. Indeed, any diagonal observable in the photon number basis can be written as a finite PNP.

The PNP first moment $\mu_1(\hat{C})$ is just given by the sum over PNM first moments weighted by $c_{\bm\alpha}$. The first moment of the squared PNP is given by $\sum_{\bm\alpha, \bm\beta \in A} c_{\bm\alpha} c_{\bm\beta} \, \mu_1(\bm\alpha + \bm\beta)$ with each term again given by Theorem~\ref{thm:first_moment_main}.
Meanwhile, to compute the second moment requires computing $||\hat{C}_k^{(n)}||_2^2$, which we carry out in Appendix \ref{appendix:pnp_second_moment}. We will first discuss the sample variance.

Without making any assumptions on the coefficients, the sample variance is bounded by the constituent monomials using the Cauchy-Schwarz inequality, $\mu_1(\bm\alpha + \bm\beta) \le \sqrt{\mu_1(2\bm\alpha) \, \mu_1(2\bm\beta)}$. This gives the bound
\begin{align}
\label{eq:pnp_sample_variance_bound_main}
    \var_{\text{samp}}\big[\hat{C}\big] \le \mu_1(\hat{C}^2) \le \bigg( \sum_{\bm{\alpha} \in A} |c_{\bm{\alpha}}| \sqrt{\mu_1(2\bm{\alpha})} \bigg)^{\!2} &\le |A|^2 \, \max_{\bm{\alpha} \in A} c_{\bm{\alpha}}^2 \, \mu_1(2\bm{\alpha})
\end{align}
where $|A|$ is the cardinality of the set $A$, equal to the number of PNM terms present in the PNP.
If the polynomial has a maximum constant order, and the coefficients are bounded, each $\mu_1(2\bm{\alpha})$ is a system-size independent constant, so a PNP with polynomially many terms can be estimated to fixed additive precision $\varepsilon_{\textrm{additive}}$ with a number of shots scaling as $\mathrm{O}(|A|^2)$ -- more detail is given in Appendix \ref{appendix:pnp_estimation}.
 
For typical coefficients we can strengthen this bound. Consider a set of random coefficients drawn from a distribution satisfying $\ex[c_{\bm\alpha}] = 0$ and $|c_{\bm\alpha}|$ being independent of $M$ for all $\bm\alpha$. In Appendix \ref{appendix:pnp_estimation}, we show that for such coefficients,
\begin{align}
\label{eq:pnp_random_sign_main}
    \ex_{\{c_{\bm\alpha}\}}\Big[\var_{\text{samp}}(\hat{C})\Big] &\le \sum_{\bm{\alpha} \in A} \,\ex_{\{c_{\bm\alpha}\}}[c_{\bm{\alpha}}^2]\, \mu_1(\bm\alpha) \sim \mathrm{O}(|A|)
\end{align}
In this case, the expected sample variance improves from $\mathrm{O}(|A|^2)$ to $\mathrm{O}(|A|)$. Similarly, the same off-diagonal cancellation takes place at the level of the circuit variance, leading to
\begin{align}
\begin{split}
\label{eq:pnp_random_sign_circuit_variance_main}
    \ex_{\{c_{\bm\alpha}\}}\Big[ \var_{\text{circ}}(\hat{C})\Big] = \sum_{\bm\alpha \in A} \ex_{\{c_{\bm\alpha}\}}[c_{\bm\alpha}^2] \, \var_{\text{circ}}(\bm\alpha).
\end{split}
\end{align}
As such, the ratio of the expected variances, $\ex_{\{c_{\bm\alpha}\}} [\var_{\text{samp}}(\hat{C})] / \ex_{\{c_{\bm\alpha}\}}[\var_{\text{circ}}(\hat{C})]$, becomes a weighted average of the constituent monomial averages. A typical fixed-order PNP therefore inherits the $\mathrm{poly}(M)$ variance ratio of Eq.~\eqref{eq:pnm_variance_ratio_main} and continues to be trainable. In contrast, PNPs with a non-negligible set of monomials that scale with system size inherit their exponentially scaling variance ratio and are not trainable. Beyond these stochastically chosen PNPs, we note that highly structured PNPs can have coefficients that lead to cancellations that suppress either variance.

In Fig.~\ref{fig:pnp_variance_numerics}, we plot the sample variance and circuit variance for classes of low-weight PNPs; in (a) and (b), we consider all summing over all PNMs with weight $|\bm\alpha|\le 2$ and coefficient $c_{\bm\alpha} \sim \textrm{Uniform}[-1,1]$, while in (c) and (d), we take a set of $|A|=5M$ PNMs with weight $\le 4$. We see the sample variance for both PNMs tracks has the predicted scaling of $\mathrm{O}(|A|)$, going as $|A|\sim\mathrm{O}(M^2)$ in (a) and $|A|\sim\mathrm{O}(M)$ in (c). 

The circuit variances also increase over the range of system sizes considered, despite the circuit variance of each constituent PNM decreasing with $M$, as shown in Figs.~\ref{fig:fixed_pnm_scaling_analysis} and \ref{fig:n_alpha_plot_fixed_pnm_decay}. This behaviour can be understood as the increasing number of terms compensating for the polynomial decay of the individual PNM circuit variances. For example, in panel (b), each constituent PNM has a circuit variance that decreases with $M$, with similar behaviour to the $\alpha=2$ curve in Fig.~\ref{fig:n_alpha_plot_fixed_pnm_decay}. However, Eq.~\eqref{eq:pnp_random_sign_circuit_variance_main} shows that the coefficient-averaged circuit variance is obtained by summing the contributions from all $|A|$ constituent PNMs. Over the plotted range, the growth in $|A|$ therefore outweighs the decay of the individual contributions, resulting in an overall increase in the circuit variance.
\\
\\
\textit{Comparison to classical estimation.}
Previous literature has strongly suggested that efficiently trainable variational quantum circuits must admit efficient classical simulation, trading quantum advantage for trainability \cite{cerezo2025provable_absence}. We therefore study the complexity of classically spoofing the training of the trainable PNPs above, using one of the best-known general-purpose classical algorithms for observable estimation in linear optics, the Lim--Oh algorithm \cite{limClassicalAlgorithmsEstimating2025}. The algorithm efficiently estimates the expectation values of observables in product form over product input states, requiring a $\mathrm{O}\big(\|\hat{O}\|_2^2/ \varepsilon_{\textrm{additive}}^2\big)$ time-complexity to produce an estimate to additive error $\varepsilon_{\textrm{additive}}$. This is the classical analogue to the sample complexity of a quantum circuit. A generic PNP admits no such product form, so the algorithm must instead be applied term by term, estimating each constituent monomial and summing them, or via Monte-Carlo importance sampling over the most heavily weighted terms. In Appendix \ref{appendix:lim_oh}, we show that both strategies carry a complexity of $\mathrm{O}\big(|A|^2 /\varepsilon_{\textrm{additive}}^2\big)$. The Lim-Oh simulation method cannot exploit the cancellation of coefficients in the sample variance ratio as it must be applied to the constituent product observables individually. A quantum device sampling the polynomial observable directly does: for balanced coefficients the sample variance scales with $|A|$ by equation  \eqref{eq:pnp_random_sign_main}. This suggests a quadratic speed-up using quantum sampling versus Lim-Oh simulation methods for the task of estimating the loss function.

Since $|A|$ can grow rapidly with the maximum order of the polynomial, this separation can be relatively large even at modest system size. Moreover, in Sec.~\ref{subsec:nn_observables} we introduce neural network observables as an efficient means of implementing PNPs with exponentially many terms without ever expanding their constituent coefficients, a regime in which the term-wise classical simulation methods can be disadvantaged further. 

Other classical simulation methods generate samples that reproduce the statistics of an ideal boson-sampling distribution up to the $a^{
\text{th}}$ order correlators~\cite{villalonga2021efficient, dodd2025fast}. Therefore, for PNP observables satisfying $|\bm\alpha|\le a$, these samples reproduce the same expectation values relevant to $\hat{C} = \sum_{\bm\alpha\in A} c_{\bm\alpha} \hat{M}_{\bm\alpha}$. The number of generated samples required to estimate this expectation value therefore has the same scaling as the quantum sample complexity in Eq.
\eqref{eq:pnp_random_sign_main} or Eq.~\eqref{eq:pnp_random_sign_circuit_variance_main}, so there is no quantum sampling advantage with this approach. However, generating each classical sample has a classical cost $\mathrm{O}(M^a)$ for $a=\max_{\bm\alpha}\{|\bm\alpha|\}$~\cite{dodd2025fast}, leaving space for a potential polynomial quantum speed-up for sufficiently chosen $a$. Moreover, in their current form these methods require modification to accommodate photon-number-resolving observables and are formulated for Gaussian input states. Adapting them to the photonic VQA setting considered here will therefore be important for determining whether genuine polynomial quantum runtime advantage can persist within the trainable regime.

We therefore expect photonic variational circuits to exhibit a polynomial speed-up in the runtime of variational protocols on a large class of PNP observables when compared with the currently available classical algorithms. Conversely, the fact that the classical estimate remains polynomial-time motivates a train-on-classical, deploy-on-quantum regime~\cite{recio2025train, kolarovszki2026generative}. Here, the pre-training of a circuit can be performed classically with some overhead (either $|A|$ with the Lim-Oh approach~\cite{limClassicalAlgorithmsEstimating2025} or $M^{|\bm\alpha|}$ with the boson sampling emulator approach~\cite{villalonga2021efficient, dodd2025fast}).

\subsubsection{Neural network observables}
\label{subsec:nn_observables}
A natural way of implementing photon-number polynomials is to consider a neural network as an observable of a variational quantum circuit. Each photon number sample is passed in as the input to a multi-layered neural network. We take the activation functions in the neural network to be quadratic. This way, we can simply write out the action of the neural network as a multivariate polynomial, allowing us to apply previous results.

\begin{figure}[ht]
    \centering
    \hspace{0.8cm}
    \includegraphics[width=0.8\textwidth]{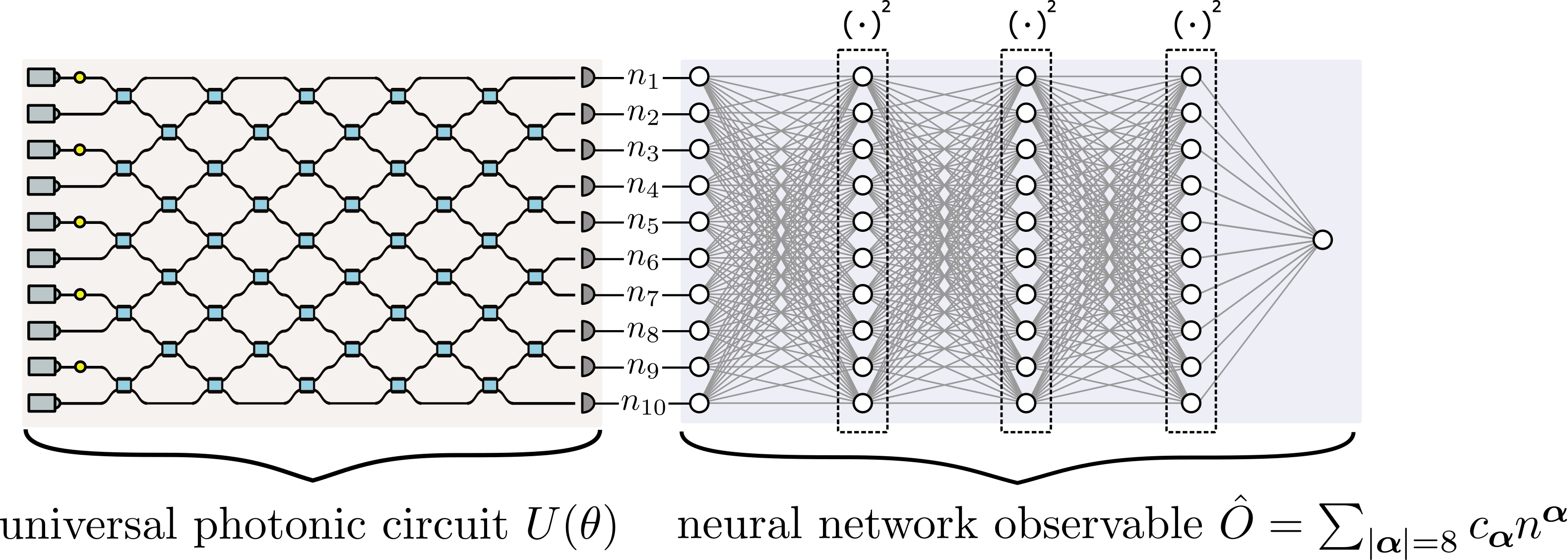}
    \caption{We show a 3-layer neural network observable implementing a photon-number polynomial operation. In red, we highlight a path that adds weight to one of the monomial terms.}
    \label{fig:nn_observable}
\end{figure}

Consider a network with L linear layers, each interleaved with a quadratic activation function, whose final layer has a single output neuron. Its action can be written as
\begin{equation}
    \mathrm{NN}(\bm{n}) = \sum_{|\bm{\alpha}| \le 2^L}  c_{\bm{\alpha}} \, \bm{n}^{\bm{\alpha}}.
\end{equation}
In words, the NN observable is equal to a PNP including all possible PNM terms up to order $2^{L}$. There are $\binom{2^L + M}{M}$ terms in this polynomial. For $2^L \ll M$, this scales as $\approx M^{2^L}/(2^L)!$. The coefficients $c_{\bm{\alpha}}$ are determined by the weights and biases of the linear layers.

While the polynomial representation grows exponentially in size, applying the forwards pass of a neural network to a sample has complexity $O(LM^2)$: we need to multiply the sample by $L$ size $M \times M$ matrices, and apply quadratic activation functions to the intermediary terms. This is a key advantage of the NN representation: it sidesteps the need to enumerate the full polynomial explicitly. By contrast, classical observable-estimation methods that require the observable to be expressed as a sum of product observables must first expand the network into this polynomial form, and then either estimate each term individually or use a weighted estimation scheme. In Appendix \ref{appendix:lim_oh}, we discuss various classical strategies for estimating polynomial observables.

\subsubsection{Output probability estimation}
Several proposals for boson sampling applications rely on estimating either a single output probability or a scaled aggregate of output probabilities. Although such observables are of theoretical interest, we show that they cannot be estimated to proportional error by a physical boson sampler with polynomial resources. This rules out their use in variational protocols, and further suggests that previous proposals relying on these observable estimates are unscalable without careful consideration of the magnitudes of probabilities being estimated\\
\\
\textit{Single-probability estimation}
A natural loss function to consider is one governed by a projection onto a state $\ketbra{\psi}{\psi}$. Indeed, estimating the expectation value $\tr\left(\rho(\theta) \ketbra{\psi}{\psi}\right)$ has been proposed as an application of boson sampling towards problems in graph theory \cite{bradler2018graph, mezher2023solving}. We will show that the trainability ratio of projector observables is exponentially large, illustrating their non-suitability for variational protocols.

As before, we will consider an input state with fixed photon number, and will set $\ket{\psi} = \ket{\mathbf{s}} = \ket{s_1, ..., s_M}$. In Appendix \ref{appendix:single_outcome_variance}, we compute the Haar-average moments of the projector observable, from which the variances follow,
\begin{align}
    \text{Var}_{\text{circ}}(\ketbra{\mathbf{s}}{\mathbf{s}}) &= (P_2(M,n) - 1)|\mathcal{H}|^{-2}, \\
    \text{Var}_{\text{samp}}(\ketbra{\mathbf{s}}{\mathbf{s}}) &= |\mathcal{H}|^{-1} \left(1 - P_2(M,n)|\mathcal{H}|^{-1}\right),
\end{align}
where $P_2(M,n)$ is the normalised average collision probability, shown to scale as $\Theta(n + n/M + 1)$ \cite{mhiri2026bosonsamplingdiluteregime}. We get a variance ratio:

\begin{align}
    \frac{\text{Var}_{\text{samp}}(\ketbra{\mathbf{s}}{\mathbf{s}})}{\text{Var}_{\text{circ}}(\ketbra{\mathbf{s}}{\mathbf{s}})} = \frac{|\mathcal{H}| - P_2(M,n)}{P_2(M,n) - 1} = \Theta\left(\frac{|\mathcal{H}|}{n + n/M + 1}\right).
\end{align}
With $|\mathcal{H}| = \binom{n+M-1}{n}$ growing exponentially in $n, M$, the difficulty of resolving the difference between two loss functions governed by the projector grows exponentially with problem size, ruling out its use in trainable variational protocols at scale. It may be natural to consider scaling the projector by a constant $c$ to normalise the circuit variance to be constant (indeed, this is done to relate estimated probabilities to permanents \cite{mezher2023solving}). Howeever, the variance ratio suggests that estimating such quantities to proportional precision remains exponentially hard -- constant factors that increase the circuit variance also increase the sample variance by the same factor, resulting in no change in the variance ratio. We expect this makes the scaling of applications relying on resolving output probabilities challenging without careful consideration of the circuit and projector. 

We note that this is the photonic analogue of a known gate-based circuit result: loss functions measuring fidelity to a target state produce barren plateaus even in shallow circuits \cite{cerezo2021cost}.\\
\\
\textit{Coarse-grained sampling}
To avoid the exponentially hard task of resolving a single output probability of a quantum circuit, it has been proposed to estimate the probability of a collection of output measurements, a regime known as `coarse-grained sampling' \cite{nikolopoulos2016decision, wang2016certification}. We show that, while the sample variance of such observables is indeed constant, the circuit variance is exponentially small. This limits the scalability of such regimes for variational protocols.

Furthermore, proposals for the use of coarse-grained boson sampling for cryptographic applications and proof-of-work schemes relying on resolving which set of output probabilities is larger are likely unscalable without careful circuit and observable selection.

A coarse-grained boson sampling observable is defined as a uniformly weighted sum of the single-outcome projectors as considered above, chosen from a set $S \subseteq \mathcal{H}_{n,M}$,
\begin{equation}
    \hat{O}_S = \sum_{\mathbf{s} \in S} \ketbra{\mathbf{s}}{\mathbf{s}}, \qquad S \subseteq \mathcal{H}_{n,M},
\end{equation}
The expectation value of this observable is the total probability mass assigned to the bin $S$. There have been several proposals which consider estimating such expectation values to decide which of several bins is most likely as a cryptographic one-way function that is hard to invert, but computable with a boson sampler \cite{singh2025proof,nikolopoulos2019cryptographic}. Following these proposals, we will consider bins comprising a constant fraction of outcomes, $|S| = |\mathcal{H}|/K$ for constant $K$.

Sampling the observable $\hat{O}_S$ with a Haar-randomly generated circuit is a Bernoulli variable with success probability $\mathbb{E}_\mu\big[\mathbb{P}(S|\theta)\big] = |S|/|\mathcal{H}| = 1/K$. As such, the total Haar-average variance is $\frac{1}{K}(1-\frac{1}{K})$, and the law of total variance gives the identity,
\begin{equation}
    \frac{1}{K}\left(1 - \frac{1}{K}\right) = \text{Var}_{\text{samp}}(\hat{O}_S) + \text{Var}_{\text{circ}}(\hat{O}_S).
\end{equation}
Sample and circuit variance of a constant-fraction coarse-grained output observable sum to a constant. We show below that the circuit variance is typically exponentially small, so we can conclude the sample variance is $\Theta(1)$.

Unlike for single outcome probabilities, the circuit variance of the coarse-grained observable depends on the choice of bin $S$, so we compute its average over sets $S$ drawn uniformly at random among the constant-fraction subsets of $\mathcal{H}$. In Appendix \ref{appendix:binning}, we derive the average circuit variance:
\begin{align}
\begin{split}
\label{eq:course_grained_results_main}
\mathbb{E}_{S}\left[\text{Var}_{\text{circ}}(\hat{O}_S)\right] \approx \frac{1}{K}\left(1 - \frac{1}{K}\right)\frac{P_2(M,n) - 1}{|\mathcal{H}|}.
\end{split}
\end{align}
As we mention previously, $P_2(M,n) = \Theta(n + n/M + 1)$ \cite{mhiri2026bosonsamplingdiluteregime}. As $|\mathcal{H}|$ is exponentially growing in $n,M$, the above result means that only an exponentially small fraction of choices of sets $S$ can have non-exponentially small circuit variance. As such, almost all choices of output aggregation concentrate around the mean and exhibit a barren plateau.

While such coarse-grained observables have a constant sample variance unlike single output probability estimation, the exponentially small circuit variance means their trainability ratio remains exponentially large: an exponentially large number of samples is needed to resolve the value of such observables with sufficient precision. This does not preclude carefully structured output binning from admitting polynomial variance ratios, but our results show that such binnings constitute an exponentially small fraction of all choices, and identifying them is a nontrivial task that proposals relying on coarse-grained estimates must address.

\section{Numerical simulations and photonic experiments} 
\label{sec:numerical_simulations_and_photonic_experiments}
\subsection{Single-photon input states}

We consider the experimentally relevant regime of training circuits with an input state of $n = M/2$ photons being propagated over $M$ modes. As discussed in Sec.~\ref{subsubsec:pnps}, for such a regime, the number of samples needed to resolve gradients for each parameter scales polynomially in system size for a fixed-order photon-number polynomial.

\begin{figure}[ht]
    \centering
    \includegraphics[width=0.9\textwidth]{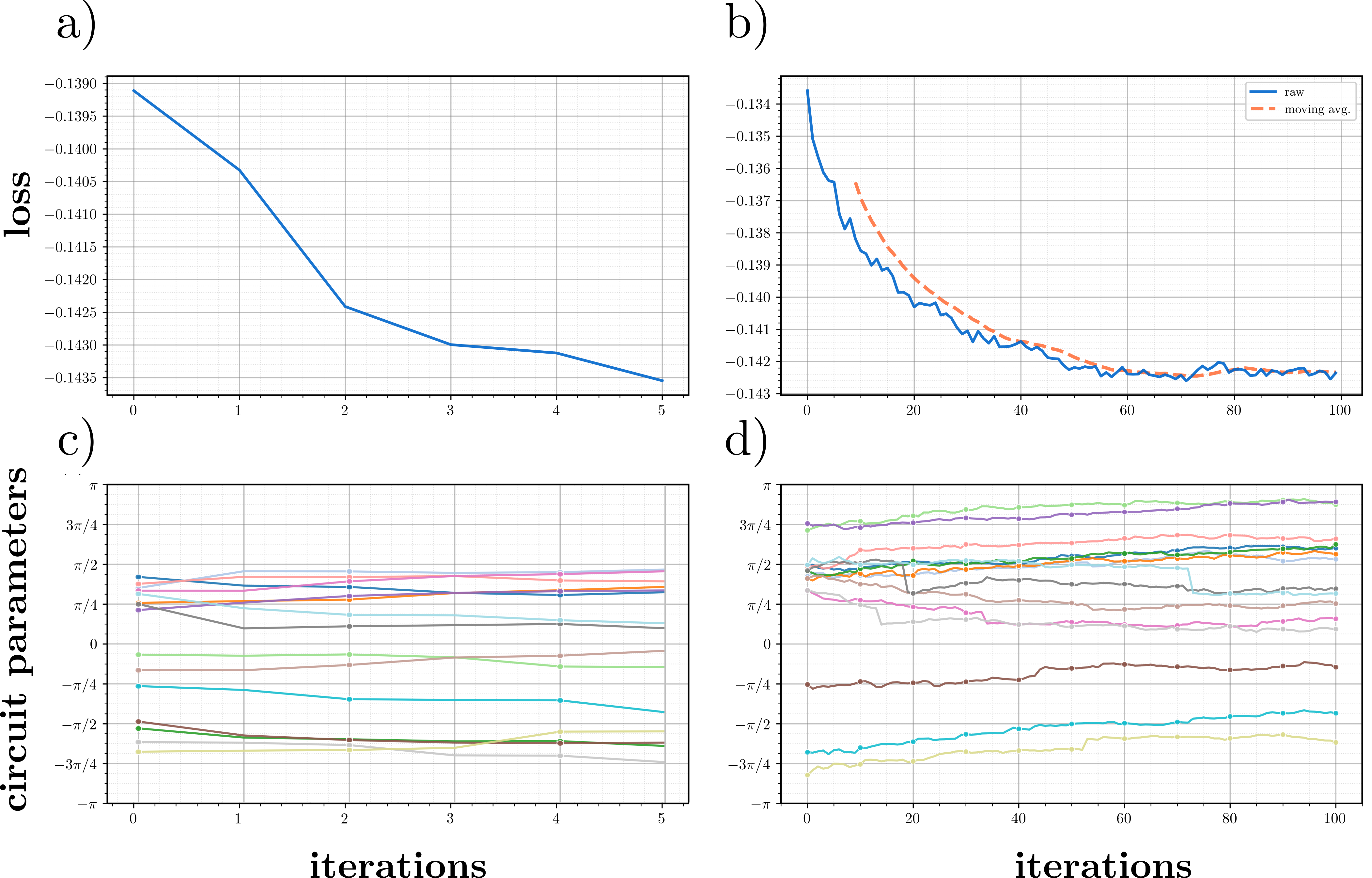}
    \caption{Training of simulated photonic circuits on a 3-layer neural network observable. (a,c) a 13-photon, 26-mode configuration trained for 6 iterations, and (b,d) a 10-photon, 20-mode configuration trained for 100 iterations. 10,000 samples are taken for each gradient estimate using the generalized parameter shift rule (GPSR) taken up to the 8th frequency. (c) and (d) illustrate the evolution of 15 randomly chosen circuit parameters, showing smooth trajectories.}
    \label{fig:simulation_training}
\end{figure}

We train simulated photonic circuits arranged in a Clements scheme beam splitter configuration, initialised at a Haar-random unitary. We take the observable to be a randomly initialised 3-layer neural network, giving a total observable that can be written as a photon-number polynomial with all terms up to order 8. Our results show that a modest sample budget of 10,000 samples per gradient update remains sufficient for training intermediate scale regimes. Due to the prohibitive simulation complexity of photonic circuits, we present training runs for only 20 and 26 modes.

\subsection{Experiments on photonic hardware}
We demonstrate training photonic hardware with the ORCA PT-2. We illustrate that near-term photonic hardware successfully trains with high accuracy on loss functions given by low-order photon-number polynomials. The ORCA PT-2 is a time-bin encoded photonic quantum processor. It samples from circuits with two cascading layers of beam splitters, corresponding to two optical delay lines. While such circuits are not universal and cannot realise arbitrary Haar-random unitaries, sampling from their output distribution remains classically hard \cite{deshpande2022quantum, novak2025boundaries}. Moreover, as they are significantly shallower than universal linear-optical circuits, they are generally expected to be easier to train, with training complexity no greater than that of the deeper architectures considered in our analysis. We illustrate a schematic of the hardware in Fig.~\ref{fig:PT2_system}.

We train the PT-2 on a 100-term photon-number polynomial up to order 2, with randomly chosen weights between $$[-5,5]$$. We post-select samples to have a fixed number of photons. Under assumptions of uniform loss, even with Gaussian input states, such a sampling regime can be equivalently described through an input superposition state of fixed photon number. The theory developed in Sec.~\ref{sec:trainability} applies directly to such a regime. In Fig.~\ref{fig:PT2_training}, we see that with a modest budget of 1000 samples per gradient estimate, the PT-2 can efficiently train on a large low-order photon-number polynomial.

\begin{figure}[t]
    \centering
    \hspace{0.9cm}
    \includegraphics[width=0.9\textwidth]{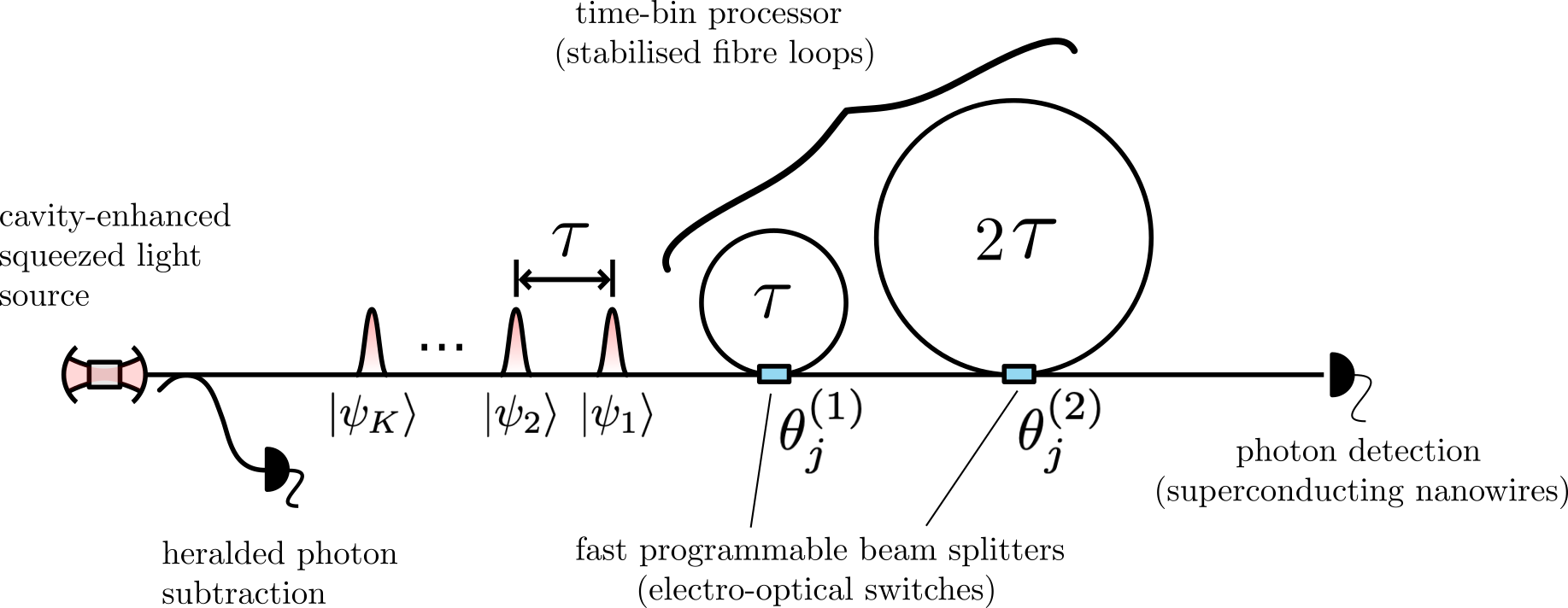}
    \caption{Illustration of the PT-2 photonic system. Squeezed states of light are generated at time intervals separated by delay $\tau$. Samples are accepted when a photon subtraction has been heralded on at least one input squeezed state, resulting in a non-Gaussian input. The input state is interfered between modes through two sequential delay loops. An array of nanowire detectors read out pseudo-photon number resolving measurements.}
    \label{fig:PT2_system}
\end{figure}

\begin{figure}[t]
    \centering
    \includegraphics[width=0.9\textwidth]{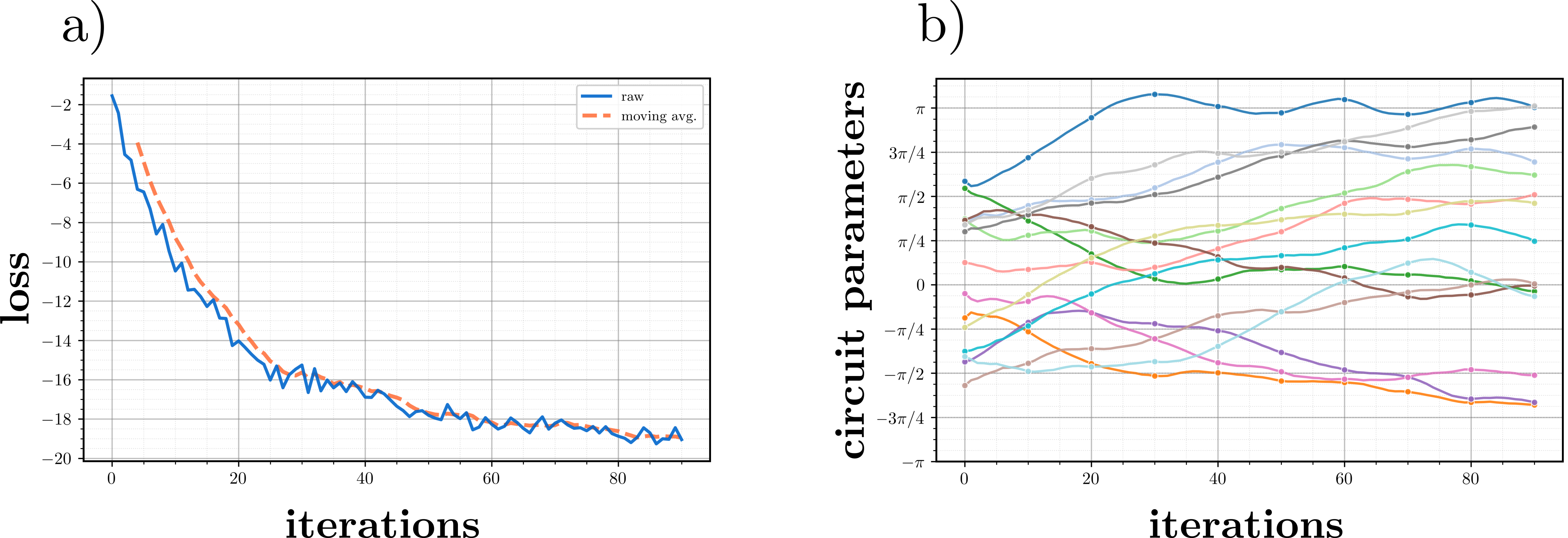}
    \caption{Training of a PT-2 system to minimise loss. The PT-2 is run for 48 modes per sample, with an alternating 1,0 input state of 24 squeezed states generated. Samples are post-selected on having exactly 6 photons present. The loss function is set with respect to minimise a 100-term photon-number polynomial up to order 2, with random weights chosen between [-5,5] for each term. 1000 samples are taken for each gradient estimate, which is done through the parameter shift rule. Figure a) shows loss convergence smoothly towards a local minimum. Figure b) illustrates the smooth training of a randomly chosen set of 15 circuit parameters.}
    \label{fig:PT2_training}
\end{figure}

This result motivates training near-term photonic system on loss functions governed by the minimisation of a low-order photon-number polynomial observable. Indeed, many relevant classes of observables, such as those arising from quadratic unconstrained binary optimisation (QUBO) problems, can be described as order two photon-number polynomials, and previous proposals consider training photonic circuits on such objectives \cite{makarovskiy2025binary, cazalis2023gaussian}. It remains an open question to quantify the effect of experimental imperfections such as errors and parameter jitter on the trainability of photonic circuits; photon losses are expected to be problematic as systems increase in size, so understanding how post-selecting on error-free outcomes or implementing photonic error mitigation~\cite{lee2022error,taylor2024quantum} impact total training time will be important.
Nonetheless, these preliminary experimental results indicate that NISQ photonic systems can be trained despite imperfections.

\section{Discussion and conclusion}
In this work, we introduced the ratio between sample variance and circuit variance as a natural measure of the complexity of training variational quantum circuits. We showed that careful consideration of both variances is necessary in order to determine trainability. Using this framework, we investigated the trainability of a broad class of photonic observables, identifying fixed-order photon-number polynomials as efficiently trainable: a number of circuit samples growing polynomially with system size is sufficient to train photonic circuits to minimise the expected value of such observables. We also showed that observables based on output probability estimation and high-order photon-number polynomials are not suitable for scalable variational protocols, as they have exponentially growing sample demands for loss resolution.

We further demonstrated that fixed-order photon-number polynomials admit efficient implementations through neural network observables, providing a practical and expressive class of loss functions for variational photonic protocols. We showed that these observables appear to only require only polynomially scaling sampling resources for training and analysed the complexity of their classical simulation, identifying regimes in which quantum devices offer a polynomial advantage over existing classical methods. In Sec.~\ref{sec:numerical_simulations_and_photonic_experiments}, we validated these predictions through successful variational circuit training in both simulation and on near-term photonic hardware. Together, these results provide formal support for the empirical scalability of recent photonic quantum machine learning experiments and identify a broad family of observables for which variational photonic quantum circuits remain trainable at large system sizes.

Several theoretical questions remain open. Our analysis relies on a conjectured scaling of the gap between the first and second Haar-average moments, supported by numerical evidence but not yet established analytically. We have also restricted our treatment to fixed-photon-number input states. Extending these results to superposition states or states with indefinite photon number would provide stronger theoretical foundations for experimentally relevant regimes such as Gaussian input states. Similarly, photon losses are the dominant error channel for linear optical networks and so understanding how these errors can change the trainability landscape is important in situations where one cannot post-select on the noise-free outputs.
Finally, while our variance analysis establishes when gradients can be estimated efficiently, it does not address other optimisation challenges that arise during training, such as avoiding local minima or understanding the influence of hardware imperfections.

Our results establish a theoretical foundation for understanding trainability in variational photonic quantum computing and provide practical guidance for the design of scalable photonic quantum machine learning protocols. By identifying observable classes that simultaneously avoid barren plateaus, remain sample-efficient, and admit favourable complexity compared with current classical estimation methods, this work helps bridge the gap between recent large-scale photonic demonstrations and the theoretical understanding of variational photonic quantum algorithms. We expect these results to provide a promising direction for the development of photonic quantum machine learning algorithms capable of delivering practical quantum advantage.

\section*{Acknowledgements}
We thank Alex Jones, Josh Nunn and the rest of the team at ORCA Computing for insightful discussions and valuable comments on the manuscript. MSK, IAW and ZL acknowledge funding from the UK EPSRC through EP/T001062/1, EP/Z53318X/1, EP/W032643/1, and EP/T517811/1. MSK thanks the KIST through Open Innovation fund and the National Research Foundation of Korea grant funded by the Korean government (MSIT) (No. RS-2024-00413957) for their financial support.  IAW is supported by the UK Research \& Innovation (project number: MR/W011794/1). ZL acknowledges partial funding support from ORCA Computing.

\paragraph{Note added.}
During the preparation of this manuscript, we became aware of independent and related work by Monbroussou et al. that studies the trainability of linear optical circuits \cite{monbroussou2026classical}.

\printbibliography

@inproceedings{aaronson2011computational,
    author = {Aaronson, Scott and Arkhipov, Alex},
    title = {The computational complexity of linear optics},
    year = {2011},
    isbn = {9781450306911},
    publisher = {Association for Computing Machinery},
    address = {New York, NY, USA},
    url = {https://doi.org/10.1145/1993636.1993682},
    doi = {10.1145/1993636.1993682},
    booktitle = {Proceedings of the Forty-Third Annual ACM Symposium on Theory of Computing},
    pages = {333–342},
    numpages = {10},
    location = {San Jose, California, USA},
    series = {STOC '11}
}

@article{larocca2025barren,
  title={Barren plateaus in variational quantum computing},
  author={Larocca, Martin and Thanasilp, Supanut and Wang, Samson and Sharma, Kunal and Biamonte, Jacob and Coles, Patrick J and Cincio, Lukasz and McClean, Jarrod R and Holmes, Zo{\"e} and Cerezo, Marco},
  journal={Nature Reviews Physics},
  pages={1--16},
  year={2025},
  publisher={Nature Publishing Group UK London}
}

@article{facelli2024exactgradients,
  title={Exact gradients for linear optics with single photons},
  author={Facelli, Giorgio and Roberts, David D and Wallner, Hugo and Makarovskiy, Alexander and Holmes, Zo{\"e} and Clements, William R},
  journal={arXiv preprint arXiv:2409.16369},
  year={2024}
}

@article{clements2016optimal,
  title={Optimal design for universal multiport interferometers},
  author={Clements, William R and Humphreys, Peter C and Metcalf, Benjamin J and Kolthammer, W Steven and Walmsley, Ian A},
  journal={Optica},
  volume={3},
  number={12},
  pages={1460--1465},
  year={2016},
  publisher={Optical Society of America}
}

@article{collins2006integration,
  title={Integration with respect to the Haar measure on unitary, orthogonal and symplectic group},
  author={Collins, Beno{\^\i}t and {\'S}niady, Piotr},
  journal={Communications in Mathematical Physics},
  volume={264},
  number={3},
  pages={773--795},
  year={2006},
  publisher={Springer}
}

@article{reck1994experimental,
  title = {Experimental realization of any discrete unitary operator},
  author = {Reck, Michael and Zeilinger, Anton and Bernstein, Herbert J. and Bertani, Philip},
  journal = {Phys. Rev. Lett.},
  volume = {73},
  issue = {1},
  pages = {58--61},
  numpages = {0},
  year = {1994},
  publisher = {American Physical Society},
  doi = {10.1103/PhysRevLett.73.58},
  url = {https://link.aps.org/doi/10.1103/PhysRevLett.73.58}
}

@article{arrasmith2022equivalence,
  title={Equivalence of quantum barren plateaus to cost concentration and narrow gorges},
  author={Arrasmith, Andrew and Holmes, Zoe and Cerezo, M., and Coles, Patrick J.},
  journal={Quantum Sci. Technol.},
  volume={7},
  number={045015},
  year={2022}}

@article{cerezo2025provable_absence,
  title={Does provable absence of barren plateaus imply classical simulability?},
  author={Cerezo, Marco and Larocca, Martin and Garc{\'\i}a-Mart{\'\i}n, Diego and Diaz, Nelson L and Braccia, Paolo and Fontana, Enrico and Rudolph, Manuel S and Bermejo, Pablo and Ijaz, Aroosa and Thanasilp, Supanut and others},
  journal={Nature Communications},
  volume={16},
  number={1},
  pages={7907},
  year={2025},
  publisher={Nature Publishing Group UK London}
}

@misc{limClassicalAlgorithmsEstimating2025,
  title = {Classical Algorithms for Estimating Expectation Values in Linear-Optical Circuits},
  author = {Lim, Youngrong and Oh, Changhun},
  year = 2025,
  number = {arXiv:2502.12882},
  eprint = {2502.12882},
  primaryclass = {quant-ph},
  publisher = {arXiv},
  doi = {10.48550/arXiv.2502.12882},
  urldate = {2025-12-03},
  archiveprefix = {arXiv},
}

@article{dodd2025fast,
  title={A fast and frugal Gaussian Boson Sampling emulator},
  author={Dodd, Tom and Mart{\'\i}nez-Cifuentes, Javier and Brown, Oliver Thomson and Quesada, Nicol{\'a}s and Garc{\'\i}a-Patr{\'o}n, Ra{\'u}l},
  journal={arXiv preprint arXiv:2511.14923},
  year={2025}
}

@article{villalonga2021efficient,
  title={Efficient approximation of experimental Gaussian boson sampling},
  author={Villalonga, Benjamin and Niu, Murphy Yuezhen and Li, Li and Neven, Hartmut and Platt, John C and Smelyanskiy, Vadim N and Boixo, Sergio},
  journal={arXiv preprint arXiv:2109.11525},
  year={2021}
}

@article{lee2022error,
  title={Error-mitigated photonic variational quantum eigensolver using a single-photon ququart},
  author={Lee, Donghwa and Lee, Jinil and Hong, Seongjin and Lim, Hyang-Tag and Cho, Young-Wook and Han, Sang-Wook and Shin, Hyundong and ur Rehman, Junaid and Kim, Yong-Su},
  journal={Optica},
  volume={9},
  number={1},
  pages={88--95},
  year={2022},
  publisher={Optical Society of America}
}

@article{taylor2024quantum,
  title = {Quantum error cancellation in photonic systems: Undoing photon losses},
  author = {Taylor, Adam and Bressanini, Gabriele and Kwon, Hyukjoon and Kim, M. S.},
  journal = {Phys. Rev. A},
  volume = {110},
  issue = {2},
  pages = {022622},
  numpages = {19},
  year = {2024},
  publisher = {American Physical Society},
  doi = {10.1103/PhysRevA.110.022622},
  url = {https://link.aps.org/doi/10.1103/PhysRevA.110.022622}
}

@article{collins2014integration,
  title={Integration of invariant matrices and moments of inverses of Ginibre and Wishart matrices},
  author={Collins, Benoit and Matsumoto, Sho and Saad, Nadia},
  journal={Journal of Multivariate Analysis},
  volume={126},
  pages={1-13},
  year={2014}
}

@misc{mhiri2026bosonsamplingdiluteregime,
      title={Boson sampling beyond the dilute regime: second moments and anti-concentration}, 
      author={Hela Mhiri and Hugo Thomas and Léo Monbroussou and Ulysse Chabaud and Zoë Holmes and Elham Kashefi},
      year={2026},
      eprint={2604.14323},
      archivePrefix={arXiv},
      primaryClass={quant-ph},
      url={https://arxiv.org/abs/2604.14323}, 
}

@misc{taylor2025optimal,
      title={Optimal Quantum Information Transmission Under a Continuous-Variable Erasure Channel}, 
      author={Taylor, Adam and Hanks, Michael and Kwon, Hyukjoon and Kim, M. S.},
      year={2025},
      eprint={2510.01424},
      archivePrefix={arXiv},
      primaryClass={quant-ph},
      url={https://arxiv.org/abs/2510.01424}, 
}

@article{belovas2023asymptotic,
url = {https://doi.org/10.1515/ms-2023-0026},
title = {The Asymptotics of the Geometric Polynomials},
title = {},
author = {Igoris Belovas},
pages = {335--342},
volume = {73},
number = {2},
journal = {Mathematica Slovaca},
doi = {doi:10.1515/ms-2023-0026},
year = {2023},
lastchecked = {2026-07-19}
}

@article{cerezo2021variational,
  title={Variational quantum algorithms},
  author={Cerezo, Marco and Arrasmith, Andrew and Babbush, Ryan and Benjamin, Simon C and Endo, Suguru and Fujii, Keisuke and McClean, Jarrod R and Mitarai, Kosuke and Yuan, Xiao and Cincio, Lukasz and others},
  journal={Nature Reviews Physics},
  volume={3},
  number={9},
  pages={625--644},
  year={2021},
  publisher={Nature Publishing Group UK London}
}

@article{gong2025enhanced,
  title={Enhanced Image Recognition Using Gaussian Boson Sampling},
  author={Gong, Si-Qiu and Chen, Ming-Cheng and Liu, Hua-Liang and Su, Hao and Gu, Yi-Chao and Tang, Hao-Yang and Jia, Meng-Hao and Deng, Yu-Hao and Wei, Qian and Wang, Hui and others},
  journal={arXiv preprint arXiv:2506.19707},
  year={2025}
}

@article{cimini2026large,
  title={Large-scale quantum reservoir computing using a Gaussian Boson Sampler},
  author={Cimini, Valeria and Sohoni, Mandar M and Presutti, Federico and Malia, Benjamin K and Ma, Shi-Yuan and Yanagimoto, Ryotatsu and Wang, Tianyu and Onodera, Tatsuhiro and Wright, Logan G and McMahon, Peter L},
  journal={npj Quantum Information},
  year={2026},
  publisher={Nature Publishing Group UK London}
}

@article{ragone2024lie,
  title={A lie algebraic theory of barren plateaus for deep parameterized quantum circuits},
  author={Ragone, Michael and Bakalov, Bojko N and Sauvage, Fr{\'e}d{\'e}ric and Kemper, Alexander F and Ortiz Marrero, Carlos and Larocca, Mart{\'\i}n and Cerezo, Marco},
  journal={Nature Communications},
  volume={15},
  number={1},
  pages={7172},
  year={2024},
  publisher={Nature Publishing Group UK London}
}

@article{wang2021noise,
  title={Noise-induced barren plateaus in variational quantum algorithms},
  author={Wang, Samson and Fontana, Enrico and Cerezo, Marco and Sharma, Kunal and Sone, Akira and Cincio, Lukasz and Coles, Patrick J},
  journal={Nature communications},
  volume={12},
  number={1},
  pages={6961},
  year={2021},
  publisher={Nature Publishing Group UK London}
}

@article{mcclean2018barren,
  title={Barren plateaus in quantum neural network training landscapes},
  author={McClean, Jarrod R and Boixo, Sergio and Smelyanskiy, Vadim N and Babbush, Ryan and Neven, Hartmut},
  journal={Nature communications},
  volume={9},
  number={1},
  pages={4812},
  year={2018},
  publisher={Nature Publishing Group UK London}
}

@article{jones2019variational,
  title={Variational quantum algorithms for discovering Hamiltonian spectra},
  author={Jones, Tyson and Endo, Suguru and McArdle, Sam and Yuan, Xiao and Benjamin, Simon C},
  journal={Physical Review A},
  volume={99},
  number={6},
  pages={062304},
  year={2019},
  publisher={APS}
}

@article{mcardle2020quantum,
  title={Quantum computational chemistry},
  author={McArdle, Sam and Endo, Suguru and Aspuru-Guzik, Al{\'a}n and Benjamin, Simon C and Yuan, Xiao},
  journal={Reviews of Modern Physics},
  volume={92},
  number={1},
  pages={015003},
  year={2020},
  publisher={APS}
}

@article{abbas2021power,
  title={The power of quantum neural networks},
  author={Abbas, Amira and Sutter, David and Zoufal, Christa and Lucchi, Aur{\'e}lien and Figalli, Alessio and Woerner, Stefan},
  journal={Nature computational science},
  volume={1},
  number={6},
  pages={403--409},
  year={2021},
  publisher={Nature Publishing Group US New York}
}

@article{farhi2014quantum,
  title={A quantum approximate optimization algorithm},
  author={Farhi, Edward and Goldstone, Jeffrey and Gutmann, Sam},
  journal={arXiv preprint arXiv:1411.4028},
  year={2014}
}

@article{shor1999polynomial,
  title={Polynomial-time algorithms for prime factorization and discrete logarithms on a quantum computer},
  author={Shor, Peter W},
  journal={SIAM review},
  volume={41},
  number={2},
  pages={303--332},
  year={1999},
  publisher={SIAM}
}

@article{bradler2018graph,
  title={Graph isomorphism and Gaussian boson sampling},
  author={Br{\'a}dler, Kamil and Friedland, Shmuel and Izaac, Josh and Killoran, Nathan and Su, Daiqin},
  journal={arXiv preprint arXiv:1810.10644},
  year={2018}
}

@article{mezher2023solving,
  title={Solving graph problems with single photons and linear optics},
  author={Mezher, Rawad and Carvalho, Ana Filipa and Mansfield, Shane},
  journal={Physical Review A},
  volume={108},
  number={3},
  pages={032405},
  year={2023},
  publisher={APS}
}

@article{monbroussou2025towards,
  title = {Toward quantum advantage with photonic state injection},
  author = {Monbroussou, L\'eo and Mamon, Eliott Z. and Thomas, Hugo and Yacoub, Verena and Chabaud, Ulysse and Kashefi, Elham},
  journal = {Phys. Rev. Res.},
  volume = {7},
  issue = {3},
  pages = {033051},
  numpages = {14},
  year = {2025},
  publisher = {American Physical Society},
  doi = {10.1103/PhysRevResearch.7.033051},
  url = {https://link.aps.org/doi/10.1103/PhysRevResearch.7.033051}
}

@article{trudeau2026parameter,
  title={Parameter-Shift Rules for Gradients in Boson Sampling Experiments},
  author={Trudeau, Marius and Emeriau, Pierre-Emmanuel and Quesada, Nicol{\'a}s},
  journal={arXiv preprint arXiv:2607.15160},
  year={2026}
}

@inproceedings{cazalis2023gaussian,
  title={Gaussian Boson Sampling for binary optimization},
  author={Cazalis, Jean and Chai, Yahui and Jansen, Karl and K{\"u}hn, Stefan and Shah, Tirth},
  booktitle={2023 IEEE International Conference on Quantum Computing and Engineering (QCE)},
  volume={2},
  pages={332--333},
  year={2023},
  organization={IEEE}
}

@article{efrom1981jackknife,
author = {B. Efron and C. Stein},
title = {{The Jackknife Estimate of Variance}},
volume = {9},
journal = {The Annals of Statistics},
number = {3},
publisher = {Institute of Mathematical Statistics},
pages = {586 -- 596},
year = {1981},
doi = {10.1214/aos/1176345462},
URL = {https://doi.org/10.1214/aos/1176345462}
}

@article{russell2017direct,
  title={Direct dialling of Haar random unitary matrices},
  author={Russell, Nicholas J and Chakhmakhchyan, Levon and O’Brien, Jeremy L and Laing, Anthony},
  journal={New journal of physics},
  volume={19},
  number={3},
  pages={033007},
  year={2017},
  publisher={IOP Publishing}
}

@article{bobkov1999isoperimetric,
  title={Isoperimetric and analytic inequalities for log-concave probability measures},
  journal={The Annals of Probability},
  volume={27},
  number={4},
  pages={1903--1921},
  year={1999},
  publisher={Institute of Mathematical Statistics}
}

@article{novak2025boundaries,
  title={Boundaries for quantum advantage with single photons and loop-based time-bin interferometers},
  author={Nov{\'a}k, Samo and Roberts, David D and Makarovskiy, Alexander and Garc{\'\i}a-Patr{\'o}n, Ra{\'u}l and Clements, William R},
  journal={Quantum},
  volume={9},
  pages={1915},
  year={2025},
  publisher={Verein zur F{\"o}rderung des Open Access Publizierens in den Quantenwissenschaften}
}

@article{makarovskiy2025binary,
  title={A Binary Optimisation Algorithm for Near-Term Photonic Quantum Processors},
  author={Makarovskiy, Alexander and Slysz, Mateusz and Siera, Dawid and Farnsworth, Thorin and Clements, William R and Rydlichowski, Piotr and Kurowski, Krzysztof and others},
  journal={arXiv preprint arXiv:2510.08274},
  year={2025}
}

@article{cardin2024photon,
  title={Photon-number moments and cumulants of Gaussian states},
  author={Cardin, Yanic and Quesada, Nicol{\'a}s},
  journal={Quantum},
  volume={8},
  pages={1521},
  year={2024},
  publisher={Verein zur F{\"o}rderung des Open Access Publizierens in den Quantenwissenschaften}
}

@article{hoch2025variational,
  title={Variational approach to photonic quantum circuits via the parameter shift rule},
  author={Hoch, Francesco and Rodari, Giovanni and Giordani, Taira and Perret, Paul and Spagnolo, Nicol{\`o} and Carvacho, Gonzalo and Pentangelo, Ciro and Piacentini, Simone and Crespi, Andrea and Ceccarelli, Francesco and others},
  journal={Physical Review Research},
  volume={7},
  number={2},
  pages={023227},
  year={2025},
  publisher={APS}
}

@article{pappalardo2025photonic,
  title={Photonic parameter-shift rule: Enabling gradient computation for photonic quantum computers},
  author={Pappalardo, Axel and Emeriau, Pierre-Emmanuel and de Felice, Giovanni and Ventura, Brian and Jaunin, Hugo and Yeung, Richie and Coecke, Bob and Mansfield, Shane},
  journal={Physical Review A},
  volume={111},
  number={3},
  pages={032429},
  year={2025},
  publisher={APS}
}

@article{preskill2018quantum,
  title={Quantum computing in the NISQ era and beyond},
  author={Preskill, John},
  journal={Quantum},
  volume={2},
  pages={79},
  year={2018},
  publisher={Verein zur F{\"o}rderung des Open Access Publizierens in den Quantenwissenschaften}
}

@article{zhou2020quantum,
  title={Quantum approximate optimization algorithm: Performance, mechanism, and implementation on near-term devices},
  author={Zhou, Leo and Wang, Sheng-Tao and Choi, Soonwon and Pichler, Hannes and Lukin, Mikhail D},
  journal={Physical Review X},
  volume={10},
  number={2},
  pages={021067},
  year={2020},
  publisher={APS}
}

@article{biamonte2017quantum,
  title={Quantum machine learning},
  author={Biamonte, Jacob and Wittek, Peter and Pancotti, Nicola and Rebentrost, Patrick and Wiebe, Nathan and Lloyd, Seth},
  journal={Nature},
  volume={549},
  number={7671},
  pages={195--202},
  year={2017},
  publisher={Nature Publishing Group UK London}
}

@article{bacarreza2025quantum,
  title={Quantum latent distributions in deep generative models},
  author={Bacarreza, Omar and Farnsworth, Thorin and Makarovskiy, Alexander and Wallner, Hugo and Hicks, Tessa and Sempere-Llagostera, Santiago and Price, John and Francis-Jones, Robert JA and Clements, William R},
  journal={arXiv preprint arXiv:2508.19857},
  year={2025}
}

@article{mcclean2016theory,
  title={The theory of variational hybrid quantum-classical algorithms},
  author={McClean, Jarrod R and Romero, Jonathan and Babbush, Ryan and Aspuru-Guzik, Al{\'a}n},
  journal={New Journal of Physics},
  volume={18},
  number={2},
  pages={023023},
  year={2016},
  publisher={IOP Publishing}
}

@article{peruzzo2014variational,
  title={A variational eigenvalue solver on a photonic quantum processor},
  author={Peruzzo, Alberto and McClean, Jarrod and Shadbolt, Peter and Yung, Man-Hong and Zhou, Xiao-Qi and Love, Peter J and Aspuru-Guzik, Al{\'a}n and O’brien, Jeremy L},
  journal={Nature communications},
  volume={5},
  number={1},
  pages={4213},
  year={2014},
  publisher={Nature Publishing Group UK London}
}

@article{gidney2021factor,
  title={How to factor 2048 bit RSA integers in 8 hours using 20 million noisy qubits},
  author={Gidney, Craig and Eker{\aa}, Martin},
  journal={Quantum},
  volume={5},
  pages={433},
  year={2021},
  publisher={Verein zur F{\"o}rderung des Open Access Publizierens in den Quantenwissenschaften}
}

@article{wierichs2022general,
  title={General parameter-shift rules for quantum gradients},
  author={Wierichs, David and Izaac, Josh and Wang, Cody and Lin, Cedric Yen-Yu},
  journal={Quantum},
  volume={6},
  pages={677},
  year={2022},
  publisher={Verein zur F{\"o}rderung des Open Access Publizierens in den Quantenwissenschaften}
}

@article{monbroussou2025photonic,
  title={Photonic quantum convolutional neural networks with adaptive state injection},
  author={Monbroussou, L{\'e}o and Polacchi, Beatrice and Yacoub, Verena and Caruccio, Eugenio and Rodari, Giovanni and Hoch, Francesco and Carvacho, Gonzalo and Spagnolo, Nicol{\`o} and Giordani, Taira and Bossi, Mattia and others},
  journal={Advanced Photonics},
  volume={7},
  number={6},
  pages={066012--066012},
  year={2025},
  publisher={Society of Photo-Optical Instrumentation Engineers}
}

@article{wang2016certification,
  title={Certification of boson sampling devices with coarse-grained measurements},
  author={Wang, Sheng-Tao and Duan, Lu-Ming},
  journal={arXiv preprint arXiv:1601.02627},
  year={2016}
}

@article{nikolopoulos2019cryptographic,
  title={Cryptographic one-way function based on boson sampling: GM Nikolopoulos},
  author={Nikolopoulos, Georgios M},
  journal={Quantum Information Processing},
  volume={18},
  number={8},
  pages={259},
  year={2019},
  publisher={Springer}
}

@article{nikolopoulos2016decision,
  title={Decision and function problems based on boson sampling},
  author={Nikolopoulos, Georgios M and Brougham, Thomas},
  journal={Physical Review A},
  volume={94},
  number={1},
  pages={012315},
  year={2016},
  publisher={APS}
}

@article{singh2025proof,
  title={Proof-of-work consensus by quantum sampling},
  author={Singh, Deepesh and Muraleedharan, Gopikrishnan and Fu, Boxiang and Cheng, Chen-Mou and Roussy Newton, Nicolas and Rohde, Peter P and Brennen, Gavin K},
  journal={Quantum Science and Technology},
  volume={10},
  number={2},
  pages={025020},
  year={2025},
  publisher={IOP Publishing}
}

@article{Chabaud2021quantummachine,
  doi = {10.22331/q-2021-07-05-496},
  url = {https://doi.org/10.22331/q-2021-07-05-496},
  title = {Quantum machine learning with adaptive linear optics},
  author = {Chabaud, Ulysse and Markham, Damian and Sohbi, Adel},
  journal = {{Quantum}},
  issn = {2521-327X},
  publisher = {{Verein zur F{\"{o}}rderung des Open Access Publizierens in den Quantenwissenschaften}},
  volume = {5},
  pages = {496},
  year = {2021}
}

@article{hoch2025quantum,
  title={Quantum machine learning with adaptive boson sampling via post-selection},
  author={Hoch, Francesco and Caruccio, Eugenio and Rodari, Giovanni and Francalanci, Tommaso and Suprano, Alessia and Giordani, Taira and Carvacho, Gonzalo and Spagnolo, Nicol{\`o} and Koudia, Seid and Proietti, Massimiliano and others},
  journal={Nature Communications},
  volume={16},
  number={1},
  pages={902},
  year={2025},
  publisher={Nature Publishing Group UK London}
}

@article{bell2021further,
  title={Further compactifying linear optical unitaries},
  author={Bell, Bryn A and Walmsley, Ian A},
  journal={Apl Photonics},
  volume={6},
  number={7},
  year={2021},
  publisher={AIP Publishing}
}

@article{bouland2014generation,
  title={Generation of universal linear optics by any beam splitter},
  author={Bouland, Adam and Aaronson, Scott},
  journal={Physical Review A},
  volume={89},
  number={6},
  pages={062316},
  year={2014},
  publisher={APS}
}

@article{hangleiter2023computational,
  title={Computational advantage of quantum random sampling},
  author={Hangleiter, Dominik and Eisert, Jens},
  journal={Reviews of Modern Physics},
  volume={95},
  number={3},
  pages={035001},
  year={2023},
  publisher={APS}
}

@article{deshpande2022quantum,
  title={Quantum computational advantage via high-dimensional Gaussian boson sampling},
  author={Deshpande, Abhinav and Mehta, Arthur and Vincent, Trevor and Quesada, Nicol{\'a}s and Hinsche, Marcel and Ioannou, Marios and Madsen, Lars and Lavoie, Jonathan and Qi, Haoyu and Eisert, Jens and others},
  journal={Science advances},
  volume={8},
  number={1},
  pages={eabi7894},
  year={2022},
  publisher={American Association for the Advancement of Science}
}

@article{li2025complexity,
  title={A complexity transition in displaced Gaussian Boson sampling},
  author={Li, Zhenghao and Solomons, Naomi R and Bulmer, Jacob FF and Patel, Raj B and Walmsley, Ian A},
  journal={npj Quantum Information},
  volume={11},
  number={1},
  pages={119},
  year={2025},
  publisher={Nature Publishing Group UK London}
}

@article{xu2021variational,
  title={Variational algorithms for linear algebra},
  author={Xu, Xiaosi and Sun, Jinzhao and Endo, Suguru and Li, Ying and Benjamin, Simon C and Yuan, Xiao},
  journal={Science Bulletin},
  volume={66},
  number={21},
  pages={2181--2188},
  year={2021},
  publisher={Elsevier}
}

@article{endo2020variational,
  title = {Variational Quantum Simulation of General Processes},
  author = {Endo, Suguru and Sun, Jinzhao and Li, Ying and Benjamin, Simon C. and Yuan, Xiao},
  journal = {Phys. Rev. Lett.},
  volume = {125},
  issue = {1},
  pages = {010501},
  numpages = {6},
  year = {2020},
  publisher = {American Physical Society},
  doi = {10.1103/PhysRevLett.125.010501},
  url = {https://link.aps.org/doi/10.1103/PhysRevLett.125.010501}
}

@article{consiglio2024variational,
  title = {Variational Gibbs state preparation on noisy intermediate-scale quantum devices},
  author = {Consiglio, Mirko and Settino, Jacopo and Giordano, Andrea and Mastroianni, Carlo and Plastina, Francesco and Lorenzo, Salvatore and Maniscalco, Sabrina and Goold, John and Apollaro, Tony J. G.},
  journal = {Phys. Rev. A},
  volume = {110},
  issue = {1},
  pages = {012445},
  numpages = {14},
  year = {2024},
  publisher = {American Physical Society},
  doi = {10.1103/PhysRevA.110.012445},
  url = {https://link.aps.org/doi/10.1103/PhysRevA.110.012445}
}

@article{castro2024variational,
  title = {Variational quantum state preparation for quantum-enhanced metrology in noisy systems},
  author = {Zu\~niga Castro, Juan C. and Larson, Jeffrey and Narayanan, Sri Hari Krishna and Colussi, Victor E. and Perlin, Michael A. and Lewis-Swan, Robert J.},
  journal = {Phys. Rev. A},
  volume = {110},
  issue = {5},
  pages = {052615},
  numpages = {14},
  year = {2024},
  publisher = {American Physical Society},
  doi = {10.1103/PhysRevA.110.052615},
  url = {https://link.aps.org/doi/10.1103/PhysRevA.110.052615}
}

@article{cerezo2021cost,
  title={Cost function dependent barren plateaus in shallow parametrized quantum circuits},
  author={Cerezo, Marco and Sone, Akira and Volkoff, Tyler and Cincio, Lukasz and Coles, Patrick J},
  journal={Nature communications},
  volume={12},
  number={1},
  pages={1791},
  year={2021},
  publisher={Nature Publishing Group UK London}
}

@article{marrero2021entanglement,
  title = {Entanglement-Induced Barren Plateaus},
  author = {Ortiz Marrero, Carlos and Kieferov\'a, M\'aria and Wiebe, Nathan},
  journal = {PRX Quantum},
  volume = {2},
  issue = {4},
  pages = {040316},
  numpages = {16},
  year = {2021},
  publisher = {American Physical Society},
  doi = {10.1103/PRXQuantum.2.040316},
  url = {https://link.aps.org/doi/10.1103/PRXQuantum.2.040316}
}

@article{larocca2022diagnosing,
  doi = {10.22331/q-2022-09-29-824},
  url = {https://doi.org/10.22331/q-2022-09-29-824},
  title = {Diagnosing {B}arren {P}lateaus with {T}ools from {Q}uantum {O}ptimal {C}ontrol},
  author = {Larocca, Martin and Czarnik, Piotr and Sharma, Kunal and Muraleedharan, Gopikrishnan and Coles, Patrick J. and Cerezo, M.},
  journal = {{Quantum}},
  issn = {2521-327X},
  publisher = {{Verein zur F{\"{o}}rderung des Open Access Publizierens in den Quantenwissenschaften}},
  volume = {6},
  pages = {824},
  year = {2022}
}

@article{grant2019initialization,
  doi = {10.22331/q-2019-12-09-214},
  url = {https://doi.org/10.22331/q-2019-12-09-214},
  title = {An initialization strategy for addressing barren plateaus in parametrized quantum circuits},
  author = {Grant, Edward and Wossnig, Leonard and Ostaszewski, Mateusz and Benedetti, Marcello},
  journal = {{Quantum}},
  issn = {2521-327X},
  publisher = {{Verein zur F{\"{o}}rderung des Open Access Publizierens in den Quantenwissenschaften}},
  volume = {3},
  pages = {214},
  year = {2019}
}

@article{pesah2021absence,
  title = {Absence of Barren Plateaus in Quantum Convolutional Neural Networks},
  author = {Pesah, Arthur and Cerezo, M. and Wang, Samson and Volkoff, Tyler and Sornborger, Andrew T. and Coles, Patrick J.},
  journal = {Phys. Rev. X},
  volume = {11},
  issue = {4},
  pages = {041011},
  numpages = {26},
  year = {2021},
  publisher = {American Physical Society},
  doi = {10.1103/PhysRevX.11.041011},
  url = {https://link.aps.org/doi/10.1103/PhysRevX.11.041011}
}

@article{skolik2021layerwise,
  title={Layerwise learning for quantum neural networks},
  author={Skolik, Andrea and McClean, Jarrod R and Mohseni, Masoud and Van Der Smagt, Patrick and Leib, Martin},
  journal={Quantum Machine Intelligence},
  volume={3},
  number={1},
  pages={5},
  year={2021},
  publisher={Springer}
}

@article{sack2022avoiding,
  title={Avoiding barren plateaus using classical shadows},
  author={Sack, Stefan H and Medina, Raimel A and Michailidis, Alexios A and Kueng, Richard and Serbyn, Maksym},
  journal={PRX Quantum},
  volume={3},
  number={2},
  pages={020365},
  year={2022},
  publisher={APS}
}

@article{patti2021entanglement,
  title={Entanglement devised barren plateau mitigation},
  author={Patti, Taylor L and Najafi, Khadijeh and Gao, Xun and Yelin, Susanne F},
  journal={Physical Review Research},
  volume={3},
  number={3},
  pages={033090},
  year={2021},
  publisher={APS}
}

@article{bermejo2026quantum,
  title = {Quantum Convolutional Neural Networks are Effectively Classically Simulable},
  author = {Bermejo, Pablo and Braccia, Paolo and Rudolph, Manuel S. and Holmes, Zo\"e and Cincio, Lukasz and Cerezo, M.},
  journal = {PRX Quantum},
  volume = {7},
  issue = {2},
  pages = {020304},
  numpages = {28},
  year = {2026},
  publisher = {American Physical Society},
  doi = {10.1103/8qt9-72ts},
  url = {https://link.aps.org/doi/10.1103/8qt9-72ts}
}

@article{goh2025lie,
  title = {Lie-algebraic classical simulations for quantum computing},
  author = {Goh, Matthew L. and Larocca, Martin and Cincio, Lukasz and Cerezo, M. and Sauvage, Fr\'ed\'eric},
  journal = {Phys. Rev. Res.},
  volume = {7},
  issue = {3},
  pages = {033266},
  numpages = {32},
  year = {2025},
  publisher = {American Physical Society},
  doi = {10.1103/3y65-f5w6},
  url = {https://link.aps.org/doi/10.1103/3y65-f5w6}
}

@article{recio2025train,
  title={Train on classical, deploy on quantum: scaling generative quantum machine learning to a thousand qubits},
  author={Recio-Armengol, Erik and Ahmed, Shahnawaz and Bowles, Joseph},
  journal={arXiv preprint arXiv:2503.02934},
  year={2025}
}

@article{kolarovszki2026generative,
  title={Generative modeling with Gaussian Boson Sampling: classically trainable Bosonic Born Machines},
  author={Kolarovszki, Zolt{\'a}n and Bak{\'o}, Bence and Oszmaniec, Micha{\l} and Oh, Changhun and Zimbor{\'a}s, Zolt{\'a}n},
  journal={arXiv preprint arXiv:2603.11195},
  year={2026}
}

@article{li2026machine,
  title={Machine learning of quantum data using optimal similarity measurements},
  author={Li, Zhenghao and Zhan, Hao and Winston, Shana H and Mer, Ewan and Yin, Zhenghao and Yu, Shang and Alwehaibi, Yazeed K and Machado, Gerard J and Lopena, Dayne Marcus and Zhang, Lijian and others},
  journal={arXiv preprint arXiv:2602.23501},
  year={2026}
}

@misc{monbroussou2026classical,
  title={Classical simulation and model concentration in passive linear optics},
  author={Monbroussou, L. and Thomas, H. and Mhiri, H. and Holmes, Z. and Kashefi, E.},
  year={2026}
}

\appendix

\section{Sample complexity for resolving loss differences}
\label{app:variance_ratio_sufficiency}
 
In this section we justify the claims of Sec.~\ref{subsec:barren_plateaus_and_trainability} that a number of samples scaling with the variance ratio $\mathrm{Var}_{\mathrm{samp}}(\hat{O})/\mathrm{Var}_{\mathrm{circ}}(\hat{O})$ is sufficient to resolve loss differences at different parameter settings to within proportional error. 
We do this through Chebyshev's inequality: setting the precision of the loss estimates at the standard deviation of the loss landscape, a number of samples scaling with the variance ratio suffices to resolve any pair of parameter settings whose losses differ at this scale. 
Under weak conditions on the loss distribution, a non-negligible fraction of randomly drawn pairs are guaranteed to differ on this scale, justifying the use of this ratio as a measure of trainability.
 
\subsection{Setup}
We consider a fixed parameter setting $\bm{\theta}$. A single measurement of the observable $\hat{O}$ on $\rho(\bm{\theta})$ returns an outcome $\hat{O}^{(i)}$ drawn from a distribution with mean $\ell(\bm{\theta}) = \tr[\hat{O}\rho(\bm{\theta})]$ and variance
\begin{align}
    \var[\hat{O}|\bm{\theta}] = \tr[\hat{O}^2 \rho(\bm{\theta})] - \ell(\bm{\theta})^2.
\end{align}
Given $N$ independent measurement outcomes $\hat{O}^{(1)}, \dots, \hat{O}^{(N)}$, we form the standard empirical loss estimate:
\begin{align}
    \bar{\ell}_N(\bm{\theta}) = \frac{1}{N} \sum_{i=1}^N \hat{O}^{(i)},
\end{align}
which satisfies $\ex[\bar{\ell}_N(\bm{\theta})] = \ell(\bm{\theta})$ and, by independence of the $\hat{O}^{(i)}$ measurements,
\begin{align}
\label{eq:app_estimator_variance}
    \var[\bar{\ell}_N(\bm{\theta})] = \frac{\var[\hat{O}|\bm{\theta}]}{N}.
\end{align}

\subsection{Sampling requirements through Chebyshev's inequality}
Directly applying Chebyshev's inequality, we find that for any $\varepsilon_{\textrm{additive}} > 0$,
\begin{align}
\label{eq:app_chebyshev_single}
    \Pr\Big[\big|\bar{\ell}_N(\bm{\theta}) - \ell(\bm{\theta})\big| \ge \varepsilon_{\textrm{additive}} \Big] \le \frac{\var[\hat{O}|\bm{\theta}]}{N\varepsilon_{\textrm{additive}}^2}.
\end{align}
In particular, for any $\delta \in (0,1]$, taking $N \ge \var[\hat{O}|\bm{\theta}]/(\delta\varepsilon_{\textrm{additive}}^2)$ samples suffices to estimate $\ell(\bm{\theta})$ to within additive error $\varepsilon_{\textrm{additive}}$ with probability at least $1-\delta$.
 
In order to resolve a loss-decreasing direction with parameter shift rules, it is necessary to resolve which of two parameter settings, $\bm{\theta}$, and $\bm{\theta'}$, has the lowest loss. Although photonic parameter-shift rules may generally involve several loss evaluations, we focus on this two-point comparison as a useful case study. We expect the resulting conclusions to extend naturally to more general parameter-shift rules.

We'll denote the difference in the losses by $D(\bm\theta, \bm{\theta'}) = \ell(\bm{\theta}) - \ell(\bm{\theta'})$, and the estimated difference in losses by $\bar{D}_N(\bm\theta, \bm{\theta'}) = \bar{\ell}_N(\bm{\theta}) - \bar{\ell}_N(\bm{\theta}')$. We have $\ex[\bar{D}_N] = D$ and, by independence of the observable samples,
\begin{align}
\label{eq:app_diff_variance}
    \var[\bar{D}_N] = \frac{\var[\hat{O}|\bm{\theta}] + \var[\hat{O}|\bm{\theta'}]}{N}.
\end{align}
We take the magnitude of the loss difference, $|D(\bm\theta, \bm{\theta'})|$, to define the relevant error scale. Estimating $\bar{D}_N$ to within a relative error $\varepsilon_{\textrm{relative}}|D|$ is sufficient to determine the loss decreasing direction and, in the context of parameter shift rules, estimate the gradient with a relative accuracy scaling as $\varepsilon_{\textrm{relative}}$.
Plugging this error into the above inequality, we know that in order to achieve this with constant probability $1-\delta$, we see the number of shots is bounded by
\begin{align}
\label{eq:chebyshev_sample_complexity}
    N \ge \frac{\var[\hat{O}|\bm{\theta}] + \var[\hat{O}|\bm{\theta}']}{\delta \, \varepsilon^2 \, \big(\ell(\bm{\theta})-\ell(\bm{\theta}')\big)^2}
\end{align}
 
\subsection{Sufficiency of the variance ratio}
 
Equation~\eqref{eq:chebyshev_sample_complexity} bounds the samples needed to resolve the loss at a pair of parameter settings in terms of the fixed-parameter variances at the two settings and their loss difference $D$. 

When $\bm{\theta},\bm{\theta}'$ are drawn independently, averaging over both parameters recovers the sample variance and circuit variance discussed in Sec.~2.
Explicitly, the numerator averages to $\ex_{\bm{\theta},\bm{\theta}'}\big[\var[\hat{O}|\bm{\theta}] + \var[\hat{O}|\bm{\theta}']\big] = 2\,\mathrm{Var}_{\mathrm{samp}}(\hat{O})$. Similarly, because $\mathbb{E}_{\bm{\theta},\bm{\theta'}}[D(\bm{\theta}, \bm{\theta'})] = 0$, the expected value of the squared loss difference in the denominator is exactly its variance over the landscape,
\begin{align}
\begin{split}
\label{eq:app_expected_squared_difference}
    \ex_{\bm{\theta},\bm{\theta'}} \big[D^2(\bm\theta, \bm{\theta'})\big] &= \var_{\bm{\theta},\bm{\theta'}}[D(\bm{\theta}, \bm{\theta'})] \\
    &= \var_{\bm{\theta}}[\ell(\bm{\theta})] + \var_{\bm{\theta}'}[\ell(\bm{\theta}')] \\
    &= 2\,\mathrm{Var}_{\mathrm{circ}}(\hat{O}),
\end{split}
\end{align}
using independence of the parameter choices. At the level of these averages, the pairwise bound \eqref{eq:chebyshev_sample_complexity} therefore becomes exactly the variance ratio, $N \ge \mathrm{Var}_{\mathrm{samp}}(\hat{O})\big/\big(\delta\,\varepsilon^2\,\mathrm{Var}_{\mathrm{circ}}(\hat{O})\big)$. 

In particular, Eq.~\eqref{eq:app_expected_squared_difference} identifies
$\sqrt{\var_{\mathrm{circ}}(\hat{O})}$ as the characteristic scale of loss differences across the circuit ensemble. It is therefore natural to define the target additive error as a fixed fraction of this scale, $\varepsilon_{\textrm{additive}} = \varepsilon \sqrt{\var_{\textrm{circ}}(\hat{O})}$ so $\varepsilon$ can be interpreted as the relative error. Estimating each loss to this accuracy corresponds to resolving variations that are proportional to typical differences encountered across the loss landscape; since $\sqrt{\var_{\textrm{circ}}(\hat{O})}$ quantifies the characteristic magnitude of these fluctuations, this choice of error scale should allow changes induced by varying the circuit parameters to be distinguished. Substituting this error scale into Chebyshev's inequality then naturally gives rise to the ratio between the expected sample variance and the circuit variance.

We will now formalise this statement. 
\begin{proposition}
\label{thm:variance_ratio_sufficiency}
    Let $\varepsilon > 0$, $\delta \in (0,1]$ and $t > 0$. Let $\bm{\theta}, \bm{\theta'}$ be drawn independently from a parameter distribution. With probability at least $1 - 2/t$, taking
    \begin{align}
    \label{eq:variance_ratio_sample_complexity}
        N \ge \frac{t}{\delta \, \varepsilon^2} \, \frac{\mathrm{Var}_{\mathrm{samp}}(\hat{O})}{\mathrm{Var}_{\mathrm{circ}}(\hat{O})}
    \end{align}
    samples at each setting yields loss estimates satisfying
    \begin{align}
        \big|\bar{\ell}_N(\bm{\theta}) - \ell(\bm{\theta})\big| \; , \; \big|\bar{\ell}_N(\bm{\theta}') - \ell(\bm{\theta}')\big| \le \varepsilon\sqrt{\mathrm{Var}_{\mathrm{circ}}(\hat{O})},
    \end{align}
    simultaneously with probability at least $1 - 2\delta$. In particular, whenever $|\ell(\bm{\theta}) - \ell(\bm{\theta}')| > 2\varepsilon\sqrt{\mathrm{Var}_{\mathrm{circ}}(\hat{O})}$, the loss-minimising direction is correctly identified.
\end{proposition}
\begin{proof}
    The sample variance is the average of the single-shot variances, $\mathrm{Var}_{\mathrm{samp}}(\hat{O}) = \ex_{\bm{\theta}}\big[\var[\hat{O}|\bm{\theta}]\big]$, so by Markov's inequality $\mathbb{P}_{\bm{\theta}}\big[\var[\hat{O}|\bm{\theta}] > t\,\mathrm{Var}_{\mathrm{samp}}(\hat{O})\big] < 1/t$, giving $\var[\hat{O}|\bm{\theta}], \var[\hat{O}|\bm{\theta}'] \le t\,\mathrm{Var}_{\mathrm{samp}}(\hat{O})$ with probability at least $1 - 2/t$. On this event, Chebyshev's inequality \eqref{eq:app_chebyshev_single} with error $\varepsilon_{\mathrm{additive}} = \varepsilon\sqrt{\mathrm{Var}_{\mathrm{circ}}(\hat{O})}$ shows that each estimate achieves this error with probability at least $1-\delta$ for $N$ satisfying \eqref{eq:variance_ratio_sample_complexity}. Both estimates achieve these errors simultaneously with probability at least $1-2\delta$. Finally, if $|D| > 2\varepsilon\sqrt{\mathrm{Var}_{\mathrm{circ}}(\hat{O})}$, the triangle inequality gives $\mathrm{sign}(\bar{D}_N) = \mathrm{sign}(D)$.
\end{proof}
 
Taking $t$, $\varepsilon$ and $\delta$ to be constants, a number of samples $N = \mathrm{O}\big(\mathrm{Var}_{\mathrm{samp}}(\hat{O})/\mathrm{Var}_{\mathrm{circ}}(\hat{O})\big)$ is sufficient to resolve losses to within their expected fluctuation, and to correctly order any pair of settings whose losses differ by more than this resolution.
 
Proposition \ref{thm:variance_ratio_sufficiency} makes explicit how trainability fails at the two ends of the variance ratio. If the sample variance is exponentially large in system size, it takes exponentially many samples to build estimates to fixed precision. If instead the circuit variance is exponentially small, the loss differences need to be resolved to exponentially small error $\sqrt{\mathrm{Var}_{\mathrm{circ}}(\hat{O})}$, and by \eqref{eq:app_chebyshev_single} reaching this precision requires exponentially many samples for constant or larger shot noise. In either case the sample complexity \eqref{eq:variance_ratio_sample_complexity} is exponential unless the other variance compensates. Conversely, if the sample variance grows at most polynomially with system size and the circuit variance decays at most polynomially, then the variance ratio is at most polynomial, and polynomially many samples are certainly sufficient to resolve loss differences at the landscape scale.
 
\subsection{Typical loss differences}
\label{subsec:app_typical_loss_differences}
 
With proposition \ref{thm:variance_ratio_sufficiency}, we show that with a number of samples scaling with the variance ratio, we can resolve the difference between loss at two different parameter settings to within the average loss fluctuations expected in the landscape. However, we don't consider how regularly we can expect the difference of loss to be of this magnitude. We now discuss general assumptions on the loss landscape that would be sufficient to guarantee that values of $|D|$ scale approximately with $\sqrt{\mathrm{Var}_{\mathrm{circ}}(\hat{O})}$ with high probability.

\begin{enumerate}
    \item[(i)] \emph{Anticoncentration of the landscape.} There exist $\alpha, \beta > 0$, scaling at worst inverse-polynomially in system size, such that
    \begin{align}
        \mathbb{P}\Big[D^2 \ge \beta \cdot 2\,\mathrm{Var}_{\mathrm{circ}}(\hat{O})\Big] \ge \alpha.
    \end{align}
    In words, this conjectures that for an inverse-polynomially large fraction of randomly drawn pairs, the loss difference is at least as large as $\sqrt{\mathrm{Var}_{\mathrm{circ}}(\hat{O})}$ up to a polynomial factor. Under this condition, proposition \ref{thm:variance_ratio_sufficiency} with any $\varepsilon < \sqrt{\beta/2}$, using $N = \mathrm{O}\big(\beta^{-1}\,\mathrm{Var}_{\mathrm{samp}}(\hat{O})/\mathrm{Var}_{\mathrm{circ}}(\hat{O})\big)$ samples is sufficient for loss resolution. As such, having a polynomially scaling variance ratio is still sufficient to claim efficient trainability, by our definition. We note that this condition is analogous to the anticoncentration properties of boson sampling output probabilities \cite{mhiri2026bosonsamplingdiluteregime}. Indeed this condition holds for $\hat{O} = \ketbra{\bm{s}}{\bm{s}}$ estimating output probabilities through the anticoncentration conjecture.
    \item[(ii)] \emph{Gaussian loss differences.} We assume $D \sim \mathcal{N}\big(0, 2\,\mathrm{Var}_{\mathrm{circ}}(\hat{O})\big)$. This gives condition (i) with sharp, system-size-independent constants:
    \begin{align}
        \mathbb{P}\Big[|D| \ge \sqrt{\mathrm{Var}_{\mathrm{circ}}(\hat{O})}\Big] = 2\,\Phi\big(-1/\sqrt{2}\big) \approx 0.48,
    \end{align}
    where $\Phi$ is the standard normal cumulative distribution function. We can conclude nearly half of all randomly drawn pairs are resolved by proposition~\ref{thm:variance_ratio_sufficiency} with any $\varepsilon < 1/2$ and $N = \mathrm{O}\big(\mathrm{Var}_{\mathrm{samp}}(\hat{O})/\mathrm{Var}_{\mathrm{circ}}(\hat{O})\big)$ samples, with no polynomial overheads.
\end{enumerate}
Under either condition, combined with proposition~\ref{thm:variance_ratio_sufficiency}, a large fraction of parameter pairs are resolvable with a number of samples scaling, up to polynomial factors, as the variance ratio $\mathrm{Var}_{\mathrm{samp}}(\hat{O})/\mathrm{Var}_{\mathrm{circ}}(\hat{O})$, which is then sufficient as a measure of trainability. Together with the equivalence between gradient and loss concentration established in Appendix~\ref{appendix:bp_equivalence}, this extends from resolving loss differences to resolving gradients themselves.

\section{Equivalence between barren plateaus and loss concentration}
\label{appendix:bp_equivalence}
In Ref.~\cite{arrasmith2022equivalence}, the authors prove an equivalence between exponential concentration of the loss function and the presence of barren plateaus across all parameters. This implies that loss concentration is a sufficient condition for diagnosing barren plateaus, but the converse is weaker; a non-concentrated loss does not guarantee that every parameter remains trainable, but only that at least one parameter has a non-vanishing gradient. To determine the absence of barren plateaus across all parameters requires additionally showing that the parameters variances are in the same trainability class, for example by showing that all gradient variances are comparable in size.

In this section, we first adapt the results of Ref.~\cite{arrasmith2022equivalence} to naturally parametrised linear optical photonic architectures. A key difference that emerges is that average gradients can be non-zero, and hence concentrated gradient variance alone is insufficient to determine whether a parameter is trainable. We therefore look at the second moment of the gradients as the analogous condition for whether a parameter is trainable. With this, we show an equivalence between loss concentration and all parameter admitting barren plateau type behaviour.

We then consider a particular class of parametrised linear optical networks and show that the ensemble averaged gradient second moments are equal for all parameters. For this class of circuits, we can expect that anti-concentration of the loss function implies the absence of barren plateaus across all parameters for this class of circuits.

\subsection{Setup}
Consider an $M$-mode parametrised linear optical network consisting of $L$ parameters with associated generator $\{\hat{B}_j\}$. This can be written as
\begin{align}
\begin{split}
    U(\bm{\theta}) =  \prod_{j=1}^L U_j(\theta_j), \quad \quad U_j(\theta_j) = e^{-i\theta_j \hat{B}_j} \hat{W}_j
\end{split}
\end{align}
where $\hat{W}_j$ is a fixed unitary. For example, the generators could be beam splitter interaction and fixed unitaries fixed phase shifters. The gradient of this parametrised circuit with respect to $\theta_k$ reads
\begin{align}
\begin{split}
\label{eq:photonic_PQC_derivative_with_operator}
    G_k(\bm{\theta}) = U_{+k}(\bm{\theta}) \big(-i \hat{B}_k \big) U_{-k}(\bm{\theta})
\end{split}
\end{align}
where we defined
\begin{align}
\begin{split}
    U_{-k}(\bm{\theta}) = \prod_{j=1}^k U_j(\bm{\theta}), \quad \quad U_{+k}(\bm{\theta}) = \prod_{j=k+1}^L U_j(\bm{\theta})
\end{split}
\end{align}
satisfying $U_{+j} U_{-j} = U$. 

Now we consider a loss function $\ell(\bm{\theta})$ given by the expectation value of some observable
\begin{align}
\begin{split}
    \ell(\bm{\theta}) &= \tr[U(\theta) \rho U^\dag(\bm{\theta}) \, \hat{O}]
\end{split}
\end{align}
The squared magnitude of loss fluctuation across the parameter landscape is captured by the variance
\begin{align}
\begin{split}
    \var_{\bm\theta}[\ell(\bm\theta)] = \ex_{\bm\theta} [\ell(\bm{\theta})^2] - \ex_{\bm\theta}[\ell(\bm\theta)]^2
\end{split}
\end{align}
A loss function is said to be concentrated if the size of fluctuations is exponentially small with respect to problem size $N$
\begin{align}
\begin{split}
    \var_{\bm\theta}[\ell(\bm\theta)] \sim \mathrm{O}(e^{-N}) \implies \ell \textrm{ is exponentially concentrated.}
\end{split}
\end{align}
In qubit-based architectures, the problem size is almost always determined by the qubit count. In our photonic case, there is more freedom how problem size is defined -- the number of modes, number of input photons, or support/weight of the observable can all meaningfully change the problem size. For now, we keep the problem size as labelled by $N$.

The derivative of the loss function with respect to a parameter $\theta_k$, denoted $G_k(\bm{\theta})$, can be written as
\begin{align}
\begin{split}
    G_k(\bm{\theta}) &= \partial_{\theta_k} \ell(\bm{\theta}) \\
    &= i \tr\big[ U_{-k}(\bm{\theta}) \rho U_{-k}^\dag(\bm{\theta}) \,\,[\hat{B}_k, U_{+k}^\dag(\bm{\theta}) \hat{O} U_{+k}(\bm{\theta})] \big]
\end{split}
\end{align}
where $[\bullet, \star] = \bullet\star-\star\bullet$ is the commutator. In qubit devices, the average gradient is zero \cite{mcclean2018barren, arrasmith2022equivalence, larocca2025barren}. In photonic devices, if averages are taken with respect to the Haar measure then this is not the case because $\mu_{\textrm{Haar}}(\bm{\theta}) \ne \textrm{Uniform}([0, 2\pi])$; instead there can be a specific generator-observable dependence,
\begin{align}
\begin{split}
    \ex_{\bm{\theta}} [G_k(\bm{\theta})] = \ex_{\bm{\theta}} \left[ i \tr[U_{-k}(\bm{\theta}) \rho U_{-k}^\dag(\bm{\theta}) \, \big[ \hat{B}_k, \, U^\dag_{+k}(\bm{\theta}) \hat{O} U_{+k} (\bm\theta)\big]] \right].
\end{split}
\end{align}

Barren plateaus are said to occur if the gradient variances decrease exponentially with respect to the problem size $N$ \cite{mcclean2018barren, arrasmith2022equivalence, larocca2025barren}
\begin{align}
\begin{split}
    \var_{\bm\theta}[ G_k(\bm{\theta})] \sim \mathrm{O}(e^{-N}) \implies \ell \textrm{ has a $\theta_k$ barren plateau.}
\end{split}
\end{align}
A barren plateau so the $k^{\textrm{th}}$ parameter with average gradient $\ex_{\bm\theta}[G_k(\bm\theta)]=0$ indicates that exponential precision is required in order to distinguish the direction (positive orsnegative) in which the loss function is decreasing. Therefore, optimising the loss function along the $\theta_k$ manifold typically requires exponentially many shots and is not efficient.
In contrast, if the average gradient is non-zero then estimating the sign of the gradient is far easier; if one can efficiently estimate if $G_k$ is positive or negative then one can update parameters in the correct direction efficiently. This suggests that the variance-based barren plateau definition is not suitable for cases when $\ex_{\bm\theta}[G_k]$ can be non-zero. For example, we could consider gradients being exponentially concentrated around some polynomially scaling average $|\ex_{\bm\theta}[G_k]| > \mathrm{O}(e^{-N})$; while the variance of $G_k$ indicates the existence of barren plateaus, the parameter $\theta_k$ could actually be efficiently updated.

In the case with potentially finite average gradient, it is more suitable to define barren plateaus with respect to the expected gradient second moments; we say the $k^{\textrm{th}}$ parameter admits a barren plateau if
\begin{align}
\begin{split}
    \ex_{\bm\theta} [ G_k^2(\bm\theta)] &= \var_{\bm\theta}[G_k(\bm\theta)] + \ex_{\bm\theta} [G_k(\bm\theta)]^2 \\
    &\sim \mathrm{O}(e^{-N}) \implies \ell \textrm{ has a $\theta_k$ barren plateau.}
\end{split}
\end{align}
This now includes a term accounting for the squared magnitude of the average gradient. If the average gradient is zero, we recover the standard definition. Similarly, if both the average and variance terms are exponentially small then this correctly indicates the parameter is not efficiently trainable. However, if the variance is exponentially small but the average is polynomially scaling, then the gradient second moment is also polynomially scaling and our definition indicates that the parameter is trainable.

For this new definition of barren plateau, we go on to prove a photonic analogue of the results from Ref.~\cite{arrasmith2022equivalence}; that there is an equivalence between the loss function being exponentially concentrated and all parameters admitting barren plateaus. First, we show that loss concentration implies barren plateaus.

\subsection{Loss concentration implies barren plateaus}
In Ref.~\cite{facelli2024exactgradients} and Ref.~\cite{pappalardo2025photonic}, the authors prove a parameter shift rule for linear optical networks with Fock input states. Namely, for a loss function $\ell(\bm{\theta}) = \tr[U(\bm{\theta}) \rho U^\dag(\bm{\theta}) \, \hat{O}]$ and input Fock state containing a definite $n$ photons, the gradients are given by~\cite{facelli2024exactgradients}
\begin{align}
\begin{split}
    \partial_{\theta_k} \ell(\bm{\theta}) &= \sum_{j = 0}^{2n} \frac{(-1)^{j+1}}{4n\sin^2(\kappa_j/2)} \ell(\bm{\theta} + \kappa_j\bm{e}_k)
\end{split}
\end{align}
for $\kappa_j = \frac{(2j-1)\pi}{2n}$. For observables given by photon number monomials (PNMs) of degree $|\bm{\alpha}|$,
the number of loss evaluations can be reduced to $\min\{|\bm{\alpha}|, n\}$ \cite{facelli2024exactgradients}. This implies that for observables with a degree scaling independently of system size, the number of loss evaluations is constant and given by $|\bm{\alpha}|$.

Loss concentration implying barren plateaus is a direct consequence of the parameter shift rule, as proven in Ref.~\cite{arrasmith2022equivalence}. The standard qubit-based derivation has a parameter shift rule containing two terms. In our photonic setup, the gradient can be written as a sum over $\min\{|\bm{\alpha}|, n\}$ loss evaluations. Using the same argument as in Ref.~\cite{arrasmith2022equivalence}, it is clear that a concentrated cost function implies all parameters admit barren plateaus, because the sum of polynomially many exponentially concentrated loss functions must itself be exponentially concentrated. 
This proves the following statement;
\begin{proposition}
    If a loss function $\ell(\bm{\theta})$ is exponentially concentrated, then this implies that all associated gradients $G_k = \partial_{\theta_k}\ell(\bm\theta)$ have an exponentially concentrated second moment and hence admit barren plateaus.
\end{proposition}

\subsection{Barren plateaus imply loss concentration}

In Ref.~\cite{arrasmith2022equivalence}, the derivation that barren plateaus imply cost concentration applied if (i) the loss function was periodic on each parameter and (ii) the parameters were sampled from a uniform distribution on their domain. In the photonic case, we are interested in sampling from Haar random linear optical networks for which both (i) and (ii) are not simultaneously true. For example, consider using the method of direct-dialling Haar random linear optical networks in Ref.~\cite{russell2017direct} -- the beam splitter parameters are not uniformly sampled and hence we need a different method of showing barren plateaus imply loss concentration.

Each beam splitter is implemented as a Mach–Zehnder interferometer with both an internal and external phase, and the circuit ends with a layer of $M$ output phases, giving $L=M^2$ trainable parameters. Russel \emph{et al}~\cite{russell2017direct} construct a product measure $\nu = \nu_1 \times \dots \times \nu_L$ on $\bm\theta$ such that $U(\bm\theta) \sim \textrm{Haar}(\mathrm{U}(M))$ whenever $\bm\theta \sim \nu$. The external output phases are uniform on $[0, 2\pi)$ while the internal phases at layer $\in\{1, \dots, M-1\}$ have probability density
\begin{align}
\begin{split}
\label{eq:dialling_density}
    p_k(\theta) = k\sin^{2k-1}(\theta/2) \cos(\theta/2)
\end{split}
\end{align}
for $\theta \in [0, \pi]$. We call $\nu$ the Haar-dialling measure; it satisfies $\var_{\textrm{circ}}(\hat{O})= \var_{\bm\theta}[\ell(\bm\theta)]$ being exactly the circuit variance from Sec.~\ref{sec:trainability}.

We now prove the following statement, that the circuit variance is upper bounded by the largest gradient second moment.
\begin{theorem}
\label{thm:BP_implies_loss_concentration}
    For $\bm\theta$ drawn from the Haar-dialling measure $\nu$ of Ref.~\cite{russell2017direct}, the circuit variance of the loss function is upper bounded by
    \begin{align}
    \begin{split}
    \label{eq:dialling_main_bound}
        \var_{\mathrm{circ}}(\hat{O}) &\le \sum_{j=1}^{L} \kappa_j \ex_{\bm\theta \sim \nu}[G_j(\bm\theta)^2] \\
        &\le \frac{3L \pi^2}{2} \max_j \ex_{\bm\theta \sim \nu}[G_j(\bm\theta)^2]
    \end{split}
    \end{align}
    where $\kappa_j = 1$ for external and output phases, and $\kappa_j \le 3\pi^2/2$ for internal phases.
\end{theorem}
\begin{proof}
    The random variables $\theta_j$ are independent under this measure $\nu$, so the Efron-Stein inequality~\cite{efrom1981jackknife} gives
    \begin{align}
    \begin{split}
    \label{eq:efron_stein_step}
        \var_{\bm\theta\sim\nu}[\ell(\bm\theta)] \le \sum_{j=1}^L \ex_{\bm\theta\sim\nu} \big[ \var_{\theta_j}[\ell(\bm\theta) | \bm\theta_{\not j} ] \big]
    \end{split}
    \end{align}
    where $\var_{\theta_j\sim\nu_j} [\bullet | \bm\theta_{\not j}]$ is the variance over $\theta_j\sim\nu_j$ with all other coordinates held fixed. 
    Each conditional variance can then be bounded by a one-dimensional Poincar\'e inequality of the form $\var_{\theta_j\sim\nu_j}[f(\theta_j)] \le \kappa_j \ex_{\theta_j\sim\nu_j}[(\partial_{\theta_j}f(\theta_j))^2]$.

    For the uniformly-distributed external and output phases, Wirtinger's inequality gives $\kappa_j = 1$. For the internal phases, the probability density~\eqref{eq:dialling_density} is log-concave on $(0,\pi)$, since
    \begin{equation}
        \frac{d^2}{d\theta^2}\log p_k(\theta) = -\frac{2k-1}{4}\csc^2(\theta/2) - \frac14\sec^2(\theta/2) < 0,
    \end{equation}
    Thus, Bobkov's Poincar\'e inequality for log-concave measures on $\mathbb{R}$~\cite{bobkov1999isoperimetric}
    gives $\kappa_j \le 12 \, \var_{\nu_j}[\theta_j]$; since $\theta_j \in [0,\pi]$, Popoviciu's inequality bounds $\var_{\nu_j}[\theta_j] \le \pi^2/4$,
    so $\kappa_j \le 3\pi^2/2$.

    Substituting these Poincar\'e inequalities into \eqref{eq:efron_stein_step} and taking expectation over the remaining coordinates gives the first inequality in \eqref{eq:dialling_main_bound}, where we've bounded every $\kappa_j$ by $3\pi^2/2$.
\end{proof}

Theorem~\ref{thm:BP_implies_loss_concentration} states that if all parameters admit barren plateaus (as indicated by an exponentially small $\ex_{\bm\theta}[G_k]$), then the loss function is also exponentially concentrated. This therefore proves that barren plateaus imply loss concentration.

Together with the last section, these results demonstrate an equivalence between exponential concentration of the loss function and the presence of barren plateaus across all circuit parameters. We next identify a class of parametrised circuits for which the ensemble-average gradient variances are equal for every parameter. The variance of the loss function can therefore be expected to provide a necessary and sufficient condition for understanding whether barren plateaus are present or not across the entire parameter landscape.

\subsection{Equality of expected gradient second moments}
A polynomially scaling loss function variance alone is insufficient to rule out barren plateaus -- it merely indicates that \emph{at least one parameter} is trainable, but perhaps a further $L-1$ parameters are untrainable. In order to rule out barren plateaus by considering only the loss variance, we will now consider a family of parametrised circuits for which the expected value of the gradient second moments, $\ex[G_k^2]$, are all equal. In this case, the above equivalence between barren plateaus and concentrated loss function implies that a polynomially scaling loss variance is indeed sufficient to determine that all parameters are trainable.

However, for a generic parametrised circuit, we can never expect all gradient variances to be in the same trainability class because some some parameters may inevitably have zero impact on the loss function. Take for example an initial state $\ket{1,0\dots, 0}$; clearly any beam splitter outside the lightcone of the first mode will have zero impact on the loss function and hence the associated parameter gradients are exactly zero. A similar argument applies to local observables.

In this section, we consider sandwiching the parametrised network $U(\bm{\theta}) = \prod_{j=1}^L U_j(\theta_j)$ between two other linear optical networks $V_1$ and $V_2$ drawn from the Haar measure of $\textrm{U}(M)$ such that the total dynamics read 
\begin{align}
\begin{split}
    U_{\textrm{total}}(\bm{\theta}; V_1, V_2) = V_2 U(\bm{\theta}) V_1.
\end{split}
\end{align}
These networks distribute the state and observable across all $M$ modes and can be expected to prevent gradients from being identically zero. The loss function and its derivatives can be written as
\begin{align}
\begin{split}
    \ell(\bm{\theta}; V_1, V_2) = \tr[ U(\bm{\theta}) V_1 \rho V_1^\dag U^\dag(\bm{\theta}) \, V_2^\dag \hat{O} V_2 ]
\end{split}
\end{align}
and
\begin{align}
\begin{split}
    G_k(\bm{\theta}; V_1, V_2) &= \partial_{\theta_k} \ell(\bm{\theta}; V_1, V_2) \\
    &= i\tr[ U_{-k}(\bm{\theta}) V_1 \rho V_1^\dag U_{-k}(\bm{\theta}) \, \,\big[\hat{B}_k, \, U_{+k}^\dag(\bm{\theta}) V_2^\dag \hat{O} V_2 U_{+k}(\bm{\theta}) \big]] 
\end{split}
\end{align}
respectively. As mentioned, the average gradient values can be non-zero and so it is the second moment of the gradient, 
\begin{align}
\begin{split}
    \mu_2(G_k) &= \ex_{\bm{\theta}}[G_k(\bm{\theta})^2] \\
    &= \var_{\bm\theta}[G_k(\bm{\theta})] + \ex_{\bm\theta}[G_k(\bm{\theta})]^2
\end{split}
\end{align}
that can be used to identify barren plateau behaviour -- if this quantity is not exponentially small, then one does not require exponential precision in order to update the parameter $\theta_k$ in the correct direction. We now state and prove the following proposition.
\begin{proposition}
    Let $G_j(\bm{\theta}; V_1,V_2)$ and $G_k(\bm{\theta};V_1,V_2)$ denote the gradients associated with parameters $\theta_j$ and $\theta_k$ respectively, which have generators $\hat{B}_j$ and $\hat{B}_k$.
    If the two generators are equivalent up to a permutation of the mode labels, then their expected gradient second moments are equal. More precisely, if
    \begin{equation}
        \hat{B}_j = \hat{\pi}\circ\hat{B}_k
    \end{equation}
    for some operator $\hat{\pi}$ that permutes the mode labels, then
    \begin{equation}
        \ex_{V_1,V_2}\left[\, \mu_2(G_j(V_1,V_2)) \,\right]
        =
        \ex_{V_1,V_2}\left[\, \mu_2(G_k(V_1,V_2)) \,\right].
    \end{equation}
    Consequently, all parameters whose generators belong to the same equivalence class under mode permutations have the same trainability scaling. In particular, when averaging over $V_1$ and $V_2$, all beam-splitter parameters either admit barren-plateau-type behaviour or do not.
\end{proposition}
\begin{proof}
    We can directly compute the ensemble averaged value of the second moment over the Haar random networks $V_1$ and $V_2$.
    The second moment of the gradient $G_j$ over $\bm{\theta}$ can be written as
    \begin{align}
    \begin{split}
        \mu_2(G_k(V_1, V_2)) &= -\ex_{\bm{\theta}} \tr[\big(U_{-k}V_1 \rho V_1^\dag U_{-k}^\dag\big)^{\otimes 2} \, \big[\hat{B}_k, \, U_{+k}^\dag V_2^\dag \hat{O} V_2 U_{+k}\big]^{\otimes 2}  ]
    \end{split}
    \end{align}
    where we dropped the $\bm{\theta}$ argument on $U_{\pm k}$ for ease of reading. Taking the average over $V_1$ and $V_2$ gives
    \begin{align}
    \begin{split}
        \ex_{V_1,V_2} \left[\,\mu_2(G_k(V_1,V_2))\, \right] &= -\ex_{V_1, V_2, \bm{\theta}} \tr[ \big(U_{-k} V_1 \rho V_1^\dag U_{-k}^\dag \big)^{\otimes 2} \, \big[\hat{B}_k, \, U_{+k}^\dag V_2^\dag \hat{O} V_2 U_{+k}  \big]^{\otimes 2}]
    \end{split}
    \end{align}
    Using the left/right invariance of the Haar measure, we make the transformation $V_1 \rightarrow U_{-k}^\dag V_1$ and $V_2 \rightarrow V_2 U^\dag_{+k}$ to remove all $U$ dependence.
    \begin{align}
    \begin{split}
        \ex_{V_1,V_2}[\mu_2(G_k(V_1,V_2))] &= -\tr[ \ex_{V_1} \left[ \big(V_1 \rho V_1^\dag\big)^{\otimes 2} \right] \ex_{V_2} \left[ \big[ \hat{B}_k, \, V_2 \hat{O} V_2^\dag \big]^{\otimes 2} \right] ]
    \end{split}
    \end{align}
    Now we notice that, due to the symmetry of Haar averaged operators, this quantity depends only on what class of generators $\hat{B}_k$ is in; for our linear optical network, this is either a beam splitter or phase shifter. Therefore, the ensemble averaged gradient second moments for all parameters that share a generator are equal, proving the above proposition.    
\end{proof}

Finally, consider parameterised circuits containing trainable beam splitters and randomly chosen, fixed phase shifters. Let $\bm\theta$ and $\bm\varphi$ denote the beam-splitter and phase-shifter parameters, respectively, and suppose that only $\bm\theta$ is varied during training. Since the Haar measure induced by the parameterisation of Ref.~\cite{russell2017direct} factorises over $\bm\theta$ and $\bm\varphi$, the two sets of parameters may be averaged independently. We may therefore treat the average over the fixed phase shifters as an ensemble average and study the expected beam-splitter gradient second moments. This avoids the possibility that a non-vanishing loss variance is supported entirely by trainable phase-shifter directions while all beam-splitter gradients remain exponentially suppressed. We find that, after averaging over the random phase shifters and the Haar-random unitaries $V_1$ and $V_2$, the beam-splitter gradient second moments are equal for all trainable parameters -- for $G_j(\bm\theta; \bm\varphi, V_1, V_2) = \partial_{\theta_j} \ell(\bm\theta; \bm\varphi, V_1, V_2)$, find
\begin{align}
\begin{split}
    \ex_{V_1, V_2, \bm\varphi}\big[\ex_{\bm\theta}[G_j^2(\bm\theta; \bm\varphi,V_1, V_2)] \big] = \ex_{V_1, V_2, \bm\varphi}\big[\ex_{\bm\theta}[G_k^2(\bm\theta; \bm\varphi,V_1, V_2)] \big]
\end{split}
\end{align}
for all $j$ and $k$.
Thus, averaged across this ensemble $(\bm\varphi, V_1, V_2)$, all loss function parameters belong to the same trainability class. In combination with the preceding variance bound, a polynomially large loss variance implies that the common gradient second moment cannot be exponentially small, provided that the number of trainable parameters grows at most polynomially. Consequently, we can expect that barren plateaus can be ruled out without explicitly evaluating the individual gradients.

\section{Average linear optical network output}
\label{appendix:average_linear_optical_network}
In this section, we compute the average output state from a linear optical network using Weingarten calculus. Alternative proofs based around Schur's lemma also exist, such as in Ref.~\cite{mhiri2026bosonsamplingdiluteregime}. 
Consider an arbitrary $M$-mode state $\rho$ with photon number representation
\begin{align}
\begin{split}
    \rho = \sum_{\bm{n}, \bm{m} \in \mathbb{N}_0^M} c_\rho(\bm{n}, \bm{m}) \ketbra{\bm{n}}{\bm{m}}.
\end{split}
\end{align}
We evolve it via a passive linear optical network $U$. In Theorem \ref{thm:average_PLON_output_state}, we compute the average output state over the Haar measure of $\mathrm{U}(M)$.

\begin{theorem}
\label{thm:average_PLON_output_state}
    For an $M$-mode linear optical network $U$, the Haar averaged output state of $U \rho U^\dag$ is given by
    \begin{align}
    \begin{split}
        \ex_{U} \, [U \rho U^\dag] &= \sum_{n = 0}^\infty p_\rho(n) \pi_{M}(n)
    \end{split}
    \end{align}
    where $\pi_M(n)$ is the maximally mixed state consisting of $n$ photons in $M$ modes, i.e., the normalised projector onto the $n$ photon subspace,
    \begin{align}
    \begin{split}
        \pi_M(n) =  \frac{1}{{n+M-1 \choose n}}  \Pi_M(n), \quad \quad \Pi_M(n)=
        \sum_{\bm{m} \in \Phi_M^n}\ketbra{\bm{m}},
    \end{split}
    \end{align}
    and $p_\rho(n)$ is the trace of the projector onto the $n$ photon subspace of $\rho$,
    \begin{align}
    \begin{split}
        p_\rho(n) &= \tr[\Pi_M(n) \,\rho].
    \end{split}
    \end{align}
\end{theorem}
Note that this derivation equally applies to arbitrary operators  $\hat{O}=\sum_{\bm{n},\bm{m}} c_{\hat{O}}(\bm{n}, \bm{m}) \ketbra{\bm{n}}{\bm{m}}$.
\begin{proof}
    The $M$-mode input state can be written as
    \begin{align}
    \begin{split}
        \rho = \sum_{n, m = 0}^\infty \sum_{\bm{n} \in \Phi_M^n} \sum_{\bm{m} \in \Phi_M^\ell} \frac{c(\bm{n}, \bm{m})}{\sqrt{n_1! \dots n_M! \, m_1! \dots m_M!}} \, \hat{a}_1^{\dag n_1} \dots \hat{a}_M^{\dag n_M} \ketbra{0} \hat{a}_1^{m_1} \dots \hat{a}_M^{m_M}
    \end{split}
    \end{align}
    A passive linear optical network $U$ transforms $\hat{a}^\dag_j \rightarrow U_{jk}\hat{a}_k^\dag$ and $\hat{a}_j \rightarrow U^*_{jk} \hat{a}_k$ for the unitary matrix $U \in \textrm{U}(M)$. Note that, unless otherwise stated, we are using Einstein summation convention where double indices are summed over. The evolved state thus reads
    \begin{align}
    \begin{split}
    \label{eq:evolved_state_for_average_state}
        U \rho U^\dag &= \sum_{n, m = 0}^\infty \sum_{\bm{n} \in \Phi_M^n} \sum_{\bm{m} \in \Phi_M^m} \frac{c(\bm{n}, \bm{m})}{\sqrt{n_1! \dots n_M! \, m_1! \dots m_M!}} \\
        & \quad \times \big(\delta_{1j} U_{jk} \hat{a}_k^\dag\big)^{n_1} \dots \big(\delta_{Mj} U_{jk} \hat{a}_k^\dag \big)^{n_M} \ketbra{0}{0} \big(\delta_{1j'} U^*_{j'k'} \hat{a}_{k'} \big)^{m_1} \dots \big(\delta_{Mj'} U^*_{j'k'} \hat{a}_{k'}  \big)^{m_M} \\
        &= \sum_{n, m = 0}^\infty \sum_{\bm{n} \in \Phi_M^n} \sum_{\bm{m} \in \Phi_M^m} \frac{c(\bm{n}, \bm{m})}{\sqrt{n_1! \dots n_M! \, m_1! \dots m_M!}} \\
        &\quad\quad \times \delta_{1 j_1} \dots \delta_{1 j_{n_1}} \delta_{2 j_{n_1 + 1}} \dots \delta_{2 j_{n_1 + n_2}} \delta_{3 j_{n_1 + n_2 + 1}} \dots \delta_{M j_{n}} \\
        & \quad \quad \times \delta_{1 j'_1} \dots \delta_{1 j'_{m_1}} \delta_{2 j'_{m_1 + 1}} \dots \delta_{2 j'_{m_1 + m_2}} \delta_{3 j'_{m_1 + m_2 + 1}} \dots \delta_{M j'_{m}} \\
        &\quad \quad \times U_{j_1 k_1} U_{j_2 k_2} \dots U_{j'_{n} k'_{n}} \times U^*_{j'_1 k'_1} U^*_{j'_2 k'_2} \dots U^*_{j'_{m} k'_{m}} \\
        & \quad \quad \times \hat{a}_{k_1}^\dag \hat{a}^\dag_{k_2} \dots \hat{a}_{k_{n}}^\dag \ketbra{0}{0} \hat{a}_{k'_1} \hat{a}_{k'_2} \dots \hat{a}_{k'_{m}}
    \end{split}
    \end{align}
    The $j$ indices and Kronecker deltas label which mode the photons came from while the $k$ indices label which mode the photon ends up in. Weingarten calculus \cite{collins2006integration} can be used to average over $\mathrm{U}(M)$ via
    \begin{align}
    \begin{split}
        \int_{\textrm{U}(M)} \textrm{d}\mu(U) U_{j_1 k_1} &\dots U_{j_n k_n} U^*_{j'_1 k'_1} \dots U^*_{j'_m k'_m} \\
        &= \delta_{n,m} \sum_{\sigma, \tau \in \mathcal{S}_n} \textrm{Wg}_{M, n}(\tau \sigma^{-1}) \, \delta_{j_1 j'_{\sigma(1)}} \dots \sigma_{j_n j'_{\sigma(n)}} \delta_{k_1 k'_{\tau(1)}} \dots \delta_{k_n k'_{\tau(n)}}
    \end{split}
    \end{align}
    where $\mathcal{S}_n$ is the symmetric group on $n$ elements and $\textrm{Wg}_{M, n}$ is the Weingarten function.
    Plugging this in to Eq.~\eqref{eq:evolved_state_for_average_state} yields the average output state,
    \begin{align}
    \begin{split}
    \label{eq:evolved_average_state_step1}
        \ex_{\mu}\left[ U \rho U^\dag \right] &= \sum_{n = 0}^\infty \, \sum_{\bm{n}, \bm{m} \in \Phi_M^n} \frac{c(\bm{n}, \bm{m})}{\sqrt{n_1! \dots n_M! \, m_1! \dots m_M!}} \sum_{\sigma, \tau \in \mathcal{S}_n} \textrm{Wg}_{M, n}(\tau \sigma^{-1}) \\
        &\quad \quad \times \delta_{1 j_{\sigma(1)}} \dots \delta_{1 j_{\sigma(n_1)}} \delta_{2 j_{\sigma(n_1 + 1)}} \dots \delta_{2 j_{\sigma(n_1 + n_2)}} \dots \delta_{M j_{\sigma(n)}} \\
        & \quad \quad \times \delta_{1 j_1} \dots \delta_{1 j_{m_1}} \delta_{2 j_{m_1 + 1}} \dots \delta_{2 j_{m_1 + m_2}} \dots \delta_{M j_{n}} \\
        &\quad \quad \times \hat{a}_{k_{\tau(1)}}^\dag \dots \hat{a}^\dag_{k_{\tau(n)}} \ketbra{0}{0} \hat{a}_{k_1} \dots \hat{a}_{k_n} 
    \end{split}
    \end{align}
    The final line can be computed by noting that $[\hat{a}^\dag_k, \hat{a}^\dag_{k'}] = 0$ and hence $\tau$ can be dropped from the creation operator labels; the residual sum over Weingarten functions was computed in Ref~\cite{collins2014integration}
    \begin{align}
    \begin{split}
        \sum_{\tau \in \mathcal{S}_n} \textrm{Wg}_{M, n}( \tau \sigma^{-1}) &= \sum_{\tau' \in \mathcal{S}_n} \textrm{Wg}_{M, n}(\tau') \\
        &= \frac{1}{n!} {n + M - 1 \choose n}^{-1}.
    \end{split}
    \end{align}
    Similarly, we can contract the $k$ indices over the creation/annihilation operators to give the projector onto the $n$ photon subspace across $M$ modes multiplied by $n!$; this was done in Ref.~\cite{taylor2025optimal}
    \begin{align}
    \begin{split}
        \hat{a}_{k_1}^\dag \dots \hat{a}^\dag_{k_n} \ketbra{0}{0} \hat{a}_{k_1} \dots \hat{a}_{k_n} &= n! \, \Pi_M(n).
    \end{split}
    \end{align}
    Incorporating these into Eq.~\eqref{eq:evolved_average_state_step1} gives
    \begin{align}
    \begin{split}
        \ex_{\mu} \left[ U \rho U^\dag \right] &= \sum_{n = 0}^\infty \, \sum_{\bm{n}, \bm{m} \in \Phi_M^n} \frac{c(\bm{n}, \bm{m})}{\sqrt{n_1! \dots n_M! \, m_1! \dots m_M!}} \, {n + M - 1 \choose n}^{-1} \Pi_M(n) \\
        &\quad \quad \times \sum_{\sigma \in \mathcal{S}_n} \delta_{1 j_{\sigma(1)}} \dots \delta_{j_{\sigma(n_1)}} \delta_{2 j_{\sigma(n_1 + 1)}} \dots \delta_{2 j_{\sigma(n_1 + n_2)}} \dots \delta_{Mj_{\sigma(n)}}
    \end{split}
    \end{align}
    The summand on the second line is only non-zero when $\bm{n} = \bm{m}$. Further, it is only non-zero when $\sigma$ is an element of the Young subgroup $\textrm{Y}(\bm{n}) = \mathcal{S}_{\{1, \dots, n_1\}} \times \mathcal{S}_{\{n_1 + 1, \dots, n_1 + n_2\}} \times \dots \times \mathcal{S}_{\{n - n_M, \dots, n\}}$, where $\mathcal{S}_{\{\bullet\}}$ is the symmetric group acting on the set $\{\bullet\}$ (i.e., so $\mathcal{S}_n = \mathcal{S}_{\{1, 2, \dots, n\}}$). The order of this Young subgroup is clearly given by $n_1! n_2! \dots n_M!$. It follows that 
    \begin{align}
    \begin{split}
        \frac{1}{\sqrt{n_1! \dots n_M! \, m_1! \dots m_M!}} \sum_{\sigma \in \mathcal{S}_n} \delta_{1 j_{\sigma(1)}} \dots &\delta_{j_{\sigma(n_1)}} \delta_{2 j_{\sigma(n_1 + 1)}} \dots \delta_{2 j_{\sigma(n_1 + n_2)}} \dots \delta_{Mj_{\sigma(n)}} = 1
    \end{split}
    \end{align}
    and hence we derive the average output,
    \begin{align}
    \begin{split}
        \ex_\mu \left[ U \rho U^\dag \right] &= \sum_{n = 0}^\infty \sum_{\bm{n} \in \Phi_M^n} c(\bm{n}, \bm{n}) \, {n + M -1 \choose n}^{-1} \Pi_M(n) \\
        &= \sum_{n = 0}^\infty p_\rho(n) \, \pi_M(n).
    \end{split}
    \end{align}
\end{proof}

\section{First moment of the loss function}
\label{appendix:first_moment}
In this section, we compute the first moment of the loss function for different initial states. The observables under consideration are constructed from photon number monomials (PNMs), defined as
\begin{align}
\begin{split}
    \hat{M}_{\bm{\alpha}} = \bigotimes_{i=1}^K \hat{n}_i^{\alpha_i}
\end{split}
\end{align}
up to relabelling of the modes, where $\alpha=(\alpha_1, \dots, \alpha_K)$. The $k^{\mathrm{th}}$ moment of the PNMs are given by
\begin{align}
\begin{split}
    \mu_k(\alpha; \rho) &= \ex_U \left[ \tr[\hat{M}_{\bm{\alpha}} \, U \rho U^\dag] \right].
\end{split}
\end{align}
We also consider photon number polynomials (PNPs), given by linear combinations of PNMs; these can be written as
\begin{align}
\begin{split}
    \hat{O} &= \sum_{\bm\alpha \in A} c_{\bm\alpha} \hat{M}_{\bm{\alpha}}
\end{split}
\end{align}
for some $c_{\bm\alpha} \in \mathbb{R}$ and set of powers $A$. The set of PNPs are universal for any observable diagonal in the number basis. The first moment of any PNP is just a linear combination of PNM first moments;
\begin{align}
\begin{split}
    \mu_1(\hat{O}) = \sum_{\bm\alpha \in A} c_{\bm\alpha} \mu_1(\bm\alpha)
\end{split}
\end{align}
and hence we just need to compute $\mu_1(\bm\alpha)$ for different initial states.

\subsection{Derivation for PNMs}
Let us consider an input Fock state $\ket{\bm{n}} = \ket{n_1, \dots, n_M}$ with total photon number $\sum_{i=1}^M n_i = n$. From Theorem~\ref{thm:average_PLON_output_state}, the average output state of $U\ket{\bm{n}}$ is just the maximally mixed state of $n$ photons, $\pi_M(n)$. The first moment $\mu_1(\bm\alpha)$ therefore reads
\begin{align}
\begin{split}
\label{eq:fock_state_pnm_first_moment_step1}
    \mu_1(\bm\alpha; n) = {n + M - 1 \choose n}^{-1} \sum_{\bm{m}\in\Phi_M^n} \left(\prod_{i=1}^M m_i^{\alpha_i} \right)
\end{split}
\end{align}
where $\bm\alpha = (\alpha_1, \dots, \alpha_M)$ and can now include zeros (i.e., $\alpha_j=0$ for $M-K$ entries).

To compute this quantity, we rewrite $m^{\alpha_i} = \sum_{p=0}^{\alpha_i} p! {m \choose p} S(\alpha_i, p)$ where $S(\alpha_i, p)$ are the Stirling numbers of the second kind.
\begin{align}
\begin{split}
\label{eq:first_moment_fock_first_step}
    \mu_1(\bm\alpha; n) &= {n + M - 1 \choose n}^{-1} \sum_{\bm{m} \in \Phi_M^n } \sum_{\bm{p} = 0}^{\bm\alpha} \left(\prod_{i=1}^M {m_i \choose p_i} \, p_i! \, S(\alpha_i, p_i) \right) \\
    &= {n+M-1 \choose n}^{-1} \sum_{\bm{p}= 0}^{\bm\alpha} \left(\prod_{i=1}^M p_i! \, S(\alpha_i, p_i)\right) \sum_{\bm{m}\in\Phi_M^n} {m_i \choose p_i}
\end{split}
\end{align}
where $\sum_{\bm{p} = 0}^{\bm\alpha} = \sum_{p_1 = 0}^{\alpha_1} \dots \sum_{p_M=0}^{\alpha_M}$. The sums over $\bm{m}$ can be expressed using the generating function 
\begin{align}
\begin{split}
    \mathcal{G}(x, p_i) &= \sum_{m_i=0}^\infty {m_i \choose p_i} x^{m_i} \\
    &= \frac{x^{p_i}}{(1-x)^{p_i + 1}} 
\end{split}
\end{align}
via
\begin{align}
\begin{split}
    \sum_{\bm{m} \in \Phi_M^n} \left(\prod_{i=1}^M {m_i \choose k_i} \right) &= [x^n] \left(\prod_{i=1}^M \mathcal{G}(x, p_i) \right) \\
    &= [x^n]\left( x^{P} (1-x)^{-(P+M)} \right) \\
    &= [x^n] \left(\sum_{j=0}^\infty (-1)^j {-(P + M) \choose j} \, x^{j + P} \right) \\
    &= (-1)^{n-P} {-(P + M) \choose n-P} \\
    &= {n + M - 1 \choose n - P}
\end{split}
\end{align}
for $P = \sum_{i=1}^M p_i$. We can feed this into Eq.~\eqref{eq:first_moment_fock_first_step} to give
\begin{align}
\begin{split}
    \mu_1(\bm\alpha; n) &= \sum_{\bm{p} = 0}^{\bm\alpha} \frac{{n+M-1\choose n- P}}{{n+M-1\choose n}} \left(\prod_{i=1}^M p_i! S(\alpha_i, p_i) \right).
\end{split}
\end{align}
Noting that if $\alpha_i = 0$ then there is no contribution to the above sum, we can replace $\alpha = (\alpha_1, \dots, \alpha_M)$ for $\alpha_i \ge 0$ with $\alpha = (\alpha_1, \dots, \alpha_K)$ for $\alpha_i > 0$. Then we can carry out each sub-sum associated with a given value of $P$ individually. Note that if $P < K$ then there is at least one $S(\alpha_i, p_i) = 0$, and hence the sum starts at $P=K$. Using this and the convention that $S(\alpha_i, x) = 0$ if $x>\alpha$, the first moment can be written as
\begin{align}
\begin{split}
    \mu_1(\bm\alpha;n) &= \sum_{P = K}^{|\bm\alpha|} \frac{{n+M-1\choose n-P}}{{n+M-1\choose n}}\sum_{\bm{p}\in\Phi_K^P} \left(\prod_{i=1}^K p_i! \, S(\alpha_i, p_i) \right).
\end{split}
\end{align}
This final sum over the $\bm{p}$ indices does not admit a simple closed-form structure, but can be computed efficiently using the generating function 
\begin{align}
\begin{split}
    \mathcal{F}(x, \alpha_i) &= \sum_{p_i=0}^\infty p_i! \, S(\alpha_i, p_i) x^{p_i}.
\end{split}
\end{align}
Extracting coefficients from a product of generating functions is equivalent to computing the discrete convolution between the analogous sequences.
\begin{align}
\begin{split}
    \sum_{\bm{p}\in\Phi_K^P} \left(\prod_{i=1}^K p_i! \, S(\alpha_i, p_i)\right) &= [x^P] \left(\prod_{i=1}^K \mathcal{F}(x, \alpha_i)\right) \\
    &= \big(\mathcal{F}^{\alpha_1} * \mathcal{F}^{\alpha_2} * \dots *\mathcal{F}^{\alpha_K} \big) (P)
\end{split}
\end{align}
where
\begin{align}
\begin{split}
    \mathcal{F}^{\alpha_i} = \big(j! \, S(\alpha_i, j)\big)_{j=0}^{\alpha_i}
\end{split}
\end{align}
is the sequence associated with generating function $\mathcal{F}(x, \alpha_i)$ and $(f * g)\,(k)$ is the discrete convolution between sequences $f$ and $g$ evaluated on the $k^{\textrm{th}}$ index. It follows that the first moment of the cost function for an observable $\hat{M}_{\bm{\alpha}}$ reads
\begin{align}
\begin{split}
\label{eq:cost_function_first_moment_fock_final}
    \mu_1(\bm\alpha; n) &= \sum_{P=K}^{|\bm\alpha|} \frac{{n+M-1\choose n-P}}{{n+M-1\choose n}} \, \big(\mathcal{F}^{\alpha_1} * \mathcal{F}^{\alpha_2} * \dots *\mathcal{F}^{\alpha_K}\big) (P)
\end{split}
\end{align}
Furthermore, from Theorem~\ref{thm:average_PLON_output_state}, this can be extended to arbitrary input states; for an initial state $\rho$ with a probability of measuring $n$ photons $p_{\rho}(n)$, the average cost function value is
\begin{align}
\begin{split}
\label{eq:cost_function_first_moment_general_final}
    \mu_1(\bm\alpha; \rho) &= \sum_{n=K}^\infty p_\rho(n) \sum_{P=K}^{|\bm\alpha|} \frac{{n+M-1\choose n-P}}{{n+M-1\choose n}} \big(\mathcal{F}^{\alpha_1} * \mathcal{F}^{\alpha_2} * \dots * \mathcal{F}^{\alpha_K}\big) (P).
\end{split}
\end{align}
Note again that if $n < K$ then the summand vanishes, and hence the summation for $n$ also starts at $K$.

\subsection{Independence from system size}
\label{appendix:first_moment_asymptotics}
In this section, we study the behaviour of $\mu_1(\bm\alpha)$ when $\bm\alpha$ is held constant but $M$ and $n$ are allowed to grow. We will show that the first moment of the cost function for photon number monomials $\hat{M}_{\bm{\alpha}}$ asymptotically approaches a constant value if the photon density is held constant with $n = \nu M$ for $\nu \in (0, \infty)$.

Let $\bm\alpha = (\alpha_1, \dots, \alpha_K)$ be a constant such that $K, |\bm\alpha| \ll n, M$. From Eq.~\eqref{eq:cost_function_first_moment_fock_final}, the only dependence on $n$ and $M$ is in the binomial terms ${n+M-1\choose n-P}/{n+M-1\choose n} = \frac{n!}{(n-P)!} \frac{(M-1)!}{(M+P-1)!}$. These scale asymptotically as
\begin{align}
\begin{split}
    \frac{n!}{(n-P)!} = n^P \left(1 - \frac{1}{n} \frac{P(P-1)}{2} + \mathrm{O}(n^{-2}) \right)
\end{split}
\end{align}
and 
\begin{align}
\begin{split}
    \frac{(M-1)!}{(M+P-1)!} = M^{-P} \left(1 - \frac{1}{M} \frac{P(P-1)}{2} + \mathrm{O}(M^{-2})\right).
\end{split}
\end{align}
The first moment of the cost function for large $n$ and $M$ is therefore given by
\begin{align}
\begin{split}
    \mu_1(\bm\alpha; n) &= \sum_{P=K}^{|\bm\alpha|} \left(\frac{n}{M}\right)^P \left(1 - \frac{P(P-1)}{2} (n^{-1} + M^{-1})\right)\big(\mathcal{F}^{\alpha_1} * \dots *\mathcal{F}^{\alpha_K}\big)(P) \\
    & \quad \quad + \mathrm{O}(1/n^2, 1/M^2, 1/nM)
\end{split}
\end{align}
If the photon number scales with system size as $n=\nu M$, then the first moment asymptotically approaches a constant dependent only on the local photon density, $\nu$, and local observable, $\bm\alpha = (\alpha_1, \dots, \alpha_K)$,
\begin{align}
\begin{split}
    \mu_1(\bm\alpha; n=\nu M) &= \sum_{P=K}^{|\bm\alpha|} \nu^P \big(\mathcal{F}^{\alpha_1} * \mathcal{F}^{\alpha_2} * \dots * \mathcal{F}^{\alpha_K} \big) (P) \,\, - \,\,\mathrm{O}(1/M).
\end{split}
\end{align}
The leading term can be expressed as the product of $K$ ordered Bell polynomials, also known as Fubini polynomials.
\begin{align}
\begin{split}
\label{eq:first_moment_ordered_bell_asymptotics}
    \lim_{M\rightarrow \infty} \mu_1(\bm\alpha;n=\nu M) = \prod_{j=1}^K F_{\alpha_j}(\nu), \quad \quad F_{\alpha_j}(\nu) = \sum_{p=0}^{\alpha_j} p! S(\alpha_j, p) \nu^p
\end{split}
\end{align}
The ordered Bell polynomial asymptotics as $\alpha_j$ grows were derived in Ref.~\cite{belovas2023asymptotic} to be
\begin{align}
\begin{split}
\label{eq:ordered_bell_asymptotics}
    F_{\alpha_j}(\nu) \sim \frac{\alpha_j!}{(1+\nu) [\ln(1 + 1/\nu)]^{\alpha_j+1}}
\end{split}
\end{align}
with corrections exponentially small in $\alpha_j$. Combining Eq.~\eqref{eq:first_moment_ordered_bell_asymptotics} with Eq.~\eqref{eq:ordered_bell_asymptotics}, the first moment at constant photon density grows factorially in the order of the photon number operator in each mode $\alpha_j$ for $j\in\{1,\dots,K\}$, and exponentially in the number of supported modes $K$.

The fact the first moment approaches a constant also implies that to estimate $\ell(\bm\theta)$ for observable $\hat{M}_{\bm{\alpha}}$ within a given precision $\varepsilon_{\textrm{additive}}$, the expected number of shots required by Chebyshev's inequality is independent of system size. This further indicates that the \emph{second} moment of the cost function approaches a constant, via
\begin{align}
\begin{split}
    \mu_1(\bm\alpha)^2 \le \mu_2(\bm\alpha) \le \mu_1(2\bm\alpha)
\end{split}
\end{align}
where $2\bm\alpha = (2\alpha_1, 2\alpha_2, \dots, 2\alpha_K)$. Since both the lower and upper bounds are independent of $M$, this implies $\mu_2(\bm\alpha)$ is also asymptotically bounded.

\begin{figure}
    \centering
    \includegraphics[width=0.7\linewidth]{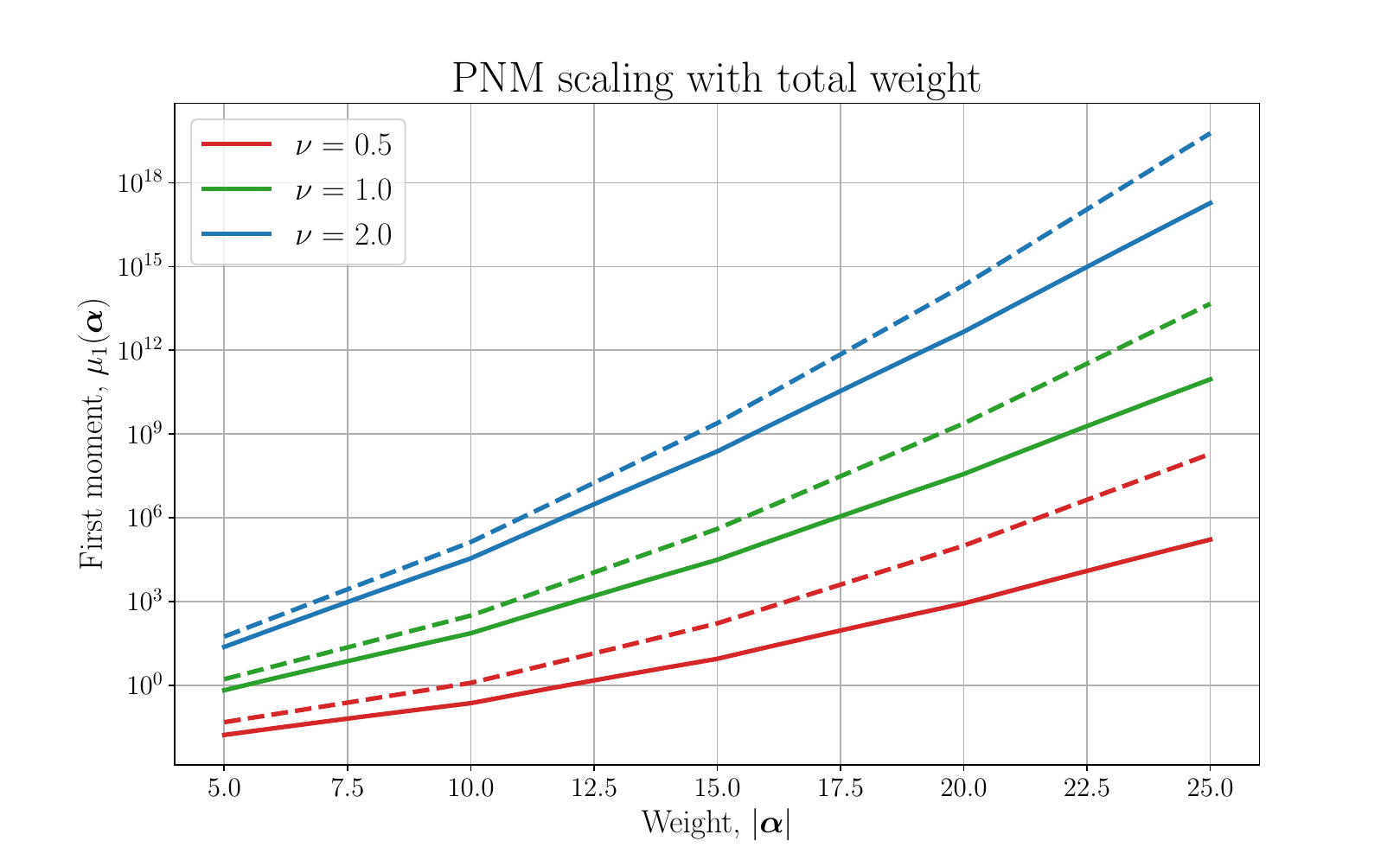}
    \caption{Plot of a PNM first moment against the total weight $|\bm\alpha|$. The solid curve is Eq.~\eqref{eq:pnm_first_moment_main} at $M=50$ and the dotted curve is the $M\rightarrow\infty$ limit from Eq.~\eqref{eq:first_moment_asymptotic_main}. The support is fixed at $K=5$, and the weight is maximally distributed to be as close to $\alpha_j = |\bm\alpha|/K$ for all $j$ as possible. We see a superexponential increase in the first moment as the total weight of the PNMs increase, as predicted in Eq.~\eqref{eq:first_moment_order_scaling_main} and Eq.~\eqref{eq:ordered_bell_asymptotics}. Increasing the density also leads to order of magnitude changes in the first moment. These results indicate that the sample complexity, $\var_{\textrm{samp}}(\bm\alpha) \approx \mu_1(2\bm\alpha)$, for estimating the loss function grows with the order of the largest constituent PNM.}
    \label{fig:pnm_scaling_with_weight}
\end{figure}

\subsection{Implications for PNP estimation}
\label{appendix:pnp_estimation}
Consider an observable $\hat{O}=\sum_{\bm\alpha \in A} c_{\bm\alpha} \hat{M}_{\bm{\alpha}}$ for PNMs of up to a constant size (i.e., $|\bm\alpha|, K_{\bm\alpha}$ are independent of $n$ or $M$ for all $\bm\alpha \in A$). The expected sample variance for estimating $\ell(\bm\theta)$ is upper bounded by
\begin{align}
\begin{split}
    \var_{\textrm{samp}}(\hat{O}) &\le \mu_1(\hat{O}^2) \\
    &= \sum_{\bm\alpha, \bm\beta \in A} c_{\bm\alpha} c_{\bm\beta} \, \mu_1(\bm\alpha + \bm\beta)
\end{split}
\end{align}
and hence the expected number of shots required to estimate $\ell(\theta)$ within an accuracy $\varepsilon_{\textrm{additive}}$ with probability $1-\delta$ satisfies
\begin{align}
\begin{split}
    N \le \frac{\mu_1(\hat{O}^2)}{\delta \, \varepsilon_{\textrm{additive}}^2}.
\end{split}
\end{align}
In cases where all $c_{\bm\alpha} > 0$ and $|\bm\alpha|$ are upper bounded by a constant such that $\mu_1(\bm\alpha + \bm\beta) \sim \mathrm{O}(1)$, the number of shots required scales as $N\sim\mathrm{O}(|A|^2)$. However, if we allow coefficients to be negative then cancellations take place that mean $\mu_1(\hat{O}^2)$ grows more slowly. 
For example, if we let $c_{\bm\alpha}$ be stochastically selected from a distribution satisfying $\ex[c_{\bm\alpha}] = 0$ and $|c_{\bm\alpha}|\sim\mathrm{O}(1)$ for all $\bm\alpha$, then the expected value of $\mu_1(\hat{O}^2)$ over the coefficients is given by
\begin{align}
\begin{split}
    \ex_{\{c_{\bm\alpha}\}} \mu_1(\hat{O}^2) &= \sum_{\bm\alpha\in A} \ex_{\{c_{\bm\alpha}\}}[c_{\bm\alpha}^2] \mu_1(2\bm\alpha) + \sum_{\substack{\bm\alpha,\bm\beta \in A,\\ \bm\alpha\ne\bm\beta}} \ex_{\{c_{\bm\alpha}\}}[c_{\bm\alpha}] \ex_{\{c_{\bm\beta}\}}[c_{\bm\beta}] \, \mu_1(\bm\alpha+\bm\beta) \\
    &= \sum_{\bm\alpha \in A} \, \ex_{\{c_{\bm\alpha}\}}[c_{\bm\alpha}^2] \, \mu_1(2\bm\alpha) \\
    &\sim \mathrm{O}(|A|).
\end{split}
\end{align}
assuming that $|\bm\alpha| \sim \mathrm{O}(1)$.
We would therefore generally expect off-diagonal terms cancelling if the coefficients are equally likely to be positive as negative, leading to a number of shots scaling as $N
\sim \mathrm{O}(|A|)$. For a given accuracy $\varepsilon_{\textrm{additive}}$, the number of samples required by the quantum device therefore asymptotically scales as
\begin{align}
\begin{split}
\label{eq:quantum_sample_complexity_overview}
    N \sim \begin{cases}
        \frac{\mathrm{O}(|A|^2)}{\delta \varepsilon_{\textrm{additive}}^2} & \textrm{generic coefficients}, \\
        \frac{\mathrm{O}(|A|)}{\delta \varepsilon_{\textrm{additive}}^2} & \textrm{with off-diagonal cancellations}
    \end{cases}
\end{split}
\end{align}

\subsection{Comparison with the Lim-Oh algorithm}
\label{appendix:lim_oh}
The Lim-Oh algorithm~\cite{limClassicalAlgorithmsEstimating2025} gives a sampling based method for computing expectation values from passive linear optical circuits when the input state and observable are factorisable. 
To estimate the expectation value of an observable $\hat{O}$ with additive error $\varepsilon_{\textrm{additive}}$ and success probability $1-\delta$, the required number of samples, $N^{\textrm{(LO)}}$, scales as
\begin{align}
\begin{split}
   N^{(\textrm{LO})} \sim \mathrm{O}\left( \frac{||\hat{O}_{\textrm{eff}}||_2^2 \log(1/\delta)}{\varepsilon_{\textrm{additive}}^2} \right)
\end{split}
\end{align}
for pure product states, where $\hat{O}_{\textrm{eff}}$ is the effective observable such that $\hat{O} = \hat{O}_{\textrm{eff}} \otimes \mathds{1}$. Generating each sample has a time complexity $\mathrm{O}(M^2)$.

To estimate a PNP, one approach would be to use the Lim-Oh algorithm to directly estimate the observable $\hat{O}_{\textrm{eff}}$. However, for PNPs that span the entire system, this norm term $||\hat{O}_{\textrm{eff}}||_2^2$ grows exponentially, leading to exponential classical overhead. Instead, one can use the Lim-Oh algorithm to estimate the constituent PNM expectation values $\tr[\hat{M}_{\bm\alpha} \rho(\bm\theta)]$ and then sum these estimates together. We assume that the constituent PNMs are constant with respect to system size such that $||\hat{M}_{\bm\alpha, \textrm{eff}}||_2^2 \sim \mathrm{O}(1)$. One can then either (i) directly estimate every single PNM term and sum them together or (ii) Monte-Carlo sample from $\{\bm\alpha\}$ weighted by $|c_{\bm\alpha}| / \sum_{\bm\alpha \in A} |c_{\bm\alpha}|$.

Considering option (i), to derive the classical complexity of the Lim-Oh algorithm we look at the classical variance associated with estimating $\ell(\bm\theta)=\sum_{\bm\alpha\in A} c_{\bm\alpha} \tr[\hat{M}_{\bm\alpha} \rho(\bm{\theta})]$; we let $\bar{\ell}$ and $\bar{M}_{\bm{\alpha}}$ be the classical estimate of $\ell$ and $\tr[\hat{M}_{\bm\alpha} \rho]$ respectively. The total classical estimation variance reads
\begin{align}
\begin{split}
    \var[\bar{\ell}] &= \sum_{\bm{\alpha}\in A} c_{\bm\alpha}^2 \var[\bar{M}_{\bm\alpha}] + \sum_{\substack{\bm\alpha, \bm\beta \in A, \\ \bm\alpha \ne \bm\beta}} c_{\bm\alpha} c_{\bm\beta} \, \textrm{Cov}[\bar{M}_{\bm\alpha}, \bar{M}_{\bm\beta}]
\end{split}
\end{align}
Each $\bar{M}_{\bm\alpha}$ is estimated independently and therefore we assume they have uncorrelated errors -- this means all covariances satisfy $\textrm{Cov}[\bar{M}_{\bm\alpha}, \bar{M}_{\bm\beta}] = 0$ for $\bm\alpha \ne \bm\beta$.
From Chebyshev's inequality, to estimate $\ell$ within additive error $\varepsilon_{\textrm{additive}}$ with probability $1-\delta$ requires the variance to satisfy
\begin{align}
\begin{split}
\label{eq:chebyshev_for_lim-oh}
    \var[\bar{\ell}] &= \sum_{\bm\alpha\in A} c_{\bm\alpha}^2 \var[\bar{M}_{\bm\alpha}] \\
    &\le \delta \varepsilon_{\textrm{additive}}^2
\end{split}
\end{align}

Let us imagine each estimate $\bar{M}_{\bm\alpha}$ used $N_{\bm\alpha}$ samples from the Lim-Oh algorithm, with each sample labelled $m_{\bm\alpha}^{(x)}$ for $x\in\{1, \dots, N_{\bm\alpha}\}$. The estimator of the PNM reads
\begin{align}
\begin{split}
    \bar{M}_{\bm\alpha} = \sum_{x=1}^{N_{\bm\alpha}} m_{\bm\alpha}^{(x)}
\end{split}
\end{align}
It is unbiased (so $\bar{M}_{\bm\alpha}|_{N_{\bm\alpha}\rightarrow \infty} = \tr[\hat{M}_{\bm{\alpha}} \rho(\bm\theta)]$) and has a variance
\begin{align}
\begin{split}
    \var[\bar{M}_{\bm\alpha}] &= \frac{\var[m^{(x)}_{\bm\alpha}]}{N_{\bm\alpha}}
\end{split}
\end{align}
where $\var[m^{(x)}_{\bm\alpha}]$ is the variance of sample outcomes from the Lim-Oh algorithm. For PNMs of constant size, we assume this satisfies $\var[m_{\bm\alpha}^{(x)}] \sim \mathrm{O}(1)$.

Let us first assume that $N_{\bm\alpha} = N_{\bm\beta} = N_{A}$. In this case, Eq.~\eqref{eq:chebyshev_for_lim-oh} can be written as
\begin{align}
\begin{split}
    \frac{1}{N_A}\sum_{\bm\alpha \in A} c_{\bm\alpha}^2 \var[m^{(x)}_{\bm\alpha}] \le \delta \varepsilon_{\textrm{additive}}^2 \implies N_A \ge \frac{\mathrm{O}(|A|)}{\delta \varepsilon_{\textrm{additive}}^2}
\end{split}
\end{align}
where we used $\sum_{\bm\alpha \in A} \mathrm{O}(1) = \mathrm{O}(|A|)$. In this case, \emph{each} PNM estimator would require $\mathrm{O}(|A|)/\delta\varepsilon_{\textrm{additive}}^2$ samples. Since there are $|A|$ such PNMs, the total sample complexity of estimating $\ell(\bm{\theta})$ with the Lim-Oh algorithm scales as
\begin{align}
\begin{split}
    N^{(\textrm{LO})} \sim \mathrm{O}\left(\frac{|A|^2}{\delta \varepsilon_{\textrm{additive}}^2}\right)
\end{split}
\end{align}
to achieve additive error $\varepsilon_{\textrm{additive}}$ and a success of at least probability $1-\delta$.
Intuitively, if we want $\bar{\ell}$ to have an additive error $\varepsilon_{\textrm{additive}}$ from summing $|A|$ terms together, then we need the error of each individual PNM estimate to scale as $\varepsilon_{\textrm{additive}} / \sqrt{|A|}$.
Notice that this is independent of the sign of the coefficients $c_{\bm\alpha}$ since only their squared value appeared.

We now briefly look at option (ii), a Monte-Carlo/importance sampling approach. Rather than estimating $\bar{M}_{\bm\alpha}$ for each $\alpha\in A$, we stochastically select a $\bm\alpha$ from the probability distribution $p_{\bm\alpha} = |c_{\bm\alpha}| / S$, where $S = \sum_{\bm\alpha \in A} |c_{\bm\alpha}|$. 
At each step, we compute one sample from the Lim-Oh algorithm from the stochastically selected $\bm\alpha$, until we have a total of $N^{(\textrm{LO})}$ samples $m^{(x)}$ for $x\in\{1,\dots, N^{(\textrm{LO})}\}$; each sample has an associated $\bm{\alpha}^{(x)}$. 
The unbiased estimator $\bar{\ell}$ is then given by
\begin{align}
\begin{split}
    \bar{\ell} = \frac{S}{N^{\textrm{(LO)}}} \sum_{x=1}^{N^{(\textrm{LO})}} \textrm{sgn}(\bm\alpha^{(x)}) \,m^{(x)}.
\end{split}
\end{align}
From this, it follows that the variance of the estimator scales as $S^2/N^{(\textrm{LO})}$; since $S = \sum_{\bm\alpha \in A} |c_{\bm\alpha}| \sim \mathrm{O}(|A|)$, the total variance $\var[\bar{\ell}]$ admits the same scaling with respect to $|A|$, $\delta$ and $\varepsilon_{\textrm{additive}}$ as directly estimating each PNM, namely
\begin{align}
\begin{split}
    N^{(\textrm{LO})} \sim \left(\frac{|A|^2}{\delta \varepsilon_{\textrm{additive}}^2}\right)
\end{split}
\end{align}
Therefore, the classical complexity using one of the best out-of-the-box 
Indeed, these derivations apply to any sampling based approach that can estimate a PNM to a given accuracy (i.e., $\varepsilon_{\bm\alpha}$) with a constant number of samples. 

These results combined with those from Eq.~\eqref{eq:quantum_sample_complexity_overview} suggest that a quadratic speedup over the Lim-Oh algorithm is possible for a wide class of PNPs, namely those with suitably many positive and negative terms. Furthermore, the complexity is determined by the number of terms in the PNP, rather than the total number of modes or photons. This means that the computational complexity can be scaled up much faster than the hardware; for an $M$ mode system and taking $|\bm\alpha| \le L$, the PNP can have up to $|A|\sim\mathrm{O}(M^L)$ terms. 
Finally, the quadratic speedup is independent of whether the system is efficiently trainable or not, because $\varepsilon_{\textrm{additive}}^2 = \varepsilon^2 \var_{\textrm{circ}}$ -- the term that indicates the presence or absence of barren plateaus -- remains undetermined.

\section{Second moment of the loss function}
\label{appendix:second_moment}
We now proceed to compute the second moment of the loss function. We focus on the case of photon number polynomials given by
\begin{align}
\begin{split}
    \hat{C} = \sum_{\bm\alpha \in A} c_{\bm\alpha} \hat{M}_{\bm{\alpha}}
\end{split}
\end{align}
The second moment for an $n$-photon initial state $\rho$ is then defined as 
\begin{align}
\begin{split}
    \mu_2(\hat{C})= \ex_{U} \left[ \tr[\hat{C} U \rho U^\dag] \right]
\end{split}
\end{align}
In Ref.~\cite{mhiri2026bosonsamplingdiluteregime}, the authors derived the second moment of particle-number-preserving observable expectation values for input states with a definite particular number. Their expression is in terms of irreducible representations of the unitary group $\mathrm{U}(M)$ indexed by $k$ for $k\in\{0,1,\dots,n\}$. In particular, the second moment of some expectation value reads
\begin{align}
\begin{split}
    \mu_2(\hat{C}) = \sum_{k=0}^n \frac{1}{d^{(n)}_k} ||\rho_k^{(n)} ||_2^2\, ||\hat{C}_k^{(n)} ||_2^2
\end{split}
\end{align}
where $\hat{X}_k^{(n)}$ is the operator $\hat{X}$ projected onto the $k^{\textrm{th}}$ irreducible representation, $d_k^{(n)}$ is the associated dimension, and $||\bullet||_2^2 = \tr[\bullet^\dag \bullet]$ is the Hilbert-Schmidt norm. The projections onto each irreducible representation can be computed using a lowering map~\cite{mhiri2026bosonsamplingdiluteregime}
\begin{align}
\begin{split}
    L(\bullet) = \sum_{j=1}^M \hat{a}_j \bullet \hat{a}_j^\dag
\end{split}
\end{align}
In particular, Ref.~\cite{mhiri2026bosonsamplingdiluteregime} proved that
\begin{align}
\begin{split}
    ||\hat{X}_k^{(n)}||_2^2 &= \sum_{\ell = 0}^k (-1)^{k-\ell} \frac{2k+M-1}{k+\ell+M-1} \frac{{n-\ell\choose k-\ell}}{{n+k+M-1\choose n-\ell}} \frac{g_{\ell}(\hat{X})}{(n-\ell)!^2}
\end{split} 
\end{align}
and 
\begin{align}
\begin{split}
    g_{\ell}(\hat{X}) &= \tr[\big(L^{n-\ell}(\hat{X})\big)^2]
\end{split}
\end{align}
The problem is thus reduced to computing $g_{\ell}(\hat{X})$ for $\hat{X}\in\{\rho, \hat{O}\}$.

\subsection{Fock state projected norm}
From Ref.~\cite{mhiri2026bosonsamplingdiluteregime}, for an initial Fock state $\ket{\bm{n}} = \ket{n_1, n_2, \dots, n_M}$ we have
\begin{align}
\begin{split}
    g_\ell(\ketbra{\bm{n}}) &= (n-\ell)!^2 \sum_{\bm{a} \in \Phi_M^{n-\ell}} \left(\prod_{i = 1}^M {n_i \choose a_i}^2 \right)
\end{split}
\end{align}
where ${n_i \choose a_i}$ enforces $a_i \le n_i$. This can be computed using generating the functions
\begin{align}
\begin{split}
    \mathcal{A}(x;n_i) = \sum_{a_i=0}^\infty {n_i \choose a_i}^2 x^{a_i}.
\end{split}
\end{align}
from which it follows that
\begin{align}
\begin{split}
    \sum_{\bm{a}\in\Phi_{M^{n-\ell}}} \left(\prod_{i=1}^M {n_i\choose a_i}^2\right) &= [x^{n-\ell}] \left(\prod_{i=1}^M A(x; n_i)\right).
\end{split}
\end{align}
The coefficient can be efficiently extracted using discrete convolution of sequences $\mathcal{A}_i^{n_i}$
\begin{align}
\begin{split}
    \mathcal{A}_i^{n_i} &= \big( {n_i \choose a}^2 \big)_{a=0}^n
\end{split}
\end{align}
and hence
\begin{align}
\begin{split}
    g_{\ell}(\ketbra{\bm{n}}) = (n-\ell)!^2 \big( \mathcal{A}_1^{n_1} * \mathcal{A}_2^{n_2} * \dots * \mathcal{A}_M^{n_M} \big)(n-\ell).
\end{split}
\end{align}

\subsection{Entangled state projected norm}
\label{appendix:entangled_second_moment}

We now consider equally weighted superpositions over all basis states.
\begin{align}
\begin{split}
    \ket{\psi} &= \sum_{\bm{n} \in \Phi_M^n} \frac{\ket{\bm{n}}}{\sqrt{n+M-1\choose n}}.
\end{split}
\end{align}
While not physically motivated, these are the maximally dispersed pure state containing a definite $n$ photons so we expect quantitatively different behaviour than for pure Fock states. 
Furthermore, these highly correlated states lie beyond domain of applicability of the Lim-Oh algorithm due to the large-scale global correlations between modes. 

Applying the lowering operator $p$ times to $\rho = \ketbra{\psi}{\psi}$ gives
\begin{align}
\begin{split}
    L^p(\rho) &= \frac{1}{\binom{n+M-1}{n}} \sum_{\bm{n},\bm{m}\in\Phi_M^n} \sum_{\substack{\bm{c} \in \Phi_M^p \\ \bm{c}\le\min(\bm{n},\bm{m})}} \binom{p}{c_1,\dots,c_M} \prod_{i=1}^M a_i^{c_i}\ket{n_i}\bra{m_i}(a_i^\dagger)^{c_i}.
\end{split}
\end{align}
Just as in Ref.~\cite{mhiri2026bosonsamplingdiluteregime}. The action of the creation/annihilation operators together with the multinomial coefficient simplifies to
\begin{align}
\begin{split}
    \binom{p}{c_1,\dots,c_M}\prod_{i=1}^M \sqrt{\frac{n_i!}{(n_i-c_i)!}}\sqrt{\frac{m_i!}{(m_i-c_i)!}} = p!\prod_{i=1}^M \sqrt{\binom{n_i}{c_i}\binom{m_i}{c_i}},
\end{split}
\end{align}
so that, 
\begin{align}
\begin{split}
    L^p(\rho) &= \frac{p!}{\binom{n+M-1}{n}} \sum_{\bm{n},\bm{m}\in\Phi_M^n} \sum_{\bm{c} \in \Phi_M^p } \left(\prod_{i=1}^M \sqrt{\binom{n_i}{c_i}\binom{m_i}{c_i}}\right) \ketbra{\bm{n}-\bm{c}}{\bm{m}-\bm{c}}.
    \label{eq:Lp-rho-final-ent}
\end{split}
\end{align}
where $c_i \le \min[n_i, m_i]$ is enforced by the binomial coefficient ${x\choose y}$ being zero if $y > x$.

Taking the Hilbert--Schmidt norm of Eq.~\eqref{eq:Lp-rho-final-ent} requires two independent copies of $(\bm{n},\bm{m},\bm{c})$,
\begin{align}
\begin{split}
    \tr\big[L^p(\rho)^2\big] &= \left(\frac{p!}{\binom{n+M-1}{n}}\right)^2 \sum_{\bm{n}_1,\bm{n}_2,\bm{m}_1,\bm{m}_2,\bm{c}_1,\bm{c}_2} W_1 W_2 \, \delta_{\bm{m}_1-\bm{c}_1,\,\bm{n}_2-\bm{c}_2}\, \delta_{\bm{m}_2-\bm{c}_2,\,\bm{n}_1-\bm{c}_1},
\end{split}
\end{align}
where $W_1 = \prod_i \sqrt{\binom{n_{1,i}}{c_{1,i}}\binom{m_{1,i}}{c_{1,i}}}$. Resolving the deltas leaves the free indices $\bm{n}_1,\bm{m}_1,\bm{c}_1,\bm{c}_2$,
\begin{align}
\begin{split}
    \tr\big[L^p(\rho)^2\big] &= \left(\frac{p!}{\binom{n+M-1}{n}}\right)^2 \sum_{\bm{n}_1,\bm{m}_1,\bm{c}_1,\bm{c}_2} W_1 W_2',
\end{split}
\end{align}
with $W_2' = \prod_j \sqrt{\binom{(\bm{n}_1-\bm{c}_1+\bm{c}_2)_j}{c_{2,j}}\binom{(\bm{m}_1-\bm{c}_1+\bm{c}_2)_j}{c_{2,j}}}$.
As before, the sum factorises over modes and can be evaluated with a generating function. Defining
\begin{align}
\begin{split}
    \mathcal{T}_i(x,y,u,v) &= \sum_{n_{1,i},m_{1,i},c_{1,i},c_{2,i}\ge 0} 
    \sqrt{\binom{n_{1,i}}{c_{1,i}}\binom{m_{1,i}}{c_{1,i}}\binom{n_{1,i}-c_{1,i}+c_{2,i}}{c_{2,i}}\binom{m_{1,i}-c_{1,i}+c_{2,i}}{c_{2,i}}} \\
    &\qquad \times \, x^{n_{1,i}}y^{m_{1,i}}u^{c_{1,i}}v^{c_{2,i}},
    \label{gen-func}
\end{split}
\end{align}
and $\mathcal{T}(x,y,u,v) = \prod_{i=1}^M \mathcal{T}_i(x,y,u,v)$, the constraints $\bm{n}_1,\bm{m}_1\in\Phi_M^n$ and $\bm{c}_1, \bm{c}_2 \in \Phi_M^p$ are recovered by extracting $[x^n y^n u^p v^p]\mathcal{T}(x,y,u,v)$, equivalently by the convolution
\begin{align}
\begin{split}
    \big(\mathcal{T}_1 * \dots * \mathcal{T}_M\big)(n,n,p,p).
\end{split}
\end{align}
Replacing $p$ with $n-\ell$ gives
\begin{align}
\begin{split}
    g_{\ell}(\ketbra{\psi}{\psi}) &= \left(\frac{(n-\ell)!}{\binom{n+M-1}{n}}\right)^2 \big(\mathcal{T}_1 * \dots * \mathcal{T}_M\big)(n,n,n-\ell,n-\ell).
    \label{gell_ent}
\end{split}
\end{align}
While $\bm{c}_1\in\Phi_M^p$ is enforced automatically by the binomial coefficients in Eq.~\eqref{gen-func}, the same is not true for $\bm{c}_2$, whose total is fixed only through the coefficient extraction itself.

\subsection{Photon number polynomial projected norm}
\label{appendix:pnp_second_moment}
To compute $g_\ell(\hat{C}) = \tr[L^{n-\ell}(\hat{C})^2]$ for $\hat{C} = \sum_{\bm\alpha \in A} c_{\bm\alpha} \hat{M}_{\bm\alpha}$, we can use the fact that
\begin{align}
\begin{split}
    L^{n-\ell}(\ketbra{\bm{n}}) &= (n-\ell)! \sum_{\bm{a}\in\Phi_M^\ell} \left(\prod_{i=1}^M {n_i \choose a_i}\right) \ketbra{\bm{a}}
\end{split}
\end{align}
from Ref.~\cite{mhiri2026bosonsamplingdiluteregime}. The lowering operator on the PNP then gives
\begin{align}
\begin{split}
    L^{n-\ell}(\hat{C}) &= \sum_{\bm\alpha\in A} c_{\bm\alpha} \, L^{n-\ell}\big(\hat{M}_{\bm{\alpha}}\big) \\
    &= \sum_{\bm\alpha\in A} \sum_{\bm{n}\in\Phi_M^n} c_{\bm\alpha} \left(\prod_{i=1}^M n_i^{\alpha_i} \right) L^{n-\ell}(\ketbra{\bm{n}}) \\
\end{split}
\end{align}
from which we can derive $g_{\ell}(\hat{C}) = \tr[L^{n-\ell}(\hat{C})^2]$,
\begin{align}
\begin{split}
    g_\ell(\hat{C}) &= (n-\ell)!^2 \sum_{\bm\alpha, \beta \in A} c_{\bm\alpha} c_{\beta} \sum_{\bm{n}, \bm{m} \in \Phi_M^n} \sum_{\bm{a}, \bm{b} \in \Phi_M^\ell} \left(\prod_{i=1}^M n_i^{\alpha_i} m_i^{\beta_i} {n_i \choose a_i} {m_i \choose b_i} \right) \big|\braket{\bm{a}}{\bm{b}}\big|^2 \\
    &= (n-\ell)!^2 \sum_{\bm\alpha, \beta \in A} c_{\bm\alpha} c_{\beta} \sum_{\bm{n}, \bm{m} \in \Phi_M^n} \sum_{\bm{a}\in \Phi_M^\ell} \left(\prod_{i=1}^M n_i^{\alpha_i} {n_i \choose a_i}\right) \left(\prod_{j=1}^M m_i^{\beta_i} {m_i \choose a_i}\right).
\end{split}
\end{align}
The sums over $\bm{n}$ and $\bm{m}$ can be computed using generating functions. Let us define 
\begin{align}
\begin{split}
    \mathcal{G}(x; \alpha_i, \bm{a}) = \sum_{n=0}^\infty n^{\alpha_i} {n_i \choose a_i} x^{n_i}.    
\end{split}
\end{align}
Then it follows that
\begin{align}
\begin{split}
    \sum_{\bm{n} \in \Phi_M^n} \left(\prod_{i=1}^M n_i^{\alpha_i} {n_i \choose a_i}\right) &= [x^n] \prod_{i=1}^M \mathcal{G}(x; \alpha_i, a_i)
\end{split}
\end{align}
and hence
\begin{align}
\begin{split}
    \sum_{\bm{a}\in\Phi_M^\ell} \sum_{\bm{n},\bm{m}\in\Phi_M^n} \left(\prod_{i=1}^M n_i^{\alpha_i} m_i^{\beta_i} {n_i \choose a_i} {m_i \choose a_i} \right) \\
    = [x^n y^n] \sum_{\bm{a} \in \Phi_M^\ell} \left(\prod_{i=1}^M \mathcal{G}(x; \alpha_i, a_i) \mathcal{G}(y;\beta_i, a_i)\right).
\end{split}
\end{align}
Applying a second generating function
\begin{align}
\begin{split}
    \mathcal{F}(x, y, z; \alpha_i, \beta_i) &= \sum_{a_i=0}^\infty \mathcal{G}(x; \alpha_i, a_i) \, \mathcal{G}(y; \beta_i, a_i) \, z^{a_i}
\end{split}
\end{align}
to the sum over $\bm{a}$ allows us to write
\begin{align}
\begin{split}
    \sum_{\bm{a}\in\Phi_M^\ell} \sum_{\bm{n}, \bm{m} \in \Phi_M^n} \left( \prod_{i=1}^M n_i^{\alpha_i} m_i^{\beta_i} {n_i \choose a_i} {m_i \choose a_i} \right) &= [x^n y^n z^\ell] \left( \prod_{i=1}^M \mathcal{F}(x, y, z; \alpha_i, \beta_i)\right).
\end{split}
\end{align}
We can therefore write $g_{\ell}(\hat{C})$ as
\begin{align}
\begin{split}
    g_{\ell}(\hat{C}) = (n-\ell)!^2 \sum_{\bm\alpha, \beta \in A} c_{\bm\alpha} c_{\bm\beta} \, [x^n y^n z^\ell] \left(\prod_{i=1} \mathcal{F}(x, y, z; \alpha_i, \beta_i) \right)
\end{split}
\end{align}
The coefficient extraction can be written as the discrete convolutions of the trivariate sequences
\begin{align}
\begin{split}
    \mathcal{F}^{\alpha_i,\beta_i}_i = \big( p^{\alpha_i} q^{\beta_i} {p \choose r} {q\choose r} \big)_{p,q,r=0}^n
\end{split}
\end{align}
associated with each generating function $\mathcal{F}(x, y, z; \alpha_i, \beta_i)$ via
\begin{align}
\begin{split}
    [x^n y^n z^\ell] \left(\prod_{i=1}^M \mathcal{F}(x,y,z;\alpha_i,\beta_i)\right) = \big(\mathcal{F}_1^{\alpha_1,\beta_1} * \mathcal{F}_2^{\alpha_2,\beta_2} * \dots * \mathcal{F}_M^{\alpha_M,\beta_M}\big)(n,n,\ell).
\end{split}
\end{align}
For PNPs, we can therefore efficiently calculate $g_{\ell}$ with
\begin{align}
\begin{split}
    g_{\ell}(\hat{C}) = (n-\ell)!^2 \sum_{\bm\alpha, \bm\beta \in A} c_{\bm\alpha} c_{\bm\beta} \, \big(\mathcal{F}_1^{\alpha_1, \beta_1} * \mathcal{F}_2^{\alpha_2,\beta_2} * \dots * \mathcal{F}_M^{\alpha_M,\beta_M} \big) (n, n, \ell).
\end{split}
\end{align}

\section{Gradient estimation with parameter-shift rules}

\label{app:parameter_shift}
 
The gradient-based training considered in this work requires estimating derivatives of the loss with respect to individual circuit parameters. In linear optics with Fock input states this can be done exactly through a generalised parameter-shift rule (GPSR) \cite{facelli2024exactgradients, trudeau2026parameter, pappalardo2025photonic, hoch2025variational}, which we briefly outline here. By considering Fourier decompositions of the loss with respect to a single phase-shifter parameter, the first derivative of the loss with respect to a parameter $\theta_k$ can be written as,
\begin{align}
\label{eq:app_gpsr}
\frac{\partial \ell(\bm{\theta})}{\partial \theta_k}
     = \sum_{j=1}^{2R} \ell(\bm{\theta} + \kappa_j\,\bm{e}_k)\,\frac{(-1)^{j+1}}{4R\sin^2(\kappa_j/2)}, \qquad \kappa_j = \frac{(2j-1)\pi}{2R}.
\end{align}
The gradient of the loss with respect to any circuit parameter is therefore an exact, fixed linear combination of $2R$ loss values at shifted parameter settings. For arbitrary observables, this holds where $R = n$: the number of photons in the circuit. Unlike finite-difference approximations, replacing each loss value in \eqref{eq:app_gpsr} by its empirical estimate $\bar{\ell}_N$ yields an unbiased estimator of the gradient. The same holds for beam splitter parameters, since a tunable beam splitter decomposes into fixed beam splitters and a tunable phase \cite{facelli2024exactgradients}.

For the photon number observables considered in this work, the number of shifted loss values needed to compute the gradient can be reduced to the maximum order of the observable. For loss governed by a photon number polynomial, it is sufficient to take $R = |\bm{\alpha}|$ in \eqref{eq:app_gpsr}, where $|\bm{\alpha}|$ is the greatest order of a photon number monomial appearing in the PNP. For observables whose degree is fixed independently of system size, as for the fixed-order PNPs studied in this work, the number of loss evaluations per gradient component is therefore constant.
 
Equation~\eqref{eq:app_gpsr} reduces gradient estimation to cost function estimation, which is why the analysis of this work centres on the latter. Following the statistical analysis of ~\cite{facelli2024exactgradients}, replacing each loss value in \eqref{eq:app_gpsr} by its empirical estimate from $N$ samples yields an unbiased gradient estimator $\overline{\partial_{\theta_k}\ell}(\bm{\theta})$. Assuming that average sample variance is equal at shifted loss values (necessarily true for uniform periodic parameters), the average sample variance of the gradient estimate is given by:
\begin{align}
\label{eq:app_gpsr_variance_average}
    \ex_{\bm{\theta}}\Big[\var\big[\overline{\partial_{\theta_k}\ell}(\bm{\theta})\big]\Big] = \frac{2R^2+1}{6} \, \frac{\mathrm{Var}_{\mathrm{samp}}(\hat{O})}{N}.
\end{align}
where we have used $\sum_{j=1}^{2R} \big(4R\sin^2(\kappa_j/2)\big)^{-2} = (2R^2+1)/6$ (we derive this following appendix A4 of \cite{wierichs2022general}).
 
 Equation~\eqref{eq:app_gpsr_variance_average} states that the sample variance controls the error of gradient estimates exactly as it does for the loss, up to a factor depending on the number of shifts. This factor scales at worst polynomially as $n^2$ and, for the case of fixed-order photon number observables, scales as a constant. By Markov's inequality the same holds at typical parameter settings, exactly as in the proof of theorem~\ref{thm:variance_ratio_sufficiency}. Appendix \ref{appendix:bp_equivalence} shows the corresponding statement for circuit variance, that the circuit variance of gradients scales with the loss variance $\text{Var}_{\text{circ}}$. Consequently, the sample complexity for estimating gradients to within proportional error inherits, up to polynomial prefactors, the $\mathrm{O}\big(\mathrm{Var}_{\mathrm{samp}}(\hat{O})/\mathrm{Var}_{\mathrm{circ}}(\hat{O})\big)$ scaling of theorem \ref{thm:variance_ratio_sufficiency}, and the sample and circuit variances remain the main quantities governing the complexity of gradient estimation.

\section{Output binning observable variance}
\label{appendix:output_probability}
In this section, we compute the Haar-average moments and variances of observables built from output probabilities: single-outcome projectors, and coarse-grained sums of projectors over a bin of outcomes.

\subsection{Single-outcome projectors}
\label{appendix:single_outcome_variance}
We consider an input state of fixed photon number $n$ over $M$ modes, and the projector observable $\ketbra{\mathbf{s}}{\mathbf{s}}$ onto an output pattern $\mathbf{s}$. We calculate the required moments:
\begin{align}
    \mu_1(\ketbra{\mathbf{s}}{\mathbf{s}}) &= \ex_{\bm{\theta}}\left[\mathbb{P}(\mathbf{s}|\bm{\theta}) \right] = |\mathcal{H}|^{-1}, \\
    \mu_1(\ketbra{\mathbf{s}}{\mathbf{s}}^2) &= \mu_1(\ketbra{\mathbf{s}}{\mathbf{s}}) = |\mathcal{H}|^{-1}, \\
    \mu_2(\ketbra{\mathbf{s}}{\mathbf{s}}) &= P_2(M,n)\mu_1(\ketbra{\mathbf{s}}{\mathbf{s}}) = P_2(M,n)|\mathcal{H}|^{-1},
\end{align}
where the first result comes from all outcomes being equally likely in the Haar-average regime, the second from projectors being idempotent, and $P_2(M,n)$ is the normalised average collision probability \cite{mhiri2026bosonsamplingdiluteregime}. Taking differences between the terms, we get the variances,
\begin{align}
    \text{Var}_{\text{circ}}(\ketbra{\mathbf{s}}{\mathbf{s}}) &= (P_2(M,n) - 1)|\mathcal{H}|^{-2}, \\
    \text{Var}_{\text{samp}}(\ketbra{\mathbf{s}}{\mathbf{s}}) &= |\mathcal{H}|^{-1} \left(1 - P_2(M,n)|\mathcal{H}|^{-1}\right).
\end{align}

\subsection{Coarse-grained binning}
\label{appendix:binning}
To calculate the average circuit variance of output binned observables, we will consider uniformly at random selecting output sets $S$ and average variance over them. We consider a fixed input state of $n$ photons across $M$ modes. The circuit variance of $\hat{O}_S$ decomposes as
\begin{equation}
    \text{Var}_{\text{circ}}(\hat{O}_S) = \sum_{\mathbf{s} \in S}\text{Var}_\mu\big(\mathbb{P}(\mathbf{s}|\theta)\big) + \sum_{\mathbf{s} \neq \mathbf{t} \in S}\text{Cov}_\mu\big(\mathbb{P}(\mathbf{s}|\theta), \mathbb{P}(\mathbf{t}|\theta)\big).
    \label{eq:varcirc_decomp}
\end{equation}

Setting $S = \mathcal{H}$ gives $\hat{O}_{\mathcal{H}} = 1$ fixed, hence $\text{Var}_{\text{circ}}(\hat{O}_{\mathcal{H}}) = 0$, and Eq.~\eqref{eq:varcirc_decomp} gives the identity
\begin{equation}
    \sum_{\mathbf{s} \neq \mathbf{t} \in \mathcal{H}}\text{Cov}_\mu\big(\mathbb{P}(\mathbf{s}|\theta), \mathbb{P}(\mathbf{t}|\theta)\big) = -\sum_{\mathbf{s} \in \mathcal{H}}\text{Var}_\mu\big(\mathbb{P}(\mathbf{s}|\theta)\big).
    \label{eq:app_cov}
\end{equation}
This equation gives us a relation between the average size of covariance between two output probabilities, and the variance of a single output probability. No statistical equivalence between outcomes is assumed, the moments of $\mathbb{P}(\mathbf{s}|\theta)$ will depend on the collision pattern of $\mathbf{s}$, but only the average relation will be required.

Let $S$ be drawn uniformly over subsets of $\mathcal{H}$ of size $|S| = |\mathcal{H}|/K$. The uniform measure over $S$ and the Haar measure are independent, and we will consider finite sums with bounded summands, so expectations may be exchanged freely by Fubini's theorem. Each outcome $\mathbf{s}$ belongs to $S$ with probability $|S|/|\mathcal{H}|$, and each ordered pair $\mathbf{s} \neq \mathbf{t}$ with probability $|S|(|S|-1)/(|\mathcal{H}|(|\mathcal{H}|-1))$, so averaging Eq.~\eqref{eq:varcirc_decomp} over $S$ and applying Eq.~\eqref{eq:app_cov} gives
\begin{align}
\begin{split}
    \mathbb{E}_{S}\left[\text{Var}_{\text{circ}}(\hat{O}_S)\right] &= \frac{|S|}{|\mathcal{H}|}\sum_{\mathbf{s}}\text{Var}_\mu\big(\mathbb{P}(\mathbf{s}|\theta)\big) + \frac{|S|(|S|-1)}{|\mathcal{H}|(|\mathcal{H}|-1)}\sum_{\mathbf{s}\neq\mathbf{t}}\text{Cov}_\mu\big(\mathbb{P}(\mathbf{s}|\theta),\mathbb{P}(\mathbf{t}|\theta)\big) \\
    &= \frac{|S|}{|\mathcal{H}|}\left(1 - \frac{|S|-1}{|\mathcal{H}|-1}\right)\sum_{\mathbf{s}}\text{Var}_\mu\big(\mathbb{P}(\mathbf{s}|\theta)\big).
    \label{eq:app_after_cov}
\end{split}
\end{align}
The summed variance is $|\mathcal{H}|$ times the outcome-averaged circuit variance of a single-outcome projector;
\begin{equation}
    \sum_{\mathbf{s} \in \mathcal{H}}\text{Var}_\mu\big(\mathbb{P}(\mathbf{s}|\theta)\big) = |\mathcal{H}| \cdot \text{Var}_{\text{circ}}(\ketbra{\mathbf{s}}{\mathbf{s}}) = \frac{P_2(M,n) - 1}{|\mathcal{H}|}.
\end{equation}
Substituting into Eq.~\eqref{eq:app_after_cov} with $|S| = |\mathcal{H}|/K$ and $|\mathcal{H}| \gg 1$ gives the result listed in Eq.~\eqref{eq:course_grained_results_main}, namely
\begin{equation}
    \mathbb{E}_{S}\left[\text{Var}_{\text{circ}}(\hat{O}_S)\right] \approx \frac{1}{K}\left(1 - \frac{1}{K}\right)\frac{P_2(M,n) - 1}{|\mathcal{H}|}.
\end{equation}

\end{document}